\documentclass[11pt]{article}

\usepackage[dvipsnames]{xcolor}
\definecolor{Gred}{RGB}{219, 50, 54}
\definecolor{ToCgreen}{RGB}{0, 128, 0}

\usepackage[margin=1in]{geometry}
\usepackage[T1]{fontenc}

\usepackage[scale=0.97]{XCharter} 
\usepackage[libertine,bigdelims,vvarbb,scaled=1.05]{newtxmath} 

\usepackage{babel}
\usepackage[spacing=true,kerning=true,babel=true,tracking=true]{microtype}

\DeclareMathAlphabet{\pazocal}{OMS}{zplm}{m}{n} 
\renewcommand{\mathcal}[1]{\pazocal{#1}}

\usepackage{makecell}
\usepackage{dsfont}

\usepackage{multirow}
\usepackage{amsmath,amsthm}
\usepackage{bm}
\usepackage{bbm}
\usepackage{textgreek}
\usepackage{mathtools}
\usepackage[shortlabels]{enumitem}
\usepackage[numbers,comma,sort&compress]{natbib}
\usepackage{authblk}
\usepackage{graphicx}
\usepackage[font=small]{caption}
\usepackage[labelformat=simple]{subcaption}

\usepackage{float}
\usepackage[linesnumbered,ruled,vlined]{algorithm2e}
\SetKwInput{KwInput}{Input}
\SetKwInput{KwOutput}{Output}
\SetKwInOut{Promise}{Promise}
\SetKwInput{Goal}{Goal}
\SetKwProg{Fn}{function}{}{}
\SetKwFor{RepTimes}{repeat}{times}{}
\SetKwFunction{Ver}{Verify}
\SetKwFunction{Prep}{PrepareState}
\usepackage{algorithmic}

\usepackage{physics}
\usepackage{footnote}
\usepackage{xcolor}
\usepackage{mathrsfs}
\usepackage{bbm}
\usepackage{braket}
\usepackage[colorlinks]{hyperref}
\usepackage{cleveref}
\usepackage{tikz}
\usetikzlibrary{matrix, arrows.meta}
\hypersetup{
      colorlinks=true,
  citecolor=ToCgreen,
  linkcolor=Sepia,
  filecolor=Gred,
  urlcolor=Gred
  }
\usepackage{placeins}

\numberwithin{equation}{section}

\newtheorem{theorem}{Theorem}[section]
 
\newtheorem{definition}[theorem]{Definition}
\newtheorem{lemma}[theorem]{Lemma}

\newtheorem{corollary}[theorem]{Corollary}

\newtheorem{problem}{Problem}
\newtheorem{problemp}{Problem}

\renewcommand{\eqref}[1]{\hyperref[#1]{Eq.~(\ref{#1})}}
\newcommand{\probsref}[1]{Problems~\hyperref[#1]{\ref{#1}}}
\newcommand{\probref}[1]{\hyperref[#1]{Problem~\ref{#1}}}
\newcommand{\secref}[1]{\hyperref[#1]{Section~\ref{#1}}}
\newcommand{\lemmaref}[1]{Lemma~\hyperref[#1]{\ref{#1}}}

\newcommand{\algo}[1]{\hyperref[algo:#1]{Algorithm~\ref*{algo:#1}}}

\newcommand{\N}{\mathbb{N}}

\newcommand{\R}{\mathbb{R}}
\newcommand{\E}{\mathbb{E}}

\newcommand{\cM}{\mathcal{M}}
\newcommand{\cMh}{\mathcal{M}^{\rm h}}
\newcommand{\cN}{\mathcal{N}}
\newcommand{\cT}{\mathcal{T}}

\newcommand{\cP}{\mathcal{P}}
\newcommand{\cD}{\mathcal{D}}

\newcommand{\cS}{\mathcal{S}}

\newcommand{\tM}{\tilde{M}}
\usepackage{mathrsfs}

\newcommand{\bR}{{\mathbb R}}
\newcommand{\bE}{{\mathbb E}}
\newcommand{\bN}{{\mathbb N}}
\newcommand{\htheta}{{\hat \theta}}
\newcommand{\MoM}{{\rm MoM}}
\newcommand{\hrho}{{\hat \rho}}

\newcommand{\dtheta}{{\mathrm{d}\theta}}
\newcommand{\hato}{{\hat{o}}}
\newcommand{\DD}{{\mathscr{D}}}
\newcommand{\DDD}{{\mathscr{N}}}
\newcommand{\BB}{{\mathscr{B}}}

\usepackage{textgreek}

\DeclareMathOperator{\poly}{poly}
\newcommand{\polylog}{{\rm polylog}}
\DeclareMathOperator{\Var}{Var}
\DeclareMathOperator{\diag}{diag}

\renewcommand{\tr}{\mathrm{tr}}
\renewcommand{\Tr}{\mathrm{tr}}

\newcommand{\0}{\mathbf{0}}
\newcommand{\1}{\mathbf{1}}
\newcommand{\ob}{\mathrm{ob}}
\newcommand{\comp}{\mathrm{comp}}
\renewcommand{\emptyset}{\varnothing}
\def\Tr{\tr}\def\:{\hbox{\bf:}}

\newcommand{\tO}{\tilde{O}}

\newcommand{\tTheta}{\tilde{\Theta}}

\newcommand{\coleq}{\mathrel{\mathop:}\nobreak\mkern-1.2mu=}

\newcommand{\bb}{\begin{equation}\begin{aligned}\hspace{0pt}}
\newcommand{\bbb}{\begin{equation}\begin{aligned}}
\newcommand{\ee}{\end{aligned}\end{equation}}
\newcommand{\eee}{\end{aligned}\end{equation}}
\renewcommand{\epsilon}{\varepsilon}

\definecolor{gold}{rgb}{0.85, 0.65, 0.13}

\newcommand{\id}{{\mathbbm{1}}}

\allowdisplaybreaks

\begin{document}
%%%%%%%%%%%%%%%%%%%%%%%%%%%%%%%%%%%%%%%%%%%%%%%%%%%%%%%%%%%%%
%%%%%%%%%%%%%%%%%%%%%%%%%%%%%%%%%%%%%%%%%%%%%%%%%%%%%%%%%%%%%
%%%%%%%%%%%%%%%%%%%%%%%%%%%%%%%%%%%%%%%%%%%%%%%%%%%%%%%%%%%%%
%%%%%%%%%%%%%%%%%%%%%%%%%%%%%%%%%%%%%%%%%%%%%%%%%%%%%%%%%%%%%

\title{Instance-optimal high-precision shadow tomography with few-copy measurements: A metrological approach}
% \title{Tight bounds for high-precision shadow tomography with few-copy measurements}
\author{Senrui Chen
\thanks{IQIM, California Institute of Technology. Email: \href{mailto:csenrui@gmail.com}{csenrui@gmail.com}.}
\qquad\qquad
Weiyuan Gong
\thanks{SEAS, Harvard University. Email: \href{mailto:wgong@g.harvard.edu}{wgong@g.harvard.edu}. A part of this work was conducted while visiting IQIM.}
\qquad\qquad
Sisi Zhou
\thanks{Perimeter Institute. Email: \href{mailto:sisi.zhou26@gmail.com}{sisi.zhou26@gmail.com}.}
\,\thanks{Department of Physics and Astronomy, Department of Applied Mathematics and IQC, University of Waterloo.}\,
\thanks{Authors are listed in alphabetical order.}
}
\date{\today}
\maketitle

\pagenumbering{gobble}

\begin{abstract}

We study the sample complexity of shadow tomography in the high-precision regime under realistic measurement constraints. Given an unknown $d$-dimensional quantum state $\rho$ and a known set of observables $\{O_i\}_{i=1}^m$, the goal is to estimate expectation values $\{\mathrm{tr}(O_i\rho)\}_{i=1}^m$ to accuracy $\epsilon$ in $L_p$-norm, using possibly adaptive measurements that act on $O(\mathrm{polylog}(d))$ number of copies of $\rho$ at a time. We focus on the regime where $\epsilon$ is below an instance-dependent threshold. 

Our main contribution is an instance-optimal characterization of the sample complexity as $\tilde{\Theta}(\Gamma_p/\epsilon^2)$, where $\Gamma_p$ is a function of $\{O_i\}_{i=1}^m$ defined via an optimization formula involving the inverse Fisher information matrix. Previously, tight bounds were known only in special cases, e.g. Pauli shadow tomography with $L_\infty$-norm error. Concretely, we first analyze a simpler oblivious variant where the goal is to estimate an observable of the form $\sum_{i=1}^m \alpha_i O_i$ with $\|\alpha\|_q = 1$ (where $q$ is dual to $p$) revealed after the measurement. For single-copy measurements, we obtain a sample complexity of $\Theta(\Gamma^{\mathrm{ob}}_p/\epsilon^2)$. We then show $\tilde{\Theta}(\Gamma_p/\epsilon^2)$ is necessary and sufficient for the original problem, with the lower bound applying to unbiased, kurtosis-bounded estimators when $p \geq 1$. Our upper bounds rely on a two-step algorithm combining coarse tomography with local estimation. Notably, $\Gamma^{\mathrm{ob}}_\infty = \Gamma_\infty$. In both cases, allowing $c$-copy measurements improves the sample complexity by at most $\Omega(1/c)$. 

Our results establish a quantitative correspondence between quantum learning and metrology, unifying asymptotic metrological limits with finite-sample learning guarantees.

\end{abstract}

% \clearpage
% \newpage
\tableofcontents
\clearpage
\newpage

%%%%%%%%%%%%%%%%%%%%%%%%%%%%%%%%%%%%%%%%%%%%%%%%%%%%%%%%%%%%%%%%%%%%%%%%%%%%%%%%%%%%%%%%%%%%%%%%%%%%%%%%%%%%%%%%%%%%%%%%%%%%%%%%%%%%%%%%%%%%%%%%%%%%%%%%%%%%%%%%%%%%%%%%%%%%%%%%%%%%%%%%%%
\pagenumbering{arabic}

\section{Introduction}

A fundamental task in quantum information is to characterize an unknown or partly unknown state, with applications in \textit{e.g.} quantum sensing~\cite{degen2017quantum,pirandola2018advances}, quantum algorithms~\cite{Dalzell2025QuantumAlgorithms}, as well as benchmarking noisy quantum devices~\cite{harper2020efficient,hashim2024practical}. There are two fields of research that have been devoted to conquering this task: One is Quantum Learning; the other is Quantum Metrology. Despite the same over-arching goal of understanding how efficiently one can extract information about an unknown quantum system, the languages, techniques, and communities are surprisingly different. 

Quantum metrology, or quantum estimation theory, is a topic with a long history and is widely applied in experiments~\cite{giovannetti2011advances,degen2017quantum,pezze2018quantum,pirandola2018advances}. As a typical setting, a quantum state is parameterized by one or many unknown parameters. The goal is to understand how precisely the parameters can be determined as the number of available state copies goes to infinity. Modern theoretical quantum metrology research mostly relies on (quantum) Fisher information and the (quantum) Cram\'{e}r--Rao bound, which are powerful tools in asymptotic statistics~\cite{van2000asymptotic}. 

Quantum learning is a younger field, strongly influenced by the computer science and machine learning community~\cite{arunachalam2017guest}. One typical question, known as shadow tomography~\cite{aaronson2018shadow}, asks how many copies of an $d$-dimensional unknown state are needed to estimate a set of observables to certain precision $\varepsilon$. Crucially, instead of investigating the $\varepsilon\rightarrow0$ limit, one is interested in the scaling of complexity with finite $\varepsilon$ and $d$. This falls into the regime of non-asymptotic statistics, and people are using very different methods than those used in quantum metrology research.

Nevertheless, many have wondered about the following question:
\begin{center} 
    {\em Can we establish a rigorous correspondence between quantum metrology and quantum learning?}
\end{center}

We provide an affirmative answer by using quantum metrological approaches to solve an important open problem in quantum learning theory: namely, the instance-optimal sample complexity for high-precision shadow tomography with few-copy measurements. At a high level, the problem is to estimate many given observables of an unknown quantum states to $\varepsilon$ precision in $L_p$ norm, with the restriction that each round of measurement acts upon one or a few copies. Here, ``instance-optimal'' means the bounds depend on the set of observables, and ``high-precision'' means we restrict $\varepsilon$ to be smaller than a concrete threshold that may depend on the dimension of the Hilbert space and the specific set of observables. For this task, we essentially show that an intuitive bound one would expect using Fisher information gives a tight characterization of the sample complexity. 
Our results show that quantum metrology and learning are not only closely related in concept, but has exact mathematical correspondence in an appropriate regime.

\section{Results}

\subsection{Six problems: learning, estimation and distinguishing}

Consider a $d$-dimensional Hilbert space.
Let $\{O_i\}_{i=1}^m$ be a set of linearly-independent (thus $m\le d^2-1$) traceless Hermitian operators that we want to learn. Define the dual operator basis $\{Q_a\}_{a\in A}\cup \{T_b\}_{b\in B}$ which form a complete basis of traceless Hermitian observables satisfying
    \begin{equation}
    \label{eq:dual}
        \begin{aligned}
            \Tr(O_iQ_a)&=d\delta_{ia},\quad\forall i, a\in A,\\     
            \Tr(O_iT_b)&=0,\quad\forall i\in A ,b\in B.
        \end{aligned}
    \end{equation}
    We use $A = \{1,\ldots,m\}$ and $B = \{m+1,\ldots,d^2-1\}$ to denote the set of indices for $Q_a$ and $T_b$. 
    Now, fix any (known) reference quantum state $\rho_0$, one can always parameterize any quantum state $\rho$ by
    \begin{equation}
    \label{eq:parameterization}
        \rho_{\theta,\varphi}\coleq \rho_0+\frac{1}{d}\sum_{a \in A }\theta_a Q_a +\frac{1}{d}\sum_{b \in B}\varphi_b T_b.
    \end{equation}  
    With these definitions, we are ready to introduce the following problems about quantum state learning, estimation, and distinguishing. 
    $p \in [1,\infty]$ is a variable index in the following problems.

\begin{problem}[Learning of observables with $p$-norm error]  \label{prob:p-learning} 
Given $N$ i.i.d. copies of a quantum state $\rho$, find estimators $\{\hato_i\}_{i=1}^m$ such that $(\sum_{i=1}^m \abs{\trace(O_i\rho) - \hato_i}^p)^{1/p} < \epsilon$ with high probability. \sloppy
\end{problem}

\begin{problem}[Estimation of parameters with $p$-norm error] \label{prob:p-estimation}
Given $N$ i.i.d. copies of a parameterized quantum state $\rho_{\theta,\varphi}$ for some known (but arbitrary) choice of $\rho_0,\{Q_a\}_{a\in A}$ and $\{T_b\}_{b\in B}$, find estimators $\{\htheta_i\}_{i=1}^{m}$ such that $(\sum_{i=1}^m | \htheta_i - \theta_i |^p)^{1/p} < \epsilon$ with high probability.  
\end{problem}

\begin{problem}[Distinguishing between one and many states]\label{prob:distinguish}
  Given some known (but arbitrary) choice of $\rho_0,\{Q_a\}_{a\in A}$ and $\{T_b\}_{b\in B}$, and $N$ identical copies of a quantum state that is either $\rho_0$ or $\rho_{\theta,\varphi}$ with equal probability, where $\theta \in \bR^m$ and $\varphi \in \bR^{d^2-m-1}$ can be arbitrary unknown vectors such that $\rho_{\theta,\varphi}$ is well-defined and $\norm{\theta}_p = 3\epsilon$, distinguish the two cases with high probability.
\end{problem}

\Cref{prob:p-learning} is a standard quantum state learning scenario, where $p = \infty$ is traditionally known as shadow tomography where expectation values of observables are learned up to a constant additive error. In our work, we extend the discussion to a more general form of shadow tomography where the target precision is $p$-norm error (for any $p \geq 1$) and obtain bounds on the sample complexity in the high-precision regime where $\epsilon$ is sufficiently small. We assume the choices of observables $\{O_i\}_{i=1}^m$ are known to the experimentalists prior to the experiments, which allows optimization of the learning algorithms based on the knowledge of target observables. As a result, the bounds we obtain are \emph{instance-optimal}, as functions of observables $\{O_i\}_{i=1}^m$, which capture different levels of difficulties when learning different types of observables. 

To put the quantum state learning problem in the context of quantum state estimation where techniques from quantum metrology are available, we first observe that \Cref{prob:p-learning} is equivalent to the state estimation problem \Cref{prob:p-estimation}. 
\begin{proof}[Proof (Equivalence between \Cref{prob:p-learning} and \Cref{prob:p-estimation}).]
We note that any state $\rho$ can be parametrized as \eqref{eq:parameterization} because $\rho_0$ and $\{Q_a\}_{a\in A}\cup \{T_b\}_{b\in B}$ form a complete basis of $d$-dimensional Hermitian operators. Furthermore, $\theta_a = \trace(\rho O_a) - \trace(\rho_0 O_a)$. Therefore, estimating $\theta_a$ up to $p$-norm error is equivalent to estimating $\trace(\rho O_a)$ up to $p$-norm error because they differ only by a known, constant vector. 
\end{proof}

Hypothesis testing or distinguishing tasks are in general easier than learning (or estimation) tasks, and they conveniently provide lower bounds on the resource required for learning tasks from an information-theoretic point of view. The intuition is successful learning of quantum state properties can be used to distinguish different types of states. In our case, we consider a many-versus-one distinguishing task where a reference state $\rho_0$ is to be distinguished from $\rho_{\theta,\varphi}$ which is close to $\rho_0$ under our $p$-norm metric. 

\begin{proof}[Proof (\Cref{prob:distinguish} is no harder than \Cref{prob:p-estimation}).]\sloppy
Consider the unknown state $\rho$ given in \Cref{prob:distinguish}. Let $\hat{\theta}$ be the estimator constructed from \Cref{prob:p-estimation}. If $\rho = \rho_0$, $\|\hat\theta\|_p < \epsilon$ with high probability; and if $\rho = \rho_{\theta,\varphi}$, $\|\hat\theta - \theta\|_p < \epsilon$ and $\|\hat\theta\|_p \geq \|\theta\|_p - \|\hat\theta - \theta\|_p > 2\epsilon$ with high probability. Thus, by determining whether $\|\hat\theta\|_p$ is closer to $0$ or $3\epsilon$, \Cref{prob:distinguish} can be solved with high probability using $\hat{\theta}$. 
\end{proof}

As we will see later, when $p = \infty$ and $\epsilon$ is sufficiently small, \Cref{prob:distinguish} is as hard as 
\Cref{prob:p-estimation} up to a logarithmic overhead. It implies the distinguishing capability can almost capture the learning/estimation capability in certain regimes which in general does not hold.

\begin{figure}[tb]
    \centering

\begin{equation*}
\begin{tikzpicture}[baseline=(m.center),
  doublearrow/.style={
  double,
  double distance=1pt,
  -{Straight Barb[length=3pt]}
},
doublearrowrev/.style={
  double,
  double distance=1pt,
  {Straight Barb[length=3pt]}-
},
doublearrowboth/.style={        % <--- INSERTED HERE
    double,
    double distance=1pt,
    {Straight Barb[length=3pt]}-{Straight Barb[length=3pt]}
}
]

\matrix (m) [
  matrix of nodes,
  column sep=3em,
  row sep=1.5em,
  nodes={align=center}
]{
  \begin{tabular}{@{}c@{}}
Learning\\
(\Cref{prob:p-learning})\end{tabular} &   \begin{tabular}{@{}c@{}}
Estimation\\
(\Cref{prob:p-estimation})\end{tabular} &   \begin{tabular}{@{}c@{}}
Distinguishing\\
(\Cref{prob:distinguish})\end{tabular} \\
       &     &     \\
  \begin{tabular}{@{}c@{}}
Oblivious Learning\\
(\Cref{prob:oblivious-p-learning})\end{tabular} &   \begin{tabular}{@{}c@{}}
Oblivious Estimation\\
(\Cref{prob:oblivious-p-estimation})\end{tabular} & \begin{tabular}{@{}c@{}}
Oblivious Distinguishing\\
(\Cref{prob:oblivious-distinguish})\end{tabular} \\
};

% horizontal arrows:
\draw[doublearrowboth] (m-1-1) -- (m-1-2);  % double-ended (⇔)
\draw[doublearrow]                 (m-1-2) -- (m-1-3);  % right-pointing (⇒)

\draw[doublearrowboth] (m-3-1) -- (m-3-2);  % double-ended (⇔)
\draw[doublearrow]                 (m-3-2) -- (m-3-3);  % right-pointing (⇒)

% vertical arrows:
\draw[doublearrow] (m-1-1) -- (m-3-1);
\draw[doublearrow] (m-1-2) -- (m-3-2);
\draw[doublearrow] (m-1-3) -- (m-3-3);

\end{tikzpicture}
\end{equation*}

    \caption{Relationship between six problems. $\Leftrightarrow$ means two problems are equivalent, and $\Leftarrow$ means one problem is no harder than (reduces to) the other. Note that within each problem, increasing $p$ (or decreasing $q$) will not increase the hardness of the problem. }
    \label{fig:six-prob}
\end{figure}
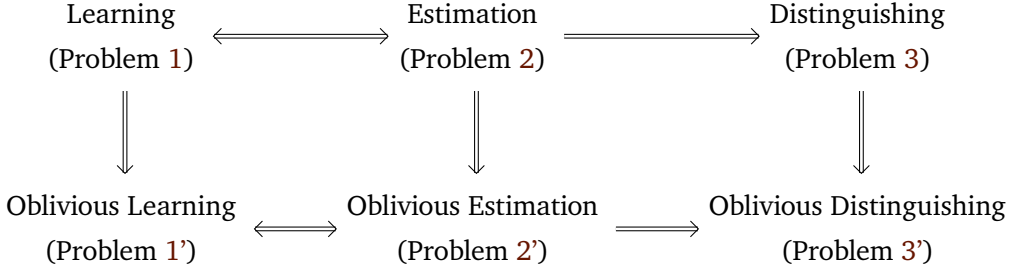

Next, we consider a different but also practically relevant setting which we call \emph{oblivious} learning or oblivious estimation. Here, instead of aiming at precisely estimating all expectation values of observables, we need to only estimate a single observable chosen arbitrarily from a set which we call the $L^q$-ellipsoid of $\{O_i\}_{i=1}^m$, 
\begin{equation}
    \left\{O_\alpha := \sum_{i=1}^m \alpha_i O_i, \quad \forall \alpha \in \bR^m, \norm{\alpha}_q = 1 \right\},
\end{equation}
where $q \in [1,\infty]$ satisfying $1/p+1/q=1$ is the dual index to $p$. The choice of $O_\alpha$ (i.e. the value of $\alpha$) will only be revealed after all quantum measurements are performed, forbidding any informed optimization of the measurement prior to data processing---which explains the name ``oblivious''. Below we have the three oblivious versions of the previously defined learning, estimation, and distinguishing problems.

\begin{problemp}[Oblivious learning of an observable from $L_q$-ellipsoid]  \label{prob:oblivious-p-learning} Given $N$ i.i.d. copies of quantum state $\rho$ and an arbitrary $\alpha \in \bR^m$ satisfying $\norm{\alpha}_q \leq 1$, where the values of $\alpha$ is revealed after all quantum measurements are performed, find an estimators $\hato_\alpha$ such that $|\trace(\sum_{i=1}^m \alpha_iO_i\rho) - \hato_\alpha | < \epsilon$ with high probability. \sloppy
\end{problemp}

\begin{problemp}[Oblivious estimation of a linear function of parameters]  Given $N$ i.i.d. copies of the parametrized quantum state $\rho_{\theta,\varphi}$ described by some known (but arbitrary) choice of $\rho_0,\{Q_a\}_{a\in A}$ and $\{T_b\}_{b\in B}$ and an arbitrary $\alpha \in \bR^m$ satisfying $\norm{\alpha}_q \leq 1$, with the values of $\alpha$ revealed after all quantum measurements are performed, find an estimator $\htheta_\alpha$ such that $|\theta_\alpha  - \htheta_\alpha| < \epsilon$ with high probability where $\theta_\alpha : = \alpha\cdot\theta$.  
\label{prob:oblivious-p-estimation}
\end{problemp}

\begin{problemp}[Oblivious distinguishing between two states]\label{prob:oblivious-distinguish}
   Given some known (but arbitrary) choice of $\rho_0,\{Q_a\}_{a\in A}$ and $\{T_b\}_{b\in B}$ and $N$ i.i.d. copies of a quantum state that is either $\rho_0$ or $\rho_{\theta,\varphi}$ with equal probability, where $\theta \in \bR^m$ and $\varphi \in \bR^{d^2-m-1}$ can be arbitrary unknown vectors such that $\rho_{\theta,\varphi}$ is well-defined and $\norm{\theta}_p = 3\epsilon$, with the values of $(\theta,\varphi)$ revealed after all quantum measurements are performed, distinguish the two cases with high probability. 
\end{problemp}

Analogous to the previous proof, we observe that  \Cref{prob:oblivious-p-learning} is equivalent to \Cref{prob:oblivious-p-estimation}, and is no easier than  \Cref{prob:oblivious-distinguish}.

\begin{proof}[Proof (\Cref{prob:oblivious-distinguish} is no harder than \Cref{prob:oblivious-p-estimation}).]\sloppy
Let $(\theta,\varphi)$ be the revealed parameters and $\rho$ be the unknown state in \Cref{prob:oblivious-distinguish}. Pick $\alpha \in \bR^m$ such that $\alpha\cdot\theta = 3\epsilon$ and $\norm{\alpha}_q = 1$. Let $\hat{\theta}_\alpha$ be the estimator constructed from \Cref{prob:oblivious-p-estimation}. If $\rho = \rho_0$, $|\hat\theta_\alpha| < \epsilon$ with high probability; and if $\rho = \rho_{\theta,\varphi}$, $|\alpha\cdot \theta - \hat\theta_\alpha| < \epsilon$ and $|\hat\theta_\alpha| \geq |\alpha\cdot \theta| - |\hat\theta_\alpha - \alpha\cdot \theta| > 2\epsilon$ with high probability. Thus, by determining whether $|\hat\theta_\alpha|$ is closer to $0$ or $3\epsilon$, \Cref{prob:distinguish} can be solved with high probability using $\hat{\theta}_\alpha$. 
\end{proof}

Moreover, the oblivious versions of tasks described above are no harder than the original learning, estimation and distinguishing tasks. While \Cref{prob:oblivious-distinguish} is trivially no harder than \Cref{prob:distinguish}, the relationship between \Cref{prob:oblivious-p-estimation} and \Cref{prob:p-estimation} can be seen from the following. 

\begin{proof}[Proof (\Cref{prob:oblivious-p-estimation} is no harder than \Cref{prob:p-estimation}).]\sloppy
Consider \Cref{prob:p-estimation} and let $\htheta$ be an estimator satisfying $\|\htheta-\theta\|_p < \epsilon$ with high probability. Then $\htheta_\alpha = \alpha\cdot\htheta$ solves \Cref{prob:oblivious-p-estimation} because $| \alpha\cdot\theta - \alpha\cdot\htheta | \leq \norm{\alpha}_q \|\htheta-\theta\|_p < \epsilon$ using H\"{o}lder's inequality. 
\end{proof}

Oblivious estimation fits well into the ``measure first, ask questions later'' regime where the observable of interest is determined after the experiments. In our case, $O_\alpha$ can be chosen freely as any linear combination of $\{O_i\}_{i=1}^m$. As a result, the oblivious estimation task is also challenging to solve. As we will see later when $p = \infty$, the oblivious estimation and the original shadow estimation tasks have the same sample complexity up to logarithmic overhead. 

\subsection{Oblivious estimation from \texorpdfstring{$L^q$}{Lq}-ellipsoid}

In this section, we summarize the results we obtain on oblivious estimation and distinguishing tasks (see also \Cref{tab:summary}). In particular, we will show tight upper and lower bounds on the sample complexity in the high-precision regime, i.e. when the target precision $\epsilon$ is sufficiently small. 

\begin{theorem}[Oblivious estimation with one-copy measurements, informal] \label{thm:main_ob_full}
    Given any $p \in [1,\infty]$, to solve \Cref{prob:oblivious-p-learning}, \Cref{prob:oblivious-p-estimation} or \Cref{prob:oblivious-distinguish} with (adaptive) single-copy measurement protocols, there exists a threshold on the target precision $\epsilon$ below which 
    \begin{equation}
        N  
        = \Theta \left( \frac{\Gamma^\ob_{p}(\{O_i\}_{i=1}^m)}{\epsilon^2} \right)
    \label{eq:sample-ob-full}
    \end{equation}
    copies of quantum states $\rho$ are necessary and sufficient. Here $\Gamma^\ob_{p}(\{O_i\}_{i=1}^m)$ is a positive function of the observables defined by 
    \begin{equation}
        \Gamma^\ob_{p}(\{O_i\}_{i=1}^m) := \sup_{\rho_0\in\cS^\circ} \inf_{M\in\cM} \max_{\substack{\|\alpha\|_q\le1 \\ \alpha \in \bR^m}} \alpha^\top(I(\rho_0,M)^{-1})_{AA}\alpha,
    \end{equation}
    where $I(\rho_0,M)$ is the Fisher information matrix (FIM) $I(\rho_{\theta,\varphi},M)$ of measuring $\rho_{\theta,\varphi}$ with measurement $M$ at $(\theta,\varphi) = (0,0)$, $\cM$ is the set of all single-copy measurements, $\cS^\circ$ is the set of all full-rank density matrices and $AA$ denotes the upper-left matrix block of $I(\rho_0,M)$ when indices are restricted to $a \in A$.  
\end{theorem}

\noindent \Cref{thm:main_ob_full} provides tight bounds on the sample complexity required to obliviously estimate observables from the $L^q$-ellipsoid of $\{O_i\}_{i=1}^m$. Here the function $\Gamma^\ob_{p}(\{O_i\}_{i=1}^m)$ fully characterizes the sample complexity required to perform oblivious estimation or distinguishing tasks in the high-precision regime. To understand the operational meaning of $\Gamma^\ob_{p}(\{O_i\}_{i=1}^m)$, we first note that the FIM \begin{equation}
    I(\rho_0,M)_{cc'} = \frac{1}{d^2}\sum_{x} \frac{\trace( R_c M_x)\trace(R_{c'} M_x)}{\trace(\rho_0 M_x)},\quad R_c = \begin{cases}
        Q_a & c=a\in A, \\
        T_b & c=b\in B. 
    \end{cases}
\end{equation} 
characterizes the amount of information a parametrized quantum state contains about unknown parameters that can be extracted from quantum measurement $M = \{M_x\}_x$. Its inverse is considered an asymptotically attainable lower bound on the variance of unbiased estimators by the Cram\'{e}r--Rao bound~\cite{kay1993fundamentals,lehmann2006theory}. Here the scalar $\alpha^\top (I(\rho_0,M)^{-1})_{AA} \alpha$ exactly corresponds to the variance of the $\theta_\alpha$ estimator for different $\alpha$. \sloppy The maximizations over $\alpha$ and $\rho_0$ are due to our requirement that the estimation is oblivious for all $\alpha$ and the algorithm applies to arbitrary reference state $\rho_0$. The minimization over measurements $M$ guarantees the optimal single-copy measurement is chosen for each $\rho_0$ adaptively. 
Finally, it is a mathematical property that the function $(I(\rho_0,M)^{-1})_{AA}$ is independent of different choices of dual bases $\{Q_a\}_{a\in A},\{T_b\}_{b\in A}$ (satisfying \eqref{eq:dual}), and is solely a function of observables $\{O_i\}_{i=1}^m$. In fact, changing different dual bases will not affect the sample complexity, but might affect the threshold below which the high-precision regime is reached. 
Note that throughout this paper, in the expression $\alpha^\top (I(\rho_0,M))^{-1}_{AA} \alpha$, $(I(\rho_0,M))^{-1}_{AA}$ should be interpreted as the pseudoinverse of $I(\rho_0,M)$ on its support when $\alpha$ is in the support; and $\infty$ when $\alpha$ is outside the support.

\begin{table}[tb]
\centering
\small
\begin{tabular}{|l|c|c|}
\hline
\multicolumn{1}{|c|}{\shortstack[c]{Tasks\\~~}} & \shortstack[c]{\rule{0pt}{2.6ex}Oblivious estimation from $L^q$-ellipsoid\\(\probsref{prob:oblivious-p-learning} and  \ref{prob:oblivious-p-estimation}, $p \in [1,\infty]$)} & 
\shortstack[c]{\rule{0pt}{2.6ex}Shadow estimation with $p$-norm error\\(\probsref{prob:p-learning} and \ref{prob:p-estimation}, $p \in [1,\infty]$)} \\
\hline
\multirow{2}{*}{\parbox{3.8cm}{Lower bounds for many-versus-one distinguishing (applying to \probsref{prob:oblivious-distinguish})}} & \shortstack[c]{\rule{0pt}{2.6ex}$\Omega\big(\Gamma^\ob_{p}(\{O_i\}_{i=1}^m)/\epsilon^2\big)$ \\($\epsilon \leq \eta^\ob$, single-copy)}  &      \multicolumn{1}{c|}{\shortstack[c]{Same as the left column\\~~}}     \\
\cline{2-3}
                           & \shortstack[c]{\rule{0pt}{2.6ex} $\Omega\big(\Gamma^\ob_{p}(\{O_i\}_{i=1}^m)/(c\epsilon^2)\big)$\\ ($\epsilon \leq \min\{\eta^\ob,\eta^{\ob}_{c}\}$, $c$-copy) } &   \multicolumn{1}{c|}{\shortstack[c]{Same as the left column\\~~}}         \\
\hline
\shortstack[l]{\rule{0pt}{2.6ex}Lower bounds for unbiased,\\ kurtosis-bounded  estimation} & \shortstack[c]{ $\Omega\big(\Gamma^\ob_{p}(\{O_i\}_{i=1}^m)/(c\epsilon^2)\big)$\\ ($c$-copy)} &             \shortstack[c]{ $\Omega\big(\Gamma_{p}(\{O_i\}_{i=1}^m)/(c\epsilon^2)\big)$\\ ($c$-copy)}  \\
\hline
\shortstack[l]{Upper bounds using \\ the two-step method} & \shortstack[c]{\rule{0pt}{3ex} $O\big(\Gamma^\ob_{p}(\{O_i\}_{i=1}^m)/\epsilon^2\big)$\\ ($\epsilon \leq \overline{\eta}^\ob$, single-copy)} &             \shortstack[c]{$O\big(\log(2m)\Gamma_{p}(\{O_i\}_{i=1}^m)/\epsilon^2\big)$\\ ($\epsilon \leq \overline{\eta}$, single-copy)}  \\
\hline
\end{tabular}
\caption{\label{tab:summary} Summary of results. We use two different methods to derive lower bounds on the sample complexity using $c$-copy measurements. The first method is detailed in \Cref{sec:distinguish} where lower bounds and corresponding thresholds are derived for \Cref{prob:oblivious-distinguish} using the learning tree method.  The second method is detailed in \Cref{sec:unbiased} where lower bounds are derived directly for \probsref{prob:p-estimation} and \ref{prob:oblivious-p-estimation} using the Cram\'{e}r--Rao method assuming unbiased and kurtosis-bounded estimators. Explicit algorithms are provided in \Cref{sec:upper} using the two-step method where state tomography is first used to find a nearby state and locally optimal estimation is then used to achieve the target precision. }
\end{table}

To understand the threshold behavior, we explain the derivation of the lower and upper bounds separately. Our lower bound is derived through the distinguishing task (\Cref{prob:oblivious-distinguish}) where the goal is to distinguish all well-defined $\rho_{\theta,\varphi}$ with $\norm{\theta}_p = 3\epsilon$ from $\rho_{0}$. For any $\epsilon > 0$ without threshold, we show a lower bound equal to 
\begin{equation}
\label{eq:lower-dist-no-thres}
    \Omega\Bigg(\sup_{\rho_0\in\cS^\circ} \inf_{M \in \cM}  \max_{\substack{(\theta,\varphi) \text{ s.t. }\rho_{\theta,\varphi} \succeq 0, \\ \norm{\theta}_p = 3\epsilon}} \Big((\theta,\varphi)^\top I(\rho_0,M) (\theta,\varphi) \Big)^{-1}\Bigg).
\end{equation}
Here $(\theta,\varphi)^\top I(\rho_0,M) (\theta,\varphi)$ can be interpreted as the Fisher information corresponding to a specific instance of state $\rho_{\theta,\varphi}$ satisfying the well-definedness constraint $\rho_{\theta,\varphi} \succeq 0$ and $\norm{\theta}_{p} = 3\epsilon$. In order to distinguish successfully all instances from $\rho_0$, we maximize the function over all possible $(\theta,\varphi)$. In fact, we show the duality between $p$- and $q$- norms guarantees 
\begin{equation}
\label{eq:duality}
    \max_{\norm{\theta}_p = 3\epsilon} \big((\theta,\varphi)^\top I(\rho_0,M) (\theta,\varphi)\big)^{-1} = \max_{\norm{\alpha}_q = 1/3\epsilon}
    \alpha^\top (I(\rho_0,M)^{-1})_{AA} \alpha. 
\end{equation}
That means in order to prove \eqref{eq:sample-ob-full} as a lower bound for the distinguishing task, it is sufficient to remove the constraint $\rho_{\theta,\varphi} \succeq 0$ from \eqref{eq:lower-dist-no-thres}, which is possible by introducing a threshold. 
One trick we apply here that allows us to consider only states in $\cS_{1/2}:= \{\rho|\rho = \sigma/2 + \id/(2d), \text{ for some density matrix }\sigma\}$ \sloppy is to observe that $I(\frac{1}{2}\rho_0+\frac{\id}{2d},M) \preceq 2 I(\rho_0,M)$, which means mixing $\rho_0$ evenly with a maximally mixed state $\frac{\id}{d}$ does not change our lower bound up to constant. This allows us to derive an explicit formula of a valid threshold $\eta^\ob$ as a function of $\{O_i\}_{i=1}^m$ below which \eqref{eq:sample-ob-full} is a lower bound. 

Conceptually, the requirement of a threshold in the lower bound implies that the distinguishing task may be fundamentally easier to solve than the estimation task when $\epsilon$ is too large, where the ensemble of states $\{\rho_{\theta,\varphi}\}$ satisfying $\norm{\theta}_p = 3\epsilon$ may no longer be a good hypothesis to characterize the difficulty in parameter estimation because $\rho_{\theta,\varphi}$ is not well-defined for too many $(\theta,\varphi)$. For the special case of $\infty$-norm Pauli estimation, however, the distinguishing task provides a tight lower bound for all $1/3 \geq \epsilon > 0$~\cite{chen2024optimal} (when using single-copy measurements)---this is because restricting $(\theta,\varphi)$ to the set of $(e_i,0)$ where $e_i \in \bR^m$ is a vector that is $1$ in the $i$-th entry and $0$ in the others does not change the value of \eqref{eq:lower-dist-no-thres}. However, it is unknown whether the same property holds for general $p \in [1,\infty]$ and general observables.

The upper bound is derived using a two-step method where state tomography is first used to find a nearby state and locally optimal estimation is then
used to achieve the target precision. The two-step method is traditionally used in quantum metrology to show the attainability of the Cram\'{e}r--Rao bound asymptotically, i.e. when taking the limit $\epsilon \rightarrow 0$ and allowing infinitely many samples. We instead demonstrate the usefulness of the two-step method in a finite-sample regime, that was rarely explored previously, bridging a gap between quantum metrology and quantum learning. One interesting result we manage to show is the locally optimal estimator from quantum metrology is globally optimal (up to a constant factor) and unbiased within a finite-size region containing $\rho_0$. The bound in \eqref{eq:sample-ob-full} is tight when the sample complexity required in the pre-estimation stage for finding the finite-size region is negligible. To bound it, we again restrict our discussion to states in $\cS_{1/2}$ using the state mixture trick, and show state tomography that determines a nearby state within $O(1/d)$ in $\infty$-norm distance away from $\rho$ is sufficient for our purpose. When the target precision $\epsilon$ is too large, the pre-estimation stage can be too costly, making our algorithm suboptimal. 

In practice, the assumption of adaptive single-copy measurements can be too restrictive sometimes because entangled measurements across a few copies of states are also feasible on certain experimental platforms~\cite{huang2022quantum} and they can sometimes bring substantial improvement in sample complexity. For example, for $\infty$-norm Pauli estimation within the low-precision regime (e.g., $\varepsilon=0.1$), even two-copy measurements can provide exponential sample complexity reduction compared to single-copy protocols~\cite{chen2024optimal,huang2021information}. However, for $\infty$-norm Pauli estimation within the high-precision regime (e.g. when $\epsilon = O(1/d^{1/2})$), multi-copy measurements provide no substantial advantages over single-copy measurements unless measurements across an exponentially large number of copies are available. Here we observe a similar phenomenon for general oblivious estimation. 

\begin{theorem}[Oblivious estimation with $c$-copy measurements, informal]\label{thm:main_ob_c}
    Given any $p \in [1,\infty]$, to solve \Cref{prob:oblivious-p-learning}, \Cref{prob:oblivious-p-estimation} or \Cref{prob:oblivious-distinguish} with (adaptive) $c$-copy measurement protocols, there exists a threshold on the target precision $\epsilon$ below which 
    \begin{equation}
        N 
        = \Omega \left( \frac{\Gamma^\ob_{p}(\{O_i\}_{i=1}^m)}{c \epsilon^2} \right)
    \label{eq:sample-ob-c}
    \end{equation}
    copies of quantum states $\rho$ are necessary.  
\end{theorem}

\noindent \Cref{thm:main_ob_c} implies allowing entangled measurements across $c$ copies of states can at most bring a $O(1/c)$ advantage in the sample complexity in the high-precision regime. When $c$ is not too large, e.g. $c = O(\polylog(d))$, the reduction in sample complexity will also be at most polynomial. We prove the result by first bounding the lowest-order part of the sample complexity for small $\epsilon$, and then find a threshold on $\epsilon$ below which the higher-order terms are negligible. 

Finally, we note that an alternative method to derive lower bounds for \Cref{prob:oblivious-p-estimation} (and \Cref{prob:p-estimation} in the later section) is through the Cram\'{e}r--Rao bound. 
It provides a lower bound on the variance of any unbiased estimator given by the inverse of the FIM. 
Assuming the unbiased estimator has a bounded kurtosis, a condition that ensures it is sufficiently concentrated at the true value, achieving an $\varepsilon$ $p$-norm error or shadow estimation error  with high probability then implies an order $\varepsilon^2$ variance of estimation, and then we can use the CR bound.

\subsection{Shadow estimation with \texorpdfstring{$p$}{p}-norm error}

Above we provide tight bounds for the oblivious estimation and distinguishing tasks. They can be fundamentally easier than the shadow estimation task. Luckily, our two-step algorithm applies also to the general shadow estimation tasks, and the Cram\'{e}r--Rao method to derive lower bounds also (partly) applies here. As a result, we have the following theorems. 

\begin{theorem}[Shadow estimation with one-copy measurements, informal]\label{thm:main_full}
    Given any $p \in [1,\infty]$, there exists an estimation algorithm using single-copy measurements that solves \Cref{prob:p-learning} (and \Cref{prob:p-estimation}) using 
    \begin{equation}
    \label{eq:sample-full}
        N 
        = \tO \left( \frac{\Gamma_{p}(\{O_i\}_{i=1}^m)}{\epsilon^2} \right)
    \end{equation}
    copies of quantum state $\rho$ when the target precision $\epsilon$ is below a threshold. Here $\Gamma_{p}(\{O_i\}_{i=1}^m)$ is a positive function of the observables defined by 
    \begin{equation}
        \Gamma_{p}(\{O_i\}_{i=1}^m) := \sup_{\rho_0 \in \cS^\circ} \inf_{M\in\cM}  {\norm{\diag\big( (I(\rho_0,M)^{-1})_{AA}\big)^{1/2}}_p^2},
    \end{equation}
    where $\diag(\cdot)$ represents a diagonal matrix whose diagonal entries are those of $(\cdot)$ and $\norm{\cdot}_p$ represents the Schatten $p$-norm of matrices. 
    $N = \Omega\big(\Gamma_{p}(\{O_i\}_{i=1}^m)/\epsilon^2\big)$ is also necessary for any $\epsilon > 0$ when assuming unbiased and kurtosis-bounded estimation (see \Cref{sec:unbiased} for the precise conditions). 
\end{theorem}

Note that here for singular $I(\rho_0,M)$, we set $I(\rho_0,M)^{-1} = \infty$. 

\begin{theorem}[Shadow estimation with $c$-copy measurements for unbiased kurtosis-bounded estimators, informal]\label{thm:main_full_c}
    Given any $p \in [1,\infty]$, to solve \Cref{prob:p-learning} (and  \Cref{prob:p-estimation}) with (adaptive) $c$-copy measurement protocols, 
    \begin{equation}
        N 
        = \Omega \left( \frac{\Gamma_{p}(\{O_i\}_{i=1}^m)}{c \epsilon^2} \right)
    \label{eq:sample-c}
    \end{equation}
    copies of quantum states $\rho$ are necessary for any $\epsilon > 0$, when assuming unbiased and kurtosis-bounded estimation.  
\end{theorem}

\noindent To understand why the Cram\'{e}r--Rao method and the two-step method still apply to the shadow estimation tasks, we note that they provide lower and upper bounds on the variance of estimators which can be converted from and to our $p$-norm error, using the bounded kurtosis for the lower bound and standard statistical techniques like median-of-means estimation for the upper bound. 
The result on few-copy measurements follows from a property of FIM, where increasing $c$, the number of copies of states, can at most increase the FIM by $O(c^2)$. 

We remark that the upper and lower bounds above are not proven for the same class of estimators. Our algorithm is built upon an unbiased estimator, but its output is not exactly unbiased due to the final median-of-means step, so the lower bounds in \Cref{thm:main_full} and \Cref{thm:main_full_c} do not directly apply to it. Without any assumption on the estimators, the lower bounds for oblivious estimation in \Cref{thm:main_ob_full} and \Cref{thm:main_ob_c} still apply, because oblivious estimation is no harder than shadow estimation (\Cref{fig:six-prob}). Therefore, in the high-precision regime, the sample complexity of shadow estimation with single-copy measurements is in general between $\Omega(\Gamma^\ob_{p}/\epsilon^2)$ and $\tO(\Gamma_{p}/\epsilon^2)$, where $\Gamma^\ob_{p} \leq \Gamma_{p}$ (\Cref{lem:simple_relation}), and the upper end cannot be improved by any unbiased, kurtosis-bounded estimator. 

\medskip

We note that  
when $p = \infty$ (which corresponds to the traditional shadow estimation scenario), the two tasks are equivalent in the high-precision regime up to logarithmic overhead. 

\begin{corollary}[Equivalence between shadow estimation and oblivious estimation, $p = \infty$]
Let $p = \infty$. When $\epsilon$ is below a threshold, the necessary and sufficient sample complexities to solve \Cref{prob:p-learning} and to solve \Cref{prob:oblivious-p-learning} using (adaptive) single-copy measurements are both 
\begin{equation}
    N = \tTheta \left( \frac{\Gamma_{\infty}(\{O_i\}_{i=1}^m)}{\epsilon^2} \right),
\end{equation}
where 
\begin{equation}
    \Gamma_{\infty}(\{O_i\}_{i=1}^m) = \Gamma_{\infty}^\ob(\{O_i\}_{i=1}^m) = \sup_{\rho_0\in\cS^\circ} \inf_{M\in\cM} \max_{a \in A}  (I(\rho_0,M)^{-1})_{aa}. 
\end{equation} 
\end{corollary}

\noindent The equivalence can be understood from two different perspectives. On one hand, it directly follows from our instance-optimal bounds (\Cref{thm:main_ob_full} and \Cref{thm:main_full}), and the mathematical property that $\Gamma_{\infty}(\{O_i\}_{i=1}^m) = \Gamma_{\infty}^\ob(\{O_i\}_{i=1}^m)$. On the other hand, we can use any algorithm that can solve the oblivious estimation task to solve the shadow estimation task by estimating each observable $O_i$ one by one. As a result, if the sample complexity of the oblivious estimation task has a $\delta$ dependence of $O(\log(1/\delta))$ where $\delta$ is the failure probability---which is satisfied by our algorithm---then the $\delta$ dependence becomes $O(\log(m/\delta))$ in shadow estimation by union bound, introducing at most a logarithmic overhead. 

Finally, as an example, we show in \Cref{sec:example} when $\{O_i\}_{i=1}^m$ is the set of all Pauli observables, 
\begin{equation}
    \Omega\!\left({d}\right)\le \Gamma^\ob_p\le
    \left\{
    \begin{aligned}
         &O\!\left({d\log d}\right),\;&&\textrm{if~}p\in[2,\infty],\\
         &O\!\left({d^{\frac4p-1}\log d}\right),\;&&\textrm{if~}p\in[1,2),\\
    \end{aligned}
    \right. \quad\; \Omega(d) \leq \Gamma_2 \leq O(d^3\log d). 
\end{equation}
The scalings of our thresholds (see \Cref{tab:summary}) are
\begin{equation}
    \eta^\ob = \Omega(d^{\frac{2}{p}-2}),\quad  \eta^\ob_c =  \begin{cases}
\Omega( d^{-\frac{2}{p}+\frac{1}{2}} (\log d)^{-\frac{1}{2}} / c)
, \quad&\textrm{if}~p\in[1,2),\\
\Omega( d^{-\frac{5}{2}+\frac{4}{p}} (\log d)^{-\frac{1}{2}} / c)
, \quad&\textrm{if}~p\in[2,\infty], 
\end{cases}\quad \overline{\eta} \geq \overline{\eta}^\ob = \Omega(d^{-1}).  
\end{equation}
Here, since the Pauli observables are complete, $\eta^\ob$ stands for $\eta^\ob_\comp$ in \Cref{thm:lower_threshold}, and $\eta^\ob_c$ stands for the threshold derived directly for Pauli observables in \Cref{thm:lower-c-copy-pauli}.
The $p < \infty$ case was previously unknown to the best of our knowledge. 

Throughout our results, the thresholds can scale inverse polynomially in dimension (i.e., exponentially in $n$). 
We emphasize that this threshold is \emph{necessary}: For Pauli observables and infinite norm learning, it was shown that learning all Paulis up to $\varepsilon$ error requires $\tilde{\Theta}(\min\{d/\varepsilon^2,1/\varepsilon^4\})$ samples if allowing 2-copy measurements~\citep{chen2024optimal}. This necessitates a threshold of $d^{-1/2}$ in order for the $\tilde{\Theta}(d/\varepsilon^2)$ term---which corresponds to our sample complexity bound---to become the dominant term. 
Note that when the number of observables is $\text{polylog}(d)$, our lower bound threshold may be improved to $1/\text{polylog}(d)$. However, our upper bound threshold is still $O(1/d)$ and it is unclear whether this is a proof artifact or fundamental. 
Finally, the high-precision regime is practically relevant for sensing. For example, in the task of imaging of weak incoherent light sources with $d$ telescopes, a $d$-dimensional single-photon state is detected where each entry of the density matrix scales as $O(1/d)$. The image information can be recovered from the Fourier transforms of these entries~\citep{zernike1938concept,khabiboulline2019quantum}. In this setting, achieving a learning precision of $O(1/d)$---and hence operating in the high-precision regime---is essential.

\section{Technical overview}

In this section, we provide an overview of the techniques for all theorems mentioned so far.

\subsection{Learning tree method for single-copy measurements}
Our lower bounds for the oblivious distinguishing task in \Cref{prob:oblivious-distinguish} (and thus for all remaining problems) throughout this paper exploit and improve the well-established ``learning tree'' framework gradually developed in a series of work~\cite{aharonov2022quantum,bubeck2020entanglement,chen2022exponential,chen2023complexity,chen2025efficient,chen2024optimal}. From a high level, we model the learning protocol as a decision tree, and a choice of the underlying unknown state $\rho$ results in a distribution on the leaves. We then consider the ensemble $\rho_{\theta,\varphi}$ parametrized by dual observable bases $\{Q_a,T_b\}_{a\in A,b\in B}$ with $A=\{1,...,m\}$ and $B=\{m+1,...,d^2-1\}$, and $\theta\in\R^m$ and $\varphi\in\R^{d^2-m-1}$. We argue that the resulting distributions on leaves are statistically indistinguishable from the distributions on leaves for $\rho_0$ unless the depth of the tree is sufficiently large. In this paper, we take the parametrized ensemble $\rho_{\theta,\varphi}$ to be randomly sampled from a probability distribution $\pi$. Here, $\rho_{\theta,\varphi}$ is defined to be $\rho_{\theta,\varphi}=\rho_0+\tfrac{1}{d}\sum_{a\in A}\theta_a Q_a+\tfrac{1}{d}\sum_{b\in B}\varphi_b T_b$, where $\rho_0$ is the density matrix in the null hypothesis, and $\theta$ and $\varphi$ can be arbitrary vectors such that $\rho_{\theta,\varphi}$ is well-defined and $\norm{\theta}_p=3\epsilon$. 

For each fixed $\rho_0$ and distribution $\pi$ of $(\theta,\varphi)$, we consider protocols with single-copy measurements, the set of which we denote as $\cM$. We utilize the learning tree model equipped with martingale analysis~\cite{chen2023complexity,chen2025efficient,chen2024optimal} to obtain a sample complexity lower bound of $1/\delta$ with 
\begin{align}
\delta = \sup_{M\in\cM}\mathbb{E}_{(\theta,\phi)\sim\pi}  \chi_{M}^2\left(\rho_{\theta,\varphi}\|\rho_0\right)=\sup_{M\in\cM}\mathbb{E}_{(\theta,\phi)\sim\pi}(\theta,\varphi)^\top I(\rho_0,M)(\theta,\varphi),  
\end{align}
where the second step follows from expanding the $\chi^2$-divergence~\cite{hu2025ansatz} and $I(\rho,M)$ denote the classical Fisher information matrix (FIM) obtained by measuring $\rho_0$ with measurement $M$. As one needs to solve the many-versus-one distinguishing any $\rho_0$ and $\pi$, we obtain the sample complexity lower bound as
\begin{align}
N &= \Omega\left(\sup_{\rho_0} \sup_\pi \inf_{M\in\cM} \big( \mathbb{E}_{(\theta,\phi)\sim\pi}  [(\theta,\varphi)^\top I(\rho_0,M)(\theta,\varphi)] \big)^{-1}\right) \\&=\Omega\left(\sup_{\rho_0} \inf_{M\in\cM} \sup_{(\theta,\varphi)} \big( (\theta,\varphi)^\top I(\rho_0,M)(\theta,\varphi) \big)^{-1}\right).\label{eq:overview-distinguish-lower}
\end{align}
The second line above uses the minimax theorem as explain in \Cref{sec:tree}. Detailed derivations are further provided in \Cref{sec:distinguish-complete} (\Cref{thm:distinguish_complete}) for the case when $A=\{1,...,d^2-1\},B=\emptyset$ and \Cref{sec:distinguish-general} (\Cref{thm:distinguish_general}) for general $A$ and $B$. There remains a gap between the lower bound in \eqref{eq:overview-distinguish-lower} and the one for \Cref{prob:oblivious-distinguish} claimed in \Cref{thm:main_ob_full}. However, we will show that these two lower bounds are equivalent in the following for $\epsilon$ below a certain threshold.

\subsection{Duality between distinguishing and estimation, and the threshold on \texorpdfstring{$\epsilon$}{eps}}
We have shown earlier that the (oblivious) distinguishing task \Cref{prob:distinguish}(\ref{prob:oblivious-distinguish}) is no harder than the (oblivious) estimation task \Cref{prob:p-estimation}(\ref{prob:oblivious-p-estimation}). For \Cref{prob:oblivious-p-estimation}, we can obtain a lower bound claimed in \Cref{thm:main_ob_full}
\begin{align}\label{eq:overview-estimation-lower}
N=\Omega\left(\sup_{\rho_0} \inf_{M\in\cM} \max_{\substack{\|\alpha\|_q\le 1/3\epsilon \\ \alpha \in \bR^m}} \alpha^\top(I(\rho_0,M)^{-1})_{AA}\alpha\right)
\end{align}
for $1/p+1/q=1$ using the following quantum metrology argument (which applies to unbiased, kurtosis-bounded estimators; see \Cref{sec:unbiased}). We first note that the inverse of the FIM is an asymptotically attainable lower bound on the variance of unbiased estimators by the Cram\'{e}r–Rao bound~\cite{kay1993fundamentals,lehmann2006theory}. We then show that the scalar $\alpha^\top(I(\rho_0,M)^{-1})_{AA}\alpha$ exactly matches the variance of the estimator for $\theta_\alpha$ asymptotically. We then minimize over measurements $M$ as we can use the optimal single-copy measurements, and maximize over vectors $\alpha$ and states $\rho_0$ as the estimation is oblivious for all $\alpha$ and any reference state $\rho_0$.

To bridge the gap between \eqref{eq:overview-distinguish-lower} and \eqref{eq:overview-estimation-lower}, we show the duality between $p$- and $q$- norm:
\begin{equation}\label{eq:overview-duality}
\max_{\norm{\theta}_p = 3\epsilon} \big((\theta,\varphi)^\top I(\rho_0,M) (\theta,\varphi)\big)^{-1} = \max_{\norm{\alpha}_q = 1/3\epsilon}\alpha^\top (I(\rho_0,M)^{-1})_{AA} \alpha. 
\end{equation}
However, we note that the $(\theta,\varphi)$ satisfy not only the constraint $\norm{\theta}_p\leq 3\epsilon$, but also the well-definedness of the quantum state $\rho_{\theta,\varphi}$ (i.e. $\rho_{\theta,\varphi}\succeq 0$). Therefore, we still need to remove the constraint of $\rho_{\theta,\varphi}\succeq 0$ to reach the lower bound for \Cref{prob:oblivious-distinguish} claimed in \Cref{thm:main_ob_full}. 

An immediate thought is to hope that there is some threshold $\eta$ such that we can drop $\rho_{\theta,\varphi}\succeq 0$ when $\epsilon<\eta$. Unfortunately, even when $\{Q_a\}$ is complete and thus $A=\{1,...,d^2-1\},B=\emptyset$, there does not exist such a choice of $\eta$ that applies to all $\rho_0$ as it can be close to the boundary of the density matrix space. However, we can restrict the range of $\rho_0$ within
\begin{equation}
\cS_{1/2} := \Big\{\rho \Big| \rho = \frac{1}{2}\Big(\sigma + \frac{\id}{d}\Big),\text{ for some density matrix }\sigma\Big\}, 
\end{equation}
without increasing the value of $\delta_{\cM}(O)$ asymptotically. When $\{Q_a\}$ is complete, there are no nuisance parameters, and we can choose
\begin{align}
\epsilon\leq\eta^\ob_\comp :=\frac{1}{6 \norm{(\norm{Q_1}_\infty,\ldots,\norm{Q_{d^2-1}}_\infty)}_q}
\end{align}
to ensure that $\rho_{\theta,\varphi}$ is always well-defined and drop the verbose constraint $\rho_{\theta,\varphi}\succeq 0$. We refer to \Cref{thm:lower_threshold} for the detailed derivation.

However, when $\{Q_a\}$ is incomplete, $I(\rho_0,M)$ is not naturally block-diagonal. To address this issue, we consider the following linear transformation between $(Q~T)=(Q_1~Q_2~\cdots~Q_m~T_{m+1}~T_{m+2}~\cdots~T_{d^2-1})$
and $(Q'~T')=(Q'_1~Q'_2~\cdots~Q'_m~T'_{m+1}~T'_{m+2}~\cdots~T'_{d^2-1})$ related by the following: 
\begin{align}
\begin{pmatrix}
    Q' & T' 
\end{pmatrix}
=
\begin{pmatrix}
    Q+TC_1 & TC_2
\end{pmatrix} = 
\begin{pmatrix}
    Q & T
\end{pmatrix}
\begin{pmatrix}
    \id & 0 \\
    C_1 & C_2 
\end{pmatrix},
\end{align}
where we $C_1 \in \bR^{\abs{B} \times \abs{A}} = \bR^{d^2-1-m \times m}$ and $C_2 \in = \bR^{\abs{B} \times \abs{B}} = \bR^{d^2-1-m \times d^2-1-m}$ to represent linear transformations on the matrix blocks. Here, $C_1$ can be arbitrary and $C_2$ needs to be invertible. As a result, the corresponding FIMs are related by the following: 
\begin{align}
I(\rho_0,M)' &= 
\begin{pmatrix}
    \id & C_1^\top \\
    0 & C_2^\top 
\end{pmatrix} 
\begin{pmatrix}
    I(\rho_0,M)_{AA} & I(\rho_0,M)_{AB} \\ 
    I(\rho_0,M)_{BA} & I(\rho_0,M)_{BB}
\end{pmatrix}
\begin{pmatrix}
    \id & 0 \\
    C_1 & C_2 
\end{pmatrix}.
\end{align}
Finally, we only need to take $C_2 = \id$, which means $T' = T$ and $C_1 =  - ( I(\rho_0,M)_{BB} ) ^{-1}I(\rho_0,M)_{BA}$, where $(I(\rho_0,M)_{BB})^{-1}$ is the pseudoinverse of $I(\rho_0,M)_{BB}$ on its support. The corresponding choice of basis $Q' = Q+TC_1$ makes the corresponding FIM $I(\rho_0,M)'$ block-diagonal. We show for any $\rho_0 \in \cS_{1/2}$ and a correspondingly near-optimal choice of $M^\star(\rho_0)$ (as a function of $\rho_0$), we can choose 
\begin{align}
\epsilon \leq \eta^{\mathrm{ob}} := 
\frac{1}{
12\sqrt{2d(d+1)}
\norm{\left(\norm{Q_1}_2,\ldots,\norm{Q_{\abs{A}}}_2\right)}_q },
\end{align}
to ensure that $\rho_{\theta,\varphi}$ is always well-defined and drop the verbose constraint $\rho_{\theta,\varphi}\succeq 0$. We refer to \Cref{thm:distinguish_thres_general} for the detailed derivation.

\subsection{Learning tree method for few-copy measurements}

We now consider generalizing the above lower bound argument to $c$-copy measurements. Using the learning tree framework, the sample complexity bound for \Cref{prob:oblivious-distinguish} using $c$-copy measurements is $\Omega(c/\delta_{\cM_c}(O))$ where
\begin{align}
\delta_{\cM_c}(O) =  \inf_{\rho_0} \min_{\pi} \sup_{M\in\cM_c} \bE_{(\theta,\varphi)\sim\pi}\, \chi_{M}^2\left(\rho_{\theta,\varphi}^{\otimes c}\|\rho_0^{\otimes c}\right).
\end{align}
Since any choice of $\rho_0$ and $\pi$ provides an upper bound on $\delta_{\cM_c}(O)$, we fix a reference state $\rho_0^\star \in \cS_{1/2}$ and a distribution $\pi^\star$ that are (nearly) optimal for single-copy measurements, and bound $\sup_{M\in\cM_c} \bE_{(\theta,\varphi)\sim\pi^\star} \chi_{M}^2\big(\rho_{\theta,\varphi}^{\otimes c}\|\rho_0^{\star \otimes c}\big)$.
We first decouple $\sqrt{\delta_{\cM_c}(O)} \leq \Delta_1 + \Delta_2$ into the first-order term $\Delta_1$ (from the first order expansion of the $\chi^2$-distance in $(\theta,\varphi)$) and the higher-order terms $\Delta_2$. We first prove that $\Delta_1^2$ is $c^2$ times its single-copy counterpart (up to constant factors), and we obtain for $\epsilon < \eta^\ob$, 
\begin{align}
\Delta_1^2 = c^2 \sup_{M\in\cM} \bE_{(\theta,\varphi)\sim\pi^\star} (\theta,\varphi)^\top I(\rho_0^\star,M)(\theta,\varphi) = \Theta\left(\frac{c^2\epsilon^2}{\Gamma^\ob_{p}(\{O_i\}_{i=1}^m)}\right). 
\end{align}
Indeed, when the deviation appears on one copy only, any $c$-copy measurement acts as an effective single-copy measurement on that copy. 

For the higher-order terms, the Cauchy--Schwarz inequality provides a bound that is independent of the measurement: the contribution from each $S \subseteq [c]$ with $\abs{S} = k$ is at most $(\bE_{(\theta,\varphi)\sim\pi^\star}[r(\theta,\varphi)^k])^{1/2}$, where $r(\theta,\varphi) := \tr\big((\rho_{\theta,\varphi}-\rho_0^\star)\,\rho_0^{\star -1}(\rho_{\theta,\varphi}-\rho_0^\star)\big)$. The main difficulty is that $r(\theta,\varphi)$ involves the nuisance parameters $\varphi$, whose norm is not necessarily small. To address this issue, we again work in the block-diagonal form of the FIM after picking a suitable nearly optimal measurement. 
The optimality of $\pi^\star$ implies that the nuisance parameters carry only a vanishing fraction of the Fisher information, and the contribution of $\varphi$ to $r(\theta,\varphi)$ is negligible. As a result, the higher-order terms are controlled by $\theta$ only, where $\norm{\theta}_p = 3\epsilon$, and we obtain
\begin{align}
\Delta_2 \leq \sum_{k=2}^c \binom{c}{k} \left(\frac{\epsilon}{\eta^\ob\sqrt{2d}}\right)^k.
\end{align}
A threshold is then obtain by assuming the higher-order terms are dominated by the first-order term, i.e., $\Delta_2 = O(\Delta_1)$. 

\subsection{Cram\'{e}r--Rao method}

The Cram\'{e}r--Rao bound~\cite{kay1993fundamentals,lehmann2006theory} states that $V(\rho_\theta,M,\hat{\theta}) \succeq I(\rho_\theta,M)^{-1}$, where $V(\rho_\theta,M,\hat{\theta})$ is the mean square error matrix (MSEM) of an unbiased estimator $\hat{\theta}$ when performing a single-copy measurement $M$ on a parameterized state $\rho_\theta$. In our case, we assume unbiased and $\kappa$-kurtosis-bounded estimators (\Cref{def:bounded-kurtosis}) when estimating $\rho_{\theta,\varphi}$.
To prove the lower bound on the shadow estimation problem (\Cref{thm:estimation-lower-single} for \Cref{prob:p-estimation}), we apply the following techniques: (1) we prove a lower bound on the number of repeated single-copy measurements needed to achieve a certain $p$-average root mean square error (RMSE), which is essentially the $p$-norm of the square root of the diagonal elements of the MSEM, directly using the CR bound; (2) we show that when $p \geq 2$ ($p \geq 1$), if an estimator has a small $p$-norm error with failure probability $\delta < 1/\kappa$, then it also has a small $p$-average RMSE up to a factor of $(1-\sqrt{\kappa\delta})^{-1/2}$ ($(1-\sqrt{\kappa\delta})^{-3/2} \kappa^{1/2}$), without any overhead; (3) we prove that adaptivity in measurements does not decrease the value of the lower bound using the chain rule of Fisher information. Here we require the estimator to have bounded kurtosis in order to rule out the possibility that the estimator takes very large values on a negligible support which may potentially contribute non-negligibly to the $p$-average RMSE.
To prove the lower bound on the oblivious estimation problem (\Cref{thm:estimation-lower-single-ob} for \Cref{prob:oblivious-p-estimation}), the steps are very similar, except that in the second step the conversion from additive error to RMSE applies to all $p \in [1,\infty]$.
Finally, in order to prove the lower bounds with $c$-copy measurements (\Cref{thm:estimation-lower-c}), we use a previously used technique in \Cref{thm:distinguish-c-copy}, which shows $I(\rho_0^{\otimes c},M) \preceq c\sum_{i=1}^c I(\rho_0,G^{[i]})$, where $M$ is a $c$-copy measurement and $G^{[i]}$ is some single-copy measurement that depends on $M$, $\rho_0$ and $i$. It then implies for any $c$-copy measurement $M$, there is some corresponding single-copy measurement $G$ such that $I(\rho_0^{\otimes c},M) \preceq c^2 I(\rho_0, G)$. As a result, the $c$-copy lower bounds are at most a factor of $1/c$ smaller than the single-copy lower bounds.

\subsection{Two-step method: from local to global estimation}

The two-step method is traditionally used in quantum metrology to show the attainability of the CR bound~\cite{barndorff2000fisher,hayashi2011comparison,yang2019attaining}. In those cases, one starts from a locally unbiased estimator, whose expected value equals the true parameter value at (and infinitesimally around) a specific point, that achieves the CR bound at that point. The goal is to show that, without any prior knowledge of the location of the point, one can still achieve the CR bound asymptotically. If $N$ is the total number of samples, the two-step method first uses a negligible number of samples, e.g. $\sqrt{N}$, to obtain a coarse estimate of the parameter, and then applies the locally unbiased estimator defined at the coarse estimation point. The attainability of the CR bound as $N \rightarrow \infty$ can be shown when suitable bounds on the convergence rate are available. 

In the context of quantum learning, we aim to obtain the sample complexity needed for a finite target precision. The main challenge is to apply the two-step method in a non-asymptotic manner, and the advantage we can leverage is that we allow a logarithmic-factor discrepancy, whereas in quantum metrology the CR bound must be strictly attained. To find algorithms that saturate our lower bounds up to a logarithmic factor (\Cref{thm:p-estimation-upper} for \Cref{prob:p-estimation}, and \Cref{thm:ob-estimation-upper} for \Cref{prob:oblivious-p-estimation}), we first show that within a region $\DDD(\rho_0)$ around a specific point $\rho_0$, the locally unbiased estimator at $\rho_0$ achieves the CR bound up to a factor of two for any state in $\DDD(\rho_0)$. The key observation is that $\rho_{\theta,\varphi}$ is linear in $\theta$, such that any locally unbiased estimator of $\theta$ must also be a globally unbiased estimator. To find a coarse estimate $\hrho_0$ of $\rho$ such that $\rho \in \DDD(\hrho_0)$, we show that it is sufficient to apply state tomography with a $O(1/d)$ inaccuracy in operator norm. Here we consider only states whose minimum eigenvalues are above $1/2d$, as we can always add artificial depolarizing noise into the system. Finally, by allowing a logarithmic overhead, we turn our estimator, which has a small $p$-average RMSE, into an estimator with a small $p$-norm error using the median-of-means estimation. In order to attain the lower bounds, we need $\epsilon$ to be small enough such that the coarse estimation step takes a negligible number of samples.

\section{Outlook}

In this work, we give an instance-optimal characterization of high-precision shadow tomography---generalized to $p$-norm error---by expressing the fundamental sample complexity through Fisher-information-based quantities $\Gamma_p$ and $\Gamma_p^{\ob}$, and showing the $\tTheta(\Gamma_p/\epsilon^2)$ scaling below explicit thresholds with matching upper and lower bounds in the relevant regimes. Our results also clarify the role of experimentally feasible entanglement: in the high-precision regime, $c$-copy measurements can improve the leading term by at most a factor $O(1/c)$, indicating substantial asymptotic gains would require access to measurements across much larger numbers of copies. Conceptually, the work bridges quantum learning and quantum metrology by demonstrating that a finite-sample two-step procedure—coarse localization via tomography followed by locally optimal estimation—can be near-globally optimal on a finite neighborhood rather than only asymptotically. Below, we mention some concrete open questions closely related to the current work.

\vspace{0.5em}\noindent\textbf{Linearly-dependent observables. }It would be interesting to consider linearly dependent observables. We can consider (wlog) estimating $o \in \mathbb{R}^m$ through $o = W \theta$, where $W \in \mathbb{R}^{m\times r}$ is isometric and $\theta \in \mathbb{R}^r$ and replace $\theta$ and $\alpha$ with $W \theta$ and $W \alpha$ in the definitions of Problems. We expect our two-step method and Cram\'{e}r-Rao method to still apply because the mean square matrix wrt $o$ is equal to that wrt $\theta$ under a congruent transformation. The learning tree method, however, may be problematic, because the reduction from \Cref{prob:oblivious-distinguish} to \Cref{prob:oblivious-p-estimation} is not straightforward. This is roughly because $p$-norm is not invariant under isometric transformation, except when $p=2$, a special case where we expect our learning tree method to directly apply.

\vspace{0.5em}\noindent\textbf{Tighter thresholds. }We have derived high-precision thresholds below which the lower and upper bounds match up to logarithmic factors. For incomplete observables, we give a construction via a basis transform that yields a valid but loose threshold (in comparison with the threshold for complete observables). It is interesting to explore if we can obtain a tighter characterization of the thresholds that better describes the real threshold of the problem. 

\vspace{0.5em}\noindent\textbf{Entangled measurements on a large number of samples. }We have investigated protocols with joint measurements across $c$ copies at $c=\polylog(d)$ in the high-precision regime. However, it remains open to study protocols with entangled measurements on a large number of samples (e.g. $c=\poly(d)$).

\vspace{0.5em}\noindent\textbf{Conditions for a $\polylog(d)^{-1}$ threshold. }As we have explicitly computed for the case of Pauli observables, the threshold for the high precision regime scales as $\poly(d)^{-1}$. It is natural to ask if we can identify the conditions on the set of observables and the index $p$ such that the threshold for the high precision regime scales as $\polylog(d)^{-1}$.

\section{Related works}

Technically, the lower bounds proved in our work rely on the tools that improved upon the refined learning tree framework from a line of aforementioned works~\cite{bubeck2020entanglement,chen2022exponential,chen2022complexity,chen2025efficient,chen2024optimal}. Here, we mention some other relevant works.

\paragraph{Shadow tomography.}The most standard task in quantum learning is quantum state tomography, which completely recovers the density matrix of an unknown quantum state to high accuracy in trace norm or fidelity~\cite{banaszek2013focus,blume2010optimal,gross2010quantum,hradil1997quantum}. Unfortunately, quantum state tomography suffers from an unavoidable exponential scaling on the system size in sample complexity~\cite{haah2016sample,odonnell2016efficient}. 

To circumvent this exponential barrier, a widely studied alternative task is shadow tomography~\cite{aaronson2018shadow}, the goal of which is to estimate the expectation values of a set of $m$ observables or measurements (up to certain additive, i.e. $\infty$-norm, error). 
A line of works has proposed sample-efficient algorithms using $\poly(\log m, n, 1/\epsilon)$ copies of unknown quantum states and highly entangled measurements~\cite{aaronson2018shadow,aaronson2018online,aaronson2019gentle,buadescu2021improved,brandao2019quantum,gong2023learning,watts2024quantum}. Up to date, the gap between the best known sample complexity upper bound scaling as $O((\log^2 m) n/\epsilon^4)$~\cite{buadescu2021improved} and the well-known lower bound of $\Omega(\log m/\epsilon^2)$ remains an open question. More recently, an optimal shadow tomography protocol with sample complexity $O(\log m/\epsilon^2)$ is proposed in the high precision regime of $\epsilon\lesssim d^{-12}$~\cite{chen2024optimalshadow}, with a very recent breakthrough improving this threshold to $\epsilon\lesssim d^{-1}$~\cite{pelecanos2025debiased}. From the lower bound perspective, it is known that an $\Omega(\min\{2^n/\epsilon^2,1/\epsilon^4\})$ scaling is necessary if one can only perform joint measurements on a restricted number of copies even for estimating all Pauli observables~\cite{chen2024optimal}. 

All above protocols and limitations assume access to highly joint measurements, rendering them impractical to implement on near-term devices. In settings of shadow tomography with single-copy measurements, the classical shadows protocol of Huang, Kueng, and Preskill~\cite{huang2020predicting} requires $O(2^n\log m/\epsilon^2)$, which is proved to be optimal up to log factors when $m \geq 2^n$~\cite{chen2022exponential}. The algorithm is designed to use random basis measurements to produce a classical, unbiased estimator of the state, which is then
used to predict arbitrary observables. Following this random scheme~\cite{elben2023randomized}, more randomized unitaries and measurements are widely applied in quantum device benchmarking~\cite{knill2008randomized,dankert2009exact,emerson2005scalable}, and quantum learning and tomography~\cite{elben2023randomized,brydges2019probing}. Note that with few-copy measurements, efficient sample complexity can be achieved for Pauli estimation when the target precision is low, but not in the high-precision regime~\cite{huang2021information,chen2024optimal,king2024triply}. 

Our work also discussed protocols with single-copy and few-copy measurements, with a pioneering focus on instance optimality and general $p$-norm error, 
which can be regarded as part of a larger body of recent results exploring how local and joint measurements for various quantum learning tasks affect the underlying statistical complexity~\cite{aharonov2022quantum,bubeck2020entanglement,chen2022quantum,chen2024tight,seif2024entanglement,chen2021hierarchy,ye2025exponential,noller2025infinite,chen2022tight2,odonnell2025instance,chen2022toward,fawzi2023quantum,huang2022quantum,huang2021information,chen2024optimalstate}. We refer to the survey~\cite{anshu2023survey} for a more thorough overview along and beyond this line of work.

\paragraph{Quantum metrology.}Quantum metrology studies optimal measurements for estimating parameters of quantum systems with locally unbiased  estimators~\cite{giovannetti2011advances,degen2017quantum,pezze2018quantum,pirandola2018advances}. For estimating quantum states parametrized by a single parameter, the quantum Cram\'{e}r--Rao bound characterizes the ultimate precision limit of estimating one parameter in quantum states through the quantum Fisher information matrix~\cite{helstrom1967minimum,helstrom1968minimum,helstrom1969quantum,holevo2011probabilistic,braunstein1994statistical,barndorff2000fisher,paris2009quantum}, which is the classical Fisher information matrix maximized over all possible quantum measurements on quantum states. However, when estimating quantum states parametrized by multiple parameters, there are cases where the quantum Cram\'{e}r--Rao bound and the quantum Fisher information matrix are not achievable even asymptotically. 

To address this measurement incompatibility issue, one solution is to minimize the weighted sum of estimation variances for a fixed given cost matrix instead of optimizing the estimation variances for all parameters.  In Ref.~\cite{holevo2011probabilistic}, a well-known lower bound, known as the Holevo Cram\'{e}r--Rao bound, is proposed for the weight sum and has been further studied by a line of works~\cite{albarelli2019evaluating,gorecki2020optimal,tsang2020quantum,sidhu2021tight,hayashi2023tight,gardner2024achieving}. In the setting of asymptotically many copies where one can perform joint measurements on infinitely many copies of quantum states, the Holevo Cram\'{e}r--Rao bound is proved to be tight~\cite{kahn2009local,yamagata2013quantum,yang2019attaining}. When restricted to single-copy measurements on single copies of states, the Holevo Cram\'{e}r--Rao bound is attainable for pure states~\cite{matsumoto2002new}. Moreover, the Holevo Cram\'{e}r--Rao bound is shown to be stronger than the Quantum Cram\'{e}r--Rao bound by at most a factor of $2$~\cite{albarelli2019upper,carollo2019quantumness,demkowicz2020multi}. Unlike these results, which assume an infinite number of samples and use the (weighted) estimation variance as the metric, our work lies in the finite-sample regime and applies to a different metric.

\paragraph{Connections between shadow tomography and quantum learning.}It is natural to connect multi-parameter quantum metrology with (local) quantum state and shadow tomography. A line of work in multi-parameter estimation focuses on local state (shadow) tomography, where optimal measurements, known as Fisher-symmetric measurements, were found
and studied for uniformly estimating all parameters in pure states~\cite{li2016fisher,zhu2018universally,vargas2024near}. Conceptually, most closely related to the present work are the aforementioned works from Pelecanos, Spilecki, and Wright~\cite{pelecanos2025debiased}, and from Chen and Zhou~\cite{zhou2026randomized}, which explore the connections between estimators in shadow tomography and quantum metrology; they apply learning tools to metrology, while we apply a metrological approach to learning.  
In Ref.~\cite{zhou2026randomized}, a protocol using randomized and single-copy measurements is proposed as the locally unbiased estimator for quantum metrology, the mean square error matrix of which is further proved to be within a factor of $4$ of being optimal for pure states. They also generalized their results to hold for low-rank states. Later, using the debiased Keyl’s estimator, which requires joint measurements, Ref.~\cite{pelecanos2025debiased} obtains a locally unbiased estimator for quantum metrology with mean square error matrix at most a factor of $2$ from being optimal.

\section{Preliminaries}

In this section, we recap the basic concepts and results required throughout this paper. We use $\norm{A}_p$ to represent the Schatten $p$-norm of matrix $A$, $\norm{A}_\infty$ to represent the operator norm of $A$, and $\norm{v}_p$ to represent the $L_p$ norm of the vector $v$. We also use $\tO$ and $\tTheta$ to hide the poly-logarithmic dependence (on $m$ and $1/\epsilon$) in big-O notations. We will use $[c]$ to denote the set $\{1,2,...,c\}$. When we say ``with high probability'' without specification, we mean with probability at least $0.99$. We use $\id[\cdot]$ as the indicator function. We use $\succeq$ and $\preceq$ to denote partial orders on positive semidefinite matrices.
Given two distributions $p$ and $q$, the total variation distance between $p$ and $q$ is defined to be $d_{TV}(p,q)\coloneqq\tfrac 12\sum_i\abs{p_i-q_i}$, and the $\chi^2$-divergence between $p$ and $q$ is defined to be $\chi^2(q\|p)\coloneqq\sum_i p_i\left(\tfrac{q_i}{p_i}-1\right)^2$.

\subsection{Basic results in quantum information}

We first introduce some standard definitions and calculations in quantum information. We consider quantum states in $d$-dimensional Hilbert spaces represented as a positive semi-definite matrices $\rho\in\mathbb{C}^{d\times d}$ with $\tr(\rho)=1$. In particular, $d=2^n$ for $n$-qubit quantum states. 
When $\rho$ is rank-$1$ and thus $\tr(\rho^2)=1$, it is called a pure state and is denoted as $\ket{\psi}$ or $\ket{\phi}$ throughout this paper. An $n$-qubit observable $O\in\mathbb{C}^{d\times d}$ is a Hermitian matrix. For an $n$-qubit quantum state $\rho$ or observable $O$ and a subset $S\in[n]$, we use $\tr_S(\rho)$ or $\tr_S(O)$ to denote the remaining state or observables after tracing out the qubits in $S$. We denote by $\id$ the identity operator.

For simplicity, we denote by $\cS$ the set of all density matrices, $\cS^\circ$ the set of all full-rank density matrices, and $\cS_{1/2}$ the set of all mixed states that can be written as a linear combination of $\id/d$ and a mixed state $\rho\in\cS$ of equal weight $1/2$. We will also consider parametrized quantum state $\rho_{\theta,\varphi}$ with $\rho_{\0,\0}=\rho_0$. We denote $\DD(\rho_0)$ to be the set of $(\theta,\varphi)$ such that $\rho_{\theta,\varphi}$ is well-defined. We will also consider linear combination of observables. Given a weight vector $\alpha$ and a set of observables $\{O_i\}_i$, we will also denote the linear combination as $O_\alpha=\sum_i\alpha_i O_i$.

\paragraph{Quantum measurements.}A general quantum measurement is represented as positive operator-valued measures (POVMs). An $n$-qubit POVM is represented a set of positive-semidefinite matrices $\{F_s\}_s$ with $\sum_s F_s=\id$ and each $F_s$ a POVM element corresponding to measurement outcome $s$. When measuring a quantum state $\rho$ with POVM $\{F_s\}_s$, the probability of observing outcome $s$ is given by $\tr(F_s\rho)$. 

More generally, POVM is defined to be a mapping of a measurable set $S \subseteq \Omega$, the outcome space, to non-negative Hermitian operators, such that $F(\emptyset) = 0$, $F(\Omega) = \id$, $F(\cup_{i=1}^\infty S_i) = \sum_{i=1}^\infty F(S_i)$ and $\trace(\rho F(S))$ is the probability of obtaining measurement outcomes $\in S$. In this paper, without specification, we will assume without loss of generality $\Omega = \{s\}_s$ is discrete and denotes $F(\{s\}) =: F_s$ as above. 

Throughout this paper, we consider quantum measurements on multiple copies (replicas) of $d$-dimensional quantum states. We denote $\cM$ as set of all qudit (single-copy) POVMs. We denote $\cM^{[d^2]}$ all qudit (single-copy) POVMs with at most $d^2$ outcomes. For any integer $c>1$, we denote $\cM_{c}$ as the set of $c$-copy POVMs on $c$ qudits. We will also consider learning protocols represented by $\cM_{c,N}$, which contains the set of POVMs on $N$ copies of qudit states, which can be decomposed into $N/c$ (possibly adaptive) $c$-copy measurements.

We will also consider Pauli observables in this paper. We define $n$-qubit Pauli group $\cP_n=\{\id_2,X,Y,Z\}^{\otimes n}$ to be the set of $n$-qubit Pauli observables, where
\begin{align}
\id_2=\begin{pmatrix}
1 & 0 \\ 0 & 1
\end{pmatrix},\qquad
X=\begin{pmatrix}
0 & 1 \\ 1 & 0
\end{pmatrix},\qquad
Y=\begin{pmatrix}
0 & -i \\ i & 0
\end{pmatrix},\qquad
Z=\begin{pmatrix}
1 & 0 \\ 0 & -1
\end{pmatrix}
\end{align}
are single-qubit Pauli operators.

\paragraph{Haar random unitaries.}A common tool to prove bounds on the sample complexity of learning problems is to consider random instances. Haar random unitaries is the most common random unitary ensemble. The Haar measure $\mu$ on the unitary group $U(d)$ is the unique probability measure that is invariant under left- and right-multiplication
\begin{align}
\E_{U\sim\mu}f(UV)=\E_{U\sim\mu}f(VU)=\E_{U\sim\mu}f(U),
\end{align}
for any unitary $V\in U(d)$ and function $f(\cdot)$. We can also define a unique rotation invariant measure on states by $U\ket{\psi}$ with $U\sim\mu$ and an arbitrary state $\ket{\psi}$. We will also write $\psi\sim\mu$ throughout this paper.

We will need to explicitly compute expectation values over the Haar measure. A key subroutine is the following folklore formula (see e.g. Ref.~\cite{harrow2013church}):
\begin{align}
\E_{\psi\sim\mu}\left[\ketbra{\psi}^{\otimes k}\right]=\frac{\Pi_k}{\binom{d+k-1}{k}}=\frac{1}{d(d+1)...(d+k-1)}\sum_{\pi\in S_k}\pi^d,
\end{align}
where $\Pi_k$ denotes the projector onto the symmetric subspace $\mathrm{Sym}^k(\mathbb{C}^d)$, $S_k$ is the set of permutation over $k$ elements and $\pi^d$ acts on $(\mathbb{C}^d)^{\otimes k}$ by
\begin{align}
\pi^d |i_1,\ldots,i_k\rangle=|i_{\pi^{-1}(1)},\ldots, i_{\pi^{-1}(k)}\rangle.
\end{align}

\subsection{Tree representations and Le Cam's method}
\label{sec:tree}

We introduce the concepts of modeling adaptive protocols for quantum learning and distinguishing tasks with learning trees~\cite{aharonov2022quantum,chen2022exponential,chen2022complexity,chen2025efficient,chen2024optimal}. Here, we consider an arbitrary protocol using $c$-copy joint measurements and $N$ copies, which is described by an element in $\cM_{c,N}$. We split the protocol into $N/c$ iterations, select $c$ copies of the unknown state $\rho$ at each node, perform a $c$-copy POVM in (a subset of) $\cM_c$, and step to the next iteration corresponding to the outcome. We describe such a procedure with adaptivity using the learning tree representation adapted from Ref.~\cite{chen2024optimal}:
\begin{definition}[Tree representation for protocols in $\cM_{c,N}$~\cite{chen2024optimal}]\label{def:learning-tree}
Given an unknown $n$-qubit quantum state $\rho$, a protocol using $c$-copy joint measurements and $N$ copies of $\rho$ in $\cM_{c,N}$ can be represented as a rooted tree $\cT$ of depth $T=N/c$ with each node on the tree recording the measurement outcome history of the algorithm. It has the following properties:
\begin{enumerate}
\item We assign a probability $p^\rho(u)$ to each node $u$ on the tree $\cT$. The probability assigned to the root $r$ is $p^\rho(r)=1$.
\item At each non-leaf node $u$, we measure a fresh batch $\rho^{\otimes c}$ containing $c$ copies of $\rho$ using a joint measurement $M_u=\{F_s^u\}\in\cM_c$, resulting in a classical outcome $s$. Each child node $v$ corresponding to the classical outcome $s$ of the node $u$ is connected through the edge $e_{u,s}$. 
\item If a node $v$ is the child of a node $u$ through the edge $e_{u,s}$,  the probability assigned to this edge is
\begin{align}
p^\rho(v)=p^\rho(u)\cdot\tr(F_s^u\rho^{\otimes c}).
\end{align}
\item Each root-to-leaf path is of length $T=N/c$. At a leaf node $\ell$, $p^\rho(\ell)$ denotes the probability of the classical memory reaching $\ell$ at the end of the protocol. We also denote the set of leaves of $\cT$ by $\text{leaf}(\cT)$.
\end{enumerate}
At the end of the protocol, the classical post-processing maps each leaf node to a desired output of the protocol.
\end{definition}

Throughout this paper, we will use the learning tree representation defined in \Cref{def:learning-tree} as a tool to prove lower bounds for (oblivious) distinguishing tasks in \Cref{prob:distinguish} and \Cref{prob:oblivious-distinguish}. In particular, we are interested in the following distinguishing task (${\mathsf{P}}_{\rho_0,Q,T,\pi}$). Given access to copies of a $d$-dimensional unknown state $\rho$, the goal is to distinguish between the following two cases with high probability:
\begin{itemize}
    \item (Null hypothesis) $\rho$ is a state $\rho_0$; or
    \item (Alternative hypothesis) $\rho$ is a parametrized state $\rho_{\theta,\varphi}$ randomly sampled from a probability distribution $\pi$. 
\end{itemize}
Here the experimentalist can perform measurement based on knowledge of $\rho_0,Q,T$ and $\pi$, and the value $(\theta,\varphi)$ will be revealed after the measurement is performed. For any protocol that solves \Cref{prob:distinguish} and \Cref{prob:oblivious-distinguish}, this protocol must also solve the task ${\mathsf{P}}_{\rho_0,Q,T,\pi}$ for any choice of $\rho_0,Q,T$ and probability distribution $\pi$.

There is a well-established framework for proving lower bounds of this distinguishing task, consisting of Le Cam’s two-point method~\cite{yu1997assouad}, one-sided likelihood ratio~\cite{chen2022exponential}, and the martingale technique~\cite{chen2022complexity}. We recap the necessary concepts here.

\begin{definition}[Likelihood ratio]\label{defn:likelihood-ratio}
Consider a protocol described by a tree representation $\cT$ for the distinguishing task ${\mathsf{P}}_{\rho_0,Q,T,\pi}$ between the null and alternative hypotheses. For any leaf node $\ell\in\text{leaf}(\cT)$, we define the likelihood ratio to be
\begin{align}
L(\ell)\coloneqq\frac{\mathbb{E}_{(\theta,\varphi)\sim\pi}[p^{\rho_{\theta,\varphi}}(\ell)]}{p^{\rho_0}(\ell)}.
\end{align}
We can also define the likelihood ratio for each edge $e_{u,s}$ and each particular choice of $\rho_{\theta,\varphi}$ as:
\begin{align}
L_{\theta,\varphi}(\ell)\coloneqq\frac{p^{\rho_{\theta,\varphi}}(\ell)}{p^{\rho_0}(\ell)},\qquad L_{\theta,\varphi}(u,s)\coloneqq\frac{p^{\rho_{\theta,\varphi}}(s|u)}{p^{\rho_0}(s|u)}.
\end{align}
\end{definition}

We summarize the toolbox for showing lower bounds under this learning tree representation. 

\begin{lemma}[Toolbox of showing lower bounds]\label{lem:likelihood-ratio}
Suppose $\cT$ is a learning tree with depth $t=N/c$ that solves the distinguishing problem ${\mathsf{P}}_{\rho_0,Q,T,\pi}$ with probability $p_{\text{suc}}$.
\begin{enumerate}
    \item  (Le Cam’s two-point method~\cite{yu1997assouad}) 
    \begin{align}
    {p_{\text{suc}} \leq \frac12+\frac12\,\mathbb{E}_{(\theta,\varphi)\sim\pi}\,d_{TV}\big(p^{\rho_{\theta,\varphi}},p^{\rho_0}\big).}
    \end{align}
    \item (One-sided likelihood ratio~\cite{chen2022exponential}) For any {$\beta\in(0,1]$}, we have
    \begin{align}
    {\mathbb{E}_{(\theta,\varphi)\sim\pi}\,d_{TV}\big(p^{\rho_{\theta,\varphi}},p^{\rho_0}\big)\leq\Pr_{\ell\sim p^{\rho_0},(\theta,\varphi)\sim\pi}[L_{\theta,\varphi}(\ell)\leq\beta]+1-\beta.}
    \end{align}
    \item (Martingale technique~\cite{chen2022complexity,chen2024optimal}) Suppose there is a $\delta>0$ such that for every node $u$ and outcome $s$ we have
    \begin{align}
    \mathbb{E}_{(\theta,\varphi)\sim\pi}\mathbb{E}_{s\sim p^{\rho_0}(s|u)}\left[\left(L_{\theta,\varphi}(u,s)-1\right)^2\right]\leq\delta.
    \end{align}
    We then have
    \begin{align}
    \Pr_{\ell\sim p^{\rho_0},(\theta,\varphi)\sim\pi}[L_{\theta,\varphi}(\ell)\leq0.9]\leq 0.1+O(\delta t).
    \end{align}
\end{enumerate}
\end{lemma}
The proof of \Cref{lem:likelihood-ratio} can be referred to Refs.~\cite{chen2022exponential,chen2022complexity,chen2024optimal}. The first bound is a tree-based one, as we upper bound the total variation distance between the probability distributions of reaching each leaf under the two cases. The second bound is a path-based one, as we prove that the likelihood ratio is not too small for most of the paths from the root to leaves. The third bound is an edge-based one, as we focus on each edge and show that the likelihood ratio over the edge concentrates around $1$.

Note that in \Cref{prob:distinguish} and \Cref{prob:oblivious-distinguish}, $\rho_{\theta,\varphi}$ is specified to be 
\begin{align}
\rho_{\theta,\varphi}=\rho_0+\frac{1}{d}\sum_{a\in A}\theta_a Q_a+\frac{1}{d}\sum_{b\in B}\varphi_b T_b,
\end{align}
where $\rho_0$ is the density matrix in the null hypothesis, $\{Q_a,T_b\}_{a\in A,b\in B}$ are dual observable bases, and $\theta\in\R^m$ and $\varphi\in\R^{d^2-m-1}$  can be arbitrary vectors such that $\rho_{\theta,\varphi}$ is well-defined and $\norm{\theta}_p=3\epsilon$. 
We denote the set of all such $(\theta,\varphi)$ as $\DD_{3\epsilon,p}^{{ Q },{ T }}(\rho_0)$, i.e.
\begin{equation}
\DD_{3\epsilon,p}^{{ Q },{ T }}(\rho_0) =   \{(\theta,\varphi)  | \rho_{\theta,\varphi} \succeq 0,\norm{\theta}_p = 3\epsilon,\theta\in\R^m,\varphi\in\R^{d^2-m-1}\}.   
\end{equation}
where ${ Q } := (Q_1~Q_2~\cdots~Q_m)$ and ${ T } := (T_{m+1}~T_{m+2}~\cdots~T_{d^2-1})$ represent the choice of the dual basis and the superscript $^{{ Q },{ T }}$ highlights the dependence of $(\theta,\varphi)$ on the choice of basis ${ Q }$ and ${ T }$.  
We also denote the set of all probability distributions over such $(\theta,\varphi)$ as $\cD_{3\epsilon,p}^{{ Q },{ T }}(\rho_0)$. Furthermore, we denote the set of all ${Q },{T}$ satisfying \eqref{eq:dual} as 
\begin{equation}
\BB(O) =   \{(Q,T)  | \Tr(O_iQ_a) =d\delta_{ia},\forall i, a\in A, \Tr(O_iT_b) =0, \forall i\in A ,b\in B\},  
\end{equation}
where ${ O } := (O_1~O_2~\cdots~O_m)$. 

Assume there is a protocol in $\cM_{c,N}$ that solves this distinguishing problem in $T=N/c$ rounds. The algorithm can be represented by a learning tree $\cT$ of depth $T$ in \Cref{def:learning-tree}. Let $u$ be an internal node in the learning tree, and $M_u=\{F_s^u\}\in\cM_c$ be the $c$-copy POVM used in the node. Then the probability of observing the outcome $s$ given underlying state $\rho$ is $p^{\rho}(s|u)=\tr(F^u_s\rho^{\otimes c})$.  The likelihood ratio is thus $L_{\theta,\varphi}(u,s)=\tr(F^u_s\rho_{\theta,\varphi}^{\otimes c})/\tr(F^u_s\rho_0^{\otimes c})$, and 
\begin{align}
\begin{split}
\mathbb{E}_{(\theta,\varphi)\sim\pi}\mathbb{E}_{s\sim p^{\rho_0}(s|u)}\left[\left(L_{\theta,\varphi}(u,s)-1\right)^2\right]&=\mathbb{E}_{(\theta,\varphi)\sim\pi}\mathbb{E}_{s\sim p^{\rho_0}(s|u)}\left[\left(\frac{\tr(F^u_s\rho_{\theta,\varphi}^{\otimes c})}{\tr(F^u_s\rho_0^{\otimes c})}-1\right)^2\right]\\
&=\mathbb{E}_{(\theta,\varphi)\sim\pi}\chi_{M_u}^2\left(\rho_{\theta,\varphi}^{\otimes c}\|\rho_0^{\otimes c}\right)\\
\end{split}
\end{align}
where $\chi_M^2(\rho\|\rho')$ denotes the $\chi^2$-distance between the probability distribution over all measurement outcomes using the POVM $M$ over $\rho$ and $\rho'$.  
In particular, in the last line 
\begin{align}
\mathbb{E}_{(\theta,\varphi)\sim\pi}\chi_{M_u}^2\left(\rho_{\theta,\varphi}^{\otimes c}\|\rho_0^{\otimes c}\right) \leq \sup_{M \in \cM_c} \mathbb{E}_{(\theta,\varphi)\sim\pi}\chi_{M}^2\left(\rho_{\theta,\varphi}^{\otimes c}\|\rho_0^{\otimes c}\right).
\end{align}
Let 
\begin{align}
\label{eq:minimax_learning-0}
\delta_{\cM_c}(O) :=  \inf_{\substack{\rho_0 \in \cS^\circ\\ (Q,T) \in \BB(O)}} \min_{\pi\in\cD_{3\epsilon,p}^{{ Q },{ T }}(\rho_0)} \sup_{M\in\cM_c} 
\mathbb{E}_{(\theta,\varphi)\sim\pi}  \chi_{M}^2\left(\rho_{\theta,\varphi}^{\otimes c}\|\rho_0^{\otimes c}\right).
\end{align}
We note that there may be some $\rho_0 \in \cS^\circ$ such that $\DD_{3\epsilon,p}^{{ Q },{ T }}(\rho_0)$ and $\cD_{3\epsilon,p}^{{ Q },{ T }}(\rho_0)$ are empty sets and the distinguishing task is not well defined. In this case, we let $\min_{\pi\in\cD_{3\epsilon,p}^{{ Q },{ T }}(\rho_0)} = \infty$. We also note that we can exchange the orders of $\sup_M$ and $\inf_\pi$ as follows: 
\begin{align}
\delta_{\cM_c}(O) &=  \inf_{\substack{\rho_0 \in \cS^\circ\\ (Q,T) \in \BB(O)}} \sup_{M\in\cM_c} \min_{\pi\in\cD_{3\epsilon,p}^{{ Q },{ T }}(\rho_0)} 
\mathbb{E}_{(\theta,\varphi)\sim\pi}  \chi_{M}^2\left(\rho_{\theta,\varphi}^{\otimes c}\|\rho_0^{\otimes c}\right)\\
&= \inf_{\substack{\rho_0 \in \cS^\circ\\ (Q,T) \in \BB(O)}} \sup_{M\in\cM_c} \min_{(\theta,\varphi)\in 
\DD_{3\epsilon,p}^{{ Q },{ T }}(\rho_0)}  \chi_{M}^2\left(\rho_{\theta,\varphi}^{\otimes c}\|\rho_0^{\otimes c}\right). \label{eq:minimax_learning}
\end{align}
This can be see using Sion's minimax theorem~\cite{sion1958general}. It states that $\inf_{x\in X} \sup_{y \in Y}f(x,y) = \sup_{y \in Y} \inf_{x\in X}  f(x,y)$ if $X$ is a convex and compact subset of a linear topological space, $Y$ is a convex subset of a linear topological space, $x \mapsto f(x,y)$ is continuous and convex, $y \mapsto f(x,y)$ is continuous and concave. First, note that $\mathbb{E}_{(\theta,\varphi)\sim\pi} \chi_{M}^2\left(\rho_{\theta,\varphi}^{\otimes c}\|\rho_0^{\otimes c}\right)$ is a linear (thus convex) function of $\pi$, which lives in a convex, compact set under weak convergence because the domain of $(\theta,\varphi)$ is compact. Next, 
% we view $\mathbb{E}_{(\theta,\varphi)\sim\pi} \chi_{M}^2\left(\rho_{\theta,\varphi}^{\otimes c}\|\rho_0^{\otimes c}\right)$ as a linear (thus concave) function of $M$ by using
we define the following convex combination of $M$: 
\begin{equation}
\label{eq:convex-meas} 
    \lambda \{F_{s_1}\}_{s_1} \oplus (1-\lambda) \{F'_{s_2}\}_{s_2} = \{ \lambda F_{s_1} \}_{s_1} \cup \{(1-\lambda)  F'_{s_2}\}_{s_2},
\end{equation}
where the new POVM operators are the union of original POVM operators multiplied by convex coefficients. 
% Then $\chi_{ \lambda F \oplus (1-\lambda) F' }^2\left(\rho_{\theta,\varphi}^{\otimes c}\|\rho_0^{\otimes c}\right) = \lambda \chi_{  F }^2\left(\rho_{\theta,\varphi}^{\otimes c}\|\rho_0^{\otimes c}\right) + (1-\lambda) \chi_{  F' }^2\left(\rho_{\theta,\varphi}^{\otimes c}\|\rho_0^{\otimes c}\right)$. 
We notice we can convert the optimization over $M$ into to an optimization over a finite-dimensional vector $v(M) \in \bR^{\binom{d^2-1+2c}{2c}}$ by the following calculation. 
\begin{align}
 \chi_{M}^2\left(\rho_{\theta,\varphi}^{\otimes c}\|\rho_0^{\otimes c}\right) =   \sum_{x}\frac{\trace((\rho_{\theta,\varphi}^{\otimes c} - \rho_0^{\otimes c}) M_x)^2}{\trace(\rho_0^{\otimes c} M_x)} = \sum_{\substack{ i_1 \in \bN^{m}, i_2 \in \bN^{d^2-1-m}, \\ \abs{i_1}+\abs{i_2} \leq 2c}} \theta^{i_1}\varphi^{i_2} v(M)_{i_1,i_2}, 
\end{align}
where $v(M)_{i_1,i_2}$ is the polynomial coefficient corresponding to term $\theta^{i_1}\varphi^{i_2}$. The degree of the polynomial does not exceed $2c$ because $\rho_{\theta,\varphi}$ is linear in $\theta,\varphi$. Furthermore, $v$ satisfies $v( \lambda F\oplus (1 - \lambda)F') = \lambda v(F) + (1-\lambda) v(F')$. Then $V_{c} = \{v(M):M \in \cM_c\}$ is a convex set. We have 
\begin{equation}
    \min_{\pi\in\cD_{3\epsilon,p}^{{ Q },{ T }}(\rho_0)}  \sup_{M\in\cM_c} 
\mathbb{E}_{(\theta,\varphi)\sim\pi}  \chi_{M}^2\left(\rho_{\theta,\varphi}^{\otimes c}\|\rho_0^{\otimes c}\right)
= \min_{\pi\in\cD_{3\epsilon,p}^{{ Q },{ T }}(\rho_0)}  \sup_{v \in V_{c}} 
\mathbb{E}_{(\theta,\varphi)\sim\pi}  \sum_{\substack{ i_1 \in \bN^{m}, i_2 \in \bN^{d^2-1-m}, \\ \abs{i_1}+\abs{i_2} \leq 2c}} \theta^{i_1}\varphi^{i_2} v_{i_1,i_2}, 
\end{equation}
ready for the application of the minimax theorem to exchange the order of $\inf_\pi$ and $\sup_v$ where the function $\mathbb{E}_{(\theta,\varphi)\sim\pi}  \sum_{i_1 \in [c], i_2 \in [c]} \theta^{i_1}\varphi^{i_2} v_{i_1,i_2}$ is bilinear in $\pi$ and $v$. Note that 
\begin{equation}
    \sup_{v \in V_{c}}  \min_{\pi\in\cD_{3\epsilon,p}^{{ Q },{ T }}(\rho_0)}  
\mathbb{E}_{(\theta,\varphi)\sim\pi}  \sum_{\substack{ i_1 \in \bN^{m}, i_2 \in \bN^{d^2-1-m}, \\ \abs{i_1}+\abs{i_2} \leq 2c}} \theta^{i_1}\varphi^{i_2} v_{i_1,i_2} = \sup_{M \in \cM_c}  \min_{\pi\in\cD_{3\epsilon,p}^{{ Q },{ T }}(\rho_0)}  
\mathbb{E}_{(\theta,\varphi)\sim\pi}   \chi_{M}^2\left(\rho_{\theta,\varphi}^{\otimes c}\|\rho_0^{\otimes c}\right).
\end{equation}
\eqref{eq:minimax_learning} is then proven.

By \Cref{lem:likelihood-ratio}, we need at least $\Omega(c/\delta_{\cM_{c}})$ to solve the problem ${\mathsf{P}}_{\rho_0,Q,T,\pi}$ for arbitrary $\rho_0$ and $\pi$ (with high probability). Thus, \Cref{prob:distinguish} and \Cref{prob:oblivious-distinguish} must also require sample complexity $\Omega(c/\delta_{\cM_{c}})$.

\subsection{Quantum metrology}

Here, we introduce basic concepts in quantum metrology (i.e. quantum estimation theory), including locally unbiased estimators, Cram\'{e}r--Rao bound, and Fisher information matrix. 

Consider a $d$-dimensional quantum state $\rho_\theta$ in Hilbert space $\mathcal{H}$, where $\theta=(\theta_1,\theta_2,\dots,\theta_m)\in\Theta\subseteq\mathbb{R}^m$. Given a POVM $M=\{M_x\}$, an estimator $\htheta(x)$, a function that maps the measurement outcome to $\mathbb{R}^m$ is called an \emph{unbiased} estimator if
\begin{equation}
\theta=\sum_x\htheta(x)\tr(\rho_\theta M_x)
\end{equation}
for all $\theta\in\Theta$. A \emph{locally unbiased} estimator describes an estimation that is unbiased in the vicinity of one specific value of $\theta$, say $\theta_0$, satisfying
\begin{equation}
\theta_0=\sum_x\htheta(x)\tr(\rho_\theta M_x)\bigg|_{\theta=\theta_0},\quad\delta_{ij}=\frac{\partial}{\partial\theta_i}\sum_x\htheta_j(x)\tr(\rho_\theta M_x)\bigg|_{\theta=\theta_0}.
\end{equation}
Unbiased estimators are locally unbiased, though the converse need not hold. In quantum metrology, the figure of merit is usually taken to be the mean square error matrix (MSEM), defined by
\begin{equation}
V(\rho_\theta,M,\hat{\theta})_{ij}=\sum_x(\hat{\theta}(x)_i-\theta_i)(\hat{\theta}(x)_j-\theta_j)\tr(\rho_\theta M_x).    
\end{equation}
We will also call $V(\rho_\theta,M,\hat{\theta})_{ii}$ the mean square error (MSE) for the estimator $\hat{\theta}_i$. For the estimator $\htheta_\alpha = \sum_i \alpha_i \htheta_i$, the corresponding MSE is given by 
\begin{equation}
\label{eq:alpha-mse}
    \sum_x(\hat{\theta}(x)_\alpha-\theta_\alpha)^2 \tr(\rho_\theta M_x) = \alpha^\top V(\rho_\theta,M,\hat{\theta}) \alpha. 
\end{equation}
Specifically, for locally unbiased estimators, it corresponds to the covariance matrix. 
The MSEM of any (locally) unbiased estimator (at $\theta$) is bounded below by the inverse of the \emph{Fisher information matrix} (FIM),
\begin{equation}
I(\rho_\theta,M)_{ij} :=\sum_{x,p_\theta(x):=\tr(\rho_\theta M_x)\ne0}\frac{1}{p_\theta(x)}\frac{\partial p_\theta(x)}{\partial\theta_i}\frac{\partial p_\theta(x)}{\partial\theta_j},
\end{equation}
through the Cram\'{e}r–Rao (CR) bound,
\begin{equation}
\label{eq:CR-bound}
V(\rho_\theta,M,\hat{\theta}) \succeq I(\rho_\theta,M)^{-1},
\end{equation}
where $V \succeq W $ means $V - W$ is positive semidefinite.
Given $N$ copies of quantum state $\rho_\theta$, using the fact that the FIM is additive, we have 
\begin{equation}
V(\rho_\theta,M^{\otimes N},\hat{\theta}^{(N)}) \succeq I(\rho_\theta^{\otimes N},M^{\otimes N})^{-1} = \frac{1}{N} I(\rho_\theta,M)^{-1},    
\end{equation}
for any locally unbiased estimator $\hat{\theta}^{(N)}$. 

The FIM is closely related to the second-order derivative of $\chi^2$-distance. Consider the probability distribution $p_\theta(x)$. Assume $p_\theta(x) > 0$ for all $x$. Then,
\begin{multline}
\label{eq:taylor}
\chi^2(p_{\theta+\dtheta}\|p_\theta)=\sum_{x}p_\theta(x)\left(\frac{p_{\theta+\dtheta}(x)}{p_\theta(x)}-1\right)^2=\sum_{x}\frac{(p_{\theta+\dtheta}(x)-p_\theta(x))^2}{p_\theta(x)}\\
=\sum_{x}\frac{(\sum_i\partial_i p_\theta(x)\dtheta_i)^2}{p_\theta(x)} =\sum_{x,i,j}\frac{\partial_i p_\theta(x)\partial_j p_\theta(x)\dtheta_i \dtheta_j}{p_\theta(x)}=\sum_{i,j}I(M)_{ij}\dtheta_i \dtheta_j.
\end{multline}
In particular, when $p_{\theta}$ is linear in $\theta$, the above calculation is exact for finite $\dtheta$, a situation that we will encounter later.

Consider the asymptotic situation where we have $N$ copies of quantum state $\rho_\theta$ and $N \rightarrow \infty$. Under certain regularity conditions, the CR bound is saturable by maximum likelihood estimators~\cite{van2000asymptotic} is the sense that 
\begin{equation}
    \sqrt{N}(\hat{\theta}^{(N)}_{\text{MLE}} - \theta) \xrightarrow{\text{d}} \cN(0, I(\rho_\theta,M)^{-1}),
\end{equation}
i.e., $\sqrt{N}(\hat{\theta}^{(N)}_{\text{MLE}} - \theta)$ converges in distribution to a normal distribution centered around $0$ with variance equal to $I(\rho_\theta,M)^{-1}$ as $N 
\rightarrow \infty$. When $N$ is finite, it is unclear whether the CR bound is always saturable. However, when we focus on local estimation at a specific point $\theta_0$ (whose value is known in prior), the following locally unbiased estimator automatically achieves the CR bound, i.e.,
\begin{equation}
V(\rho_\theta,M,\hat{\theta}^{\text{opt}})|_{\theta=\theta_0}=I(\rho_\theta,M)^{-1}|_{\theta=\theta_0}.
\end{equation}
where for a measurement outcome $y$, 
\begin{equation}
\label{eq:opt}
\hat{\theta}^{\text{opt}}(y;\theta_0) := (\theta_{0})_i +\sum_x(\gamma_{i,x}|_{\theta=\theta_0})\delta_{xy},\quad 
\gamma_{i,x}=\sum_j (I(\rho_\theta,M)^{-1})_{ij}\frac{\partial_j p_\theta(x)}{p_\theta(x)}.
\end{equation}
We can easily verify that this estimator achieves the CR bound: 
\begin{align}
V(\rho_\theta,M,\hat{\theta}^{\text{opt}})_{ij}&=\sum_x (\hat{\theta}^{\text{opt}}(x;\theta)_i - \theta_i)(\hat{\theta}^{\text{opt}}(x;\theta)_j - \theta_j)\tr(\rho_\theta M_x)\\
&=\sum_{x,x_1,x_2}\gamma_{i,x_1} \delta_{x_1,x}\gamma_{j,x_2} \delta_{x_2,x} p_\theta(x) =\sum_x\gamma_{i,x}\gamma_{j,x} p_\theta(x) \\
&=\sum_{i'j'}(I(\rho_\theta,M)^{-1})_{ii'}\frac{\partial_{i'}p_\theta(x)}{p_\theta(x)}(I(\rho_\theta,M)^{-1})_{jj'}\frac{\partial_{j'}p_\theta(x)}{p_\theta(x)}p_\theta(x)\\
&=\sum_{i'j'}(I(M)^{-1})_{ii'}I(M)_{i'j'}(I(M)^{-1})_{jj'}= (I(\rho_\theta,M)^{-1})_{ij}.
\end{align}
In practice, however, the optimal locally unbiased estimator \eqref{eq:opt} cannot be directly applied as the local point $\theta_0$ is unknown in prior. 

\subsection{Tail bounds}

We will need the following Chebyshev's inequality and Hoeffding's inequality.

\begin{lemma}[Chebyshev's inequality, see e.g., Corollary 1.6.3 in~\cite{vershynin2018high}]
Let $X$ be a random variable with finite non-zero variance $\sigma^2$ and finite expected value $\mu$. Then for any real number $k>0$, we have
\begin{align}
\Pr[\abs{X-\mu}\geq k\sigma]\leq\frac{1}{k^2}.
\end{align}
\end{lemma}

\begin{lemma}[Hoeffding's inequality~\cite{hoeffding1963probability}, see e.g., Theorem 2.2.6 in~\cite{vershynin2018high}]\label{lem:Hoeffding}
Let $X_1, \cdots, X_M$ be $M$ independent random variables such that $X_i\in [a_i, b_i]$ for every $i$. For any $t>0$, we have
\begin{align}
\Pr[\abs{\sum_{i=1}^M \left(X_i-\E[X_i]\right)}\geq t]\leq 2\exp\left(-\frac{2t^2}{\sum_{i=1}^M(b_i-a_i)^2}\right).
\end{align}
\end{lemma}

As a corollary, let $X_1, \cdots, X_M$ be i.i.d. random variables in $[0, 1]$ and $\E[X_i]=\mu$ for every $i$. Then 
\begin{align}
\Pr\left[\sum_{i=1}^M X_i\leq (1-\delta)M\mu\right] \leq e^{-2M\mu^2\delta^2}.
\end{align}

\section{Lower bounds for many-versus-one distinguishing}
\label{sec:distinguish}

In this section, we prove a lower bound for the (oblivious) many-versus-one distinguishing problem in \Cref{prob:distinguish} (\Cref{prob:oblivious-distinguish}). We recap the distinguishing task we consider here for concreteness. Given access to copies of a $d$-dimensional unknown state $\rho$, the goal is to distinguish between the following two cases:
\begin{itemize}
\item (Null hypothesis) $\rho$ is a state $\rho_0$; or
\item (Alternative hypothesis) $\rho$ is a parametrized state $\rho_{\theta,\varphi}$ randomly sampled from a probability distribution $\pi\in\cD_{3\epsilon,p}^{{ Q },{ T }}(\rho_0)$. Here, $\rho_{\theta,\varphi}$ is defined to be
\begin{align}
\rho_{\theta,\varphi}=\rho_0+\frac{1}{d}\sum_{a\in A}\theta_a Q_a+\frac{1}{d}\sum_{b\in B}\varphi_b T_b,
\end{align}
where $\rho_0$ is the density matrix in the null hypothesis, $\{Q_a,T_b\}_{a\in A,b\in B}$ are dual observable bases, and $\theta\in\R^m$ and $\varphi\in\R^{d^2-m-1}$  can be arbitrary vectors such that $\rho_{\theta,\varphi}$ is well-defined and $\norm{\theta}_p=3\epsilon$. 
\end{itemize}

If $B=\emptyset$ and $\{Q_a\}_{a\in A}$ forms a complete basis of the Hilbert space, we called $\{Q_a\}_{a\in A}$ a set of complete observables. Recall from the learning tree framework, we showed in \eqref{eq:minimax_learning} that any protocol described by $\cM_{c,N}$ in \Cref{def:learning-tree} requires $\Omega(c/\delta_{\cM_{c}})$ samples with
\begin{align}
\delta_{\cM_c}(O)
= \inf_{\substack{\rho_0 \in \cS^\circ\\ (Q,T) \in \BB(O)}} \sup_{M\in\cM_c}  \min_{\pi\in\cD_{3\epsilon,p}^{{ Q },{ T }}(\rho_0)} 
\mathbb{E}_{(\theta,\varphi)\sim\pi}\chi_{M}^2\left(\rho_{\theta,\varphi}^{\otimes c}\|\rho_0^{\otimes c}\right). 
\end{align}

\subsection{Single-copy measurements: Exact lower bounds}

We first assume that we can only use single-copy measurements ($c=1$).

\subsubsection{Complete observables}\label{sec:distinguish-complete}

We start with the case when the observables form a complete basis. The task reduces to distinguishing between the following two cases:
\begin{itemize}
\item (Null hypothesis) $\rho$ is a state $\rho_0$; or
\item (Alternative hypothesis) $\rho$ is a parametrized state $\rho_{\theta}$ randomly sampled from a probability distribution $\pi\in\cD_{3\epsilon,p}^{ Q }(\rho_0)$. Here, $\rho_{\theta}$ is defined to be
\begin{align}
\rho_{\theta}=\rho_0+\frac{1}{d}\sum_{a\in A}\theta_a Q_a,
\end{align}
where $\rho_0$ is the density matrix in the null hypothesis, $\{Q_a\}_{a\in A}$ forms a complete observable basis, and $\theta\in\R^m$ can be arbitrary vectors such that $\rho_{\theta}$ is well-defined and $\norm{\theta}_p=3\epsilon$. 
\end{itemize} 

\begin{theorem}[Lower bound for \Cref{prob:distinguish}(\ref{prob:oblivious-distinguish}) with complete observables, $c=1$, and $p$-norm error]\label{thm:distinguish_complete}
Using the adaptive measurement strategy with single-copy measurements, the sample complexity required to solve the many-versus-one distinguishing tasks above with any $\rho_\theta$ well-defined and $\theta\in\DD_{3\epsilon,p}^{ Q }(\rho_0)$ is
\begin{align}
N=\Omega\left(\sup_{\rho_0 \in \cS^\circ} \inf_{M\in\cM} \max_{\theta\in\DD_{3\epsilon,p}^{ Q }(\rho_0)}\left(\theta^\top I(\rho_0,M)\theta\right)^{-1}\right)
\end{align}
for any $\epsilon>0$, $p\in[1,\infty]$, where $I(\rho_0,M) := I(\rho_\theta,M)|_{\theta=\0}$ and $\rho_{\theta} = \rho_0 + \frac{1}{d}\sum_{a \in A} \theta_a Q_a$.
\end{theorem}

\begin{proof}
The sample complexity lower bound from \eqref{eq:minimax_learning} is then given by $\Omega(1/\delta_{\cM}(O))$ where
\begin{align}
\delta_{\cM}(O)=\inf_{\rho_0 \in \cS^\circ} \sup_{M\in\cM} \min_{\pi\in\cD_{3\epsilon,p}^{ Q }(\rho_0)}\mathbb{E}_{\theta\sim\pi}\chi_{M}^2\left(\rho_{\theta}\|\rho_0\right).
\end{align}
Note that $\{Q_a\}_{a\in A}$ here can be uniquely defined using $\{O_i\}_{i=1}^m$. 
Also note that the optimal $\pi$ distribution in this expression would trivially be the delta function at the value of $\theta$ that minimizes the $\chi^2$-distance. Thus, we have 
\begin{align}
\delta_{\cM}(O)= \inf_{\rho_0 \in \cS^\circ} \sup_{M\in\cM} \min_{\theta\in\DD_{3\epsilon,p}^{ Q }(\rho_0)} \chi_{M}^2\left(\rho_{\theta}\|\rho_0\right).
\end{align}

Given a fixed measurement $M=\{M_s\}\in\cM$ and denote $p_s=\tr(\rho_0 M_s)$, we compute the minimization part above as:
\begin{align}
\min_{\theta\in\DD_{3\epsilon,p}^{ Q }(\rho_0)}\chi^2_M(\rho_\theta \| \rho_0)=\min_{\theta\in\DD_{3\epsilon,p}^{ Q }(\rho_0)}\sum_{s}p_s\left(\frac{\tr(M_s\rho_\theta)}{p_s}-1\right)^2.
\end{align}
We can extend $\tr(M_s\rho_\theta)$ as $\tr(M_s\rho_\theta)=p_s+\sum_i\theta_i\partial_i p_s$. We thus have\allowdisplaybreaks
\begin{align}
\min_{\theta\in\DD_{3\epsilon,p}^{ Q }(\rho_0)}\chi^2_M(\rho_\theta \| \rho_0)&=\min_{\theta\in\DD_{3\epsilon,p}^{ Q }(\rho_0)}\sum_{s}p_s\left(\frac{\tr(M_s\rho_\theta)}{p_s}-1\right)^2\\
&=\min_{\theta\in\DD_{3\epsilon,p}^{ Q }(\rho_0)}\sum_{s}\frac{(\sum_i\theta_i\partial_i p_s)^2}{p_s}\\
&=\min_{\theta\in\DD_{3\epsilon,p}^{ Q }(\rho_0)}\sum_{s,i,j}\frac{\theta_i\theta_j\partial_i p_s\partial_j p_s}{p_s}\\
&=\min_{\theta\in\DD_{3\epsilon,p}^{ Q }(\rho_0)}\theta^\top I(\rho_0,M)\theta.
\end{align}
Therefore, we have
\begin{align}
\delta_{\cM}(O)= \inf_{\rho_0 \in \cS^\circ}   \sup_{M\in\cM} \min_{\theta\in\DD_{3\epsilon,p}^{ Q }(\rho_0)}\theta^\top I(\rho_0,M)\theta.
\end{align}
Taking the inverse, the sample complexity lower bound is then given by
\begin{align}
N=\Omega\left(\sup_{\rho_0 \in \cS^\circ} \inf_{M\in\cM} \max_{\theta\in\DD_{3\epsilon,p}^{ Q }(\rho_0)}\left(\theta^\top I(\rho_0,M)\theta\right)^{-1}\right)
\end{align}
as claimed.
\end{proof}

\subsubsection{A general set of observables}\label{sec:distinguish-general}

We then consider a general dual observable basis ${ Q } = \{Q_a\}_{a\in A}$ and ${ T } = \{T_b\}_{b \in B}$ with single-copy measurement protocols ($c=1$). We show the following sample complexity lower bound:
\begin{theorem}[Lower bound for \Cref{prob:distinguish}(\ref{prob:oblivious-distinguish}) at $c=1$ and $p$-norm error]\label{thm:distinguish_general}
Using the adaptive measurement strategy with single-copy measurements, the sample complexity required to solve the many-versus-one distinguishing tasks above with any $\rho_{\theta,\varphi}$ well-defined and $(\theta,\varphi)\in\DD_{3\epsilon,p}^{{ Q },{ T }}(\rho_0)$ is
\begin{align}
\label{eq:single-lower}
N=\Omega\left(\sup_{\substack{\rho_0 \in \cS^\circ\\ (Q,T) \in \BB(O)}} \inf_{M\in\cM} \max_{(\theta,\varphi)\in\DD_{3\epsilon,p}^{{ Q },{ T }}(\rho_0)}\left((\theta,\varphi)^\top I(\rho_0,M)(\theta,\varphi)\right)^{-1}\right)
\end{align}
for any $\epsilon>0$, $p\in[1,\infty]$, where $I(\rho_0,M) = I(\rho_{\theta,\varphi},M)|_{\theta=\varphi=\0}$ and $\rho_{\theta,\varphi} = \rho_0 + \frac{1}{d}\sum_{a \in A} \theta_a Q_a + \frac{1}{d}\sum_{b \in B} \varphi_b T_b$.
\end{theorem}

\begin{proof}
The sample complexity lower bound from \eqref{eq:minimax_learning} is then given by $\Omega(1/\delta_{\cM}(O))$ where
\begin{align}
\delta_{\cM}(O)
&=  \inf_{\substack{\rho_0 \in \cS^\circ\\ (Q,T) \in \BB(O)}} \sup_{M\in\cM} \min_{\pi\in\cD_{3\epsilon,p}^{{ Q },{ T }}(\rho_0)}\mathbb{E}_{(\theta,\varphi)\sim\pi}\chi_{M}^2\left(\rho_{\theta,\varphi}\|\rho_0\right)\\
&= \inf_{\substack{\rho_0 \in \cS^\circ\\ (Q,T) \in \BB(O)}} \sup_{M\in\cM}  \min_{(\theta,\varphi)\in\DD_{3\epsilon,p}^{{ Q },{ T }}(\rho_0)} \chi_{M}^2\left(\rho_{\theta,\varphi}\|\rho_0\right)
\end{align}
Given a fixed measurement $M=\{M_s\}\in\cM$ and denote $p_s=\tr(\rho_0 M_s)$, we compute the minimization part above as:
\begin{align}
\min_{(\theta,\varphi)\in\DD_{3\epsilon,p}^{{ Q },{ T }}(\rho_0)}\chi^2_M(\rho_{\theta,\varphi}\| \rho_0)=\min_{(\theta,\varphi)\in\DD_{3\epsilon,p}^{{ Q },{ T }}(\rho_0)}\sum_{s}p_s\left(\frac{\tr(M_s\rho_{\theta,\varphi})}{p_s}-1\right)^2.
\end{align}
We can extend $\tr(M_s\rho_{\theta,\varphi})$ as $\tr(M_s\rho_\theta)=p_s+\sum_i\theta_i\partial_i p_s+\sum_{i'}\varphi_{i'}\partial_{i'}p_s$. We thus have
\begin{align}
\begin{split}
\min_{(\theta,\varphi)\in\DD_{3\epsilon,p}^{{ Q },{ T }}(\rho_0)}\chi^2_M(\rho_\theta \| \rho_0)&=\min_{\theta\in\DD_{3\epsilon,p}^{{ Q },{ T }}(\rho_0)}\sum_{s}p_s\left(\frac{\tr(M_s\rho_\theta)}{p_s}-1\right)^2\\
&=\min_{(\theta,\varphi)\in\DD_{3\epsilon,p}^{{ Q },{ T }}(\rho_0)}\sum_{s}\frac{(\sum_i\theta_i\partial_i p_s+\sum_{i'}\varphi_{i'}\partial_{i'}p_s)^2}{p_s}\\
&=\min_{(\theta,\varphi)\in\DD_{3\epsilon,p}^{{ Q },{ T }}(\rho_0)}(\theta,\varphi)^\top I(\rho_0,M)(\theta,\varphi).
\end{split}
\end{align}
Therefore, we have
\begin{align}
\delta_{\cM}(O)=  \inf_{\substack{\rho_0 \in \cS^\circ\\ (Q,T) \in \BB(O)}}\sup_{M\in\cM} \min_{(\theta,\varphi)\in\DD_{3\epsilon,p}^{{ Q },{ T }}(\rho_0)}(\theta,\varphi)^\top I(\rho_0,M)(\theta,\varphi).
\end{align}
Taking the inverse, the sample complexity lower bound is then given by
\begin{align}
N=\Omega\left(\sup_{\substack{\rho_0 \in \cS^\circ\\ (Q,T) \in \BB(O)}} \inf_{M\in\cM}  \max_{(\theta,\varphi)\in\DD_{3\epsilon,p}^{{ Q },{ T }}(\rho_0)}\left((\theta,\varphi)^\top I(\rho_0,M)(\theta,\varphi)\right)^{-1}\right)
\end{align}
as claimed.
\end{proof}

\subsection{Duality between many-versus-one distinguishing and parameter estimation}

Before going further to explore the lower bounds in \Cref{thm:distinguish_complete} and \Cref{thm:distinguish_general}, we first reveal a key relationship between the lower bounds of these many-versus-one distinguishing tasks and parameter estimation tasks. 

In the complete observables case, from \Cref{thm:distinguish_complete}, we have shown that the sample lower bound for distinguishing between $\rho_0$ versus $\rho_\theta$ is given by 
\begin{align}
N=\Omega\left( \sup_{\rho_0 \in \cS^\circ} \inf_{M\in\cM}\max_{\theta\in\DD_{3\epsilon,p}^{ Q }(\rho_0)}\left(\theta^\top I(\rho_0,M)\theta\right)^{-1}\right)
\end{align}
for any $\epsilon>0$, $p\in[1,\infty]$. However, this sample complexity lower bound in the form
\begin{align}
\max_{\theta}\left(\theta^\top I(\rho_0,M)\theta\right)^{-1}
\end{align}
is different from the usual lower (upper) bound in parameter estimation, which is in the form of 
\begin{align}
\max_{\alpha} \alpha^\top I(\rho_0,M)^{-1} \alpha.
\end{align}
The latter represents a lower bound on the MSE that can be naturally derived from the CR bound (\eqref{eq:CR-bound}) using \eqref{eq:alpha-mse} in parameter estimation. 

A similar situation appears for the case of general observables. The sample complexity lower bound for distinguishing between $\rho_0$ versus $\rho_{\theta,\varphi}$ from \Cref{thm:distinguish_general} is given by
\begin{align}
N=\Omega\left( \sup_{\substack{\rho_0 \in \cS^\circ\\ (Q,T) \in \BB(O)}} \inf_{M\in\cM} \max_{(\theta,\varphi)\in\DD_{3\epsilon,p}^{{ Q },{ T }}(\rho_0)}\left((\theta,\varphi)^\top I(\rho_0,M)(\theta,\varphi)\right)^{-1}\right)
\end{align}
for any $\epsilon>0$, $p\in[1,\infty]$. This is again in the form of 
\begin{align}
\max_{\theta,\varphi}\left((\theta,\varphi)^\top I(\rho_0,M)(\theta,\varphi)\right)^{-1},
\end{align}
which is different from the lower (upper) bound in parameter estimation in the form of 
\begin{align}
\max_{\alpha} (\alpha,0)^\top I(\rho_0,M)^{-1}(\alpha,0) . 
\end{align}

Surprisingly, we show the following duality relationship between many-versus-one distinguishing and parameter estimation using the following mathematical fact, by temporarily extending the domains of each optimization to the entire space with bounded $p$- and $q$-norms.
\begin{lemma}[Duality between many-versus-one distinguishing and parameter estimation]\label{lem:duality}
Let the FIM $I(\rho_0,M)$ be a $(\abs{A}+\abs{B}) \times (\abs{A}+\abs{B})$ positive definite matrix. We use $(\theta,\varphi)$ to represent a joint vector whose first $\abs{A}$ columns are $\theta$ and last $\abs{B}$ columns are $\varphi$. We have
\begin{align}
\begin{split}
\max_{\norm{\theta}_p \geq 1, \theta \in \bR^{\abs{A}},\varphi \in \bR^{\abs{B}}} \left((\theta,\varphi)^\top I(\rho_0,M) (\theta,\varphi)\right)^{-1}&=\max_{\norm{\alpha}_q \leq 1, \alpha \in \bR^{\abs{A}}} (\alpha,0)^\top I(\rho_0,M)^{-1} (\alpha,0)\\
&=\max_{\norm{\alpha}_q \leq 1, \alpha \in \bR^{\abs{A}}} \alpha^\top (I(\rho_0,M)^{-1})_{AA} \alpha
\end{split}
\end{align}
for $1/p+1/q=1$, where the subscript $_{AA}$ represents the upper-left block of the FIM restricted to $a \in A$.
\end{lemma}

\begin{proof}
We only need to show
\begin{align}
\label{eq:temp-pf}
\min_{\norm{\theta}_p \geq 1, \theta \in \bR^{\abs{A}},\varphi \in \bR^{\abs{B}}} (\theta,\varphi)^\top I(\rho_0,M) (\theta,\varphi)  = \frac{1}{ \max_{\norm{\alpha}_q \leq 1, \alpha \in \bR^{\abs{A}}} (\alpha,0)^\top I(\rho_0,M)^{-1} (\alpha,0)}.
\end{align}
Assume $(\theta^\star,\varphi^\star)$ minimizes $(\theta,\varphi)^\top I(\rho_0,M) (\theta,\varphi)$, $\norm{\theta^\star}_p = 1$. Then there exists $\alpha^\diamond$ such that $\abs{\alpha^\diamond \cdot \theta^\star} = 1$ and $\norm{\alpha^\diamond}_q \leq 1$. The Cauchy-Schwarz inequality implies
\begin{align}
\begin{split}
(\alpha^\diamond,0)^\top I(\rho_0,M)^{-1} (\alpha^\diamond,0) \geq& \frac{1}{(\theta^\star,\varphi^\star)^\top I(\rho_0,M) (\theta^\star,\varphi^\star)}\\
=& \frac{1}{\min_{\norm{\theta}_p \geq 1, \theta \in \bR^{\abs{A}},\varphi \in \bR^{\abs{B}}} (\theta,\varphi)^\top I(\rho_0,M) (\theta,\varphi)}, 
\end{split}
\end{align}
which means the left-hand side of \eqref{eq:temp-pf} is no smaller than the right-hand side. 

On the other hand, assume $\alpha^\star$ maximizes $(\alpha,0)^\top I(\rho_0,M)^{-1} (\alpha,0)$, which must satisfy $\norm{\alpha^\star}_q = 1$. Let $(\theta^{\diamond},\varphi^{\diamond}) = \frac{I(\rho_0,M)^{-1}(\alpha^\star,0)}{\norm{( I(\rho_0,M)^{-1}(\alpha^\star,0) )_A }_p}$, we have\allowdisplaybreaks
\begin{align}
(\theta^\diamond,\varphi^{\diamond})^\top I(\rho_0,M) (\theta^\diamond,\varphi^{\diamond})&= \frac{(\alpha^\star,0)^\top I(\rho_0,M)^{-1}(\alpha^\star,0)}{\norm{(I(\rho_0,M)^{-1}(\alpha^\star,0))_A}_p^2}\\
&= \frac{(\alpha^\star,0)^\top I(\rho_0,M)^{-1}(\alpha^\star,0)}{\norm{\alpha^\star}_q^2 \norm{(I(\rho_0,M)^{-1}(\alpha^\star,0))_A}_p^2}\\
&\leq \frac{(\alpha^\star,0)^\top I(\rho_0,M)^{-1}(\alpha^\star,0)}{((\alpha^\star,0)^\top I(\rho_0,M)^{-1}(\alpha^\star,0))^2}\\
&= \frac{1}{(\alpha^\star,0)^\top I(\rho_0,M)^{-1}(\alpha^\star,0)},
\end{align}
which means the left-hand side of \eqref{eq:temp-pf} is no larger than the right-hand side. Here we use $(\cdot)_A$ to denote the first $\abs{A}$ columns of a vector, and we use the generalized Cauchy-Schwarz: $\norm{\beta}_p\norm{\alpha}_q \geq \abs{\beta\cdot\alpha}$. 
\end{proof}

\subsection{Single-copy measurements: Threshold on \texorpdfstring{$\epsilon$}{eps} to saturate the upper bound}

\Cref{lem:duality} has shown that 
\begin{align}
\min_{\norm{\theta}_p \geq 1, \theta \in \bR^{\abs{A}},\varphi \in \bR^{\abs{B}}} (\theta,\varphi)^\top I(\rho_0,M) (\theta,\varphi)  = \frac{1}{{ \max_{\norm{\alpha}_q \leq 1, \alpha \in \bR^{\abs{A}}} (\alpha,0)^\top I(\rho_0,M)^{-1} (\alpha,0)}}
\end{align}
for $1/p+1/q=1$. However, a more detailed observation on \Cref{thm:distinguish_complete} and \Cref{thm:distinguish_general} indicates that the duality cannot directly work as the maximization over $(\theta,\varphi)$ is over $\DD_{3\epsilon,p}^{{ Q },{ T }}(\rho_0)$ instead of simply the sphere $\norm{\theta}_p=3\epsilon$. 

Here, we propose a threshold $\eta$ of $\epsilon$ such that the lower bounds in \Cref{thm:distinguish_complete} and \Cref{thm:distinguish_general} can be written in the form of maximization over $\norm{\theta}_p=3\epsilon$.

\subsubsection{Complete observables}
We start with the case of complete observables in \Cref{thm:distinguish_complete}. We have shown that the lower bound is given by $\Omega(1/\delta_{\cM}(O))$, where $\delta_{\cM}(O)$ is defined by
\begin{align}
\delta_{\cM}(O)=\inf_{\rho_0 \in \cS^\circ} \sup_{M\in\cM} \min_{\theta\in\DD_{3\epsilon,p}^{ Q }(\rho_0)}\theta^\top I(\rho_0,M)\theta.
\end{align}

Our goal is to apply the duality result in \Cref{lem:duality} to the above lower bound. To this end, we prove the following result.
\begin{theorem}[Lower bound in the high-precision regime with complete observables, $c=1$, and $p$-norm error]\label{thm:lower_threshold}
Using the adaptive measurement strategy with single-copy measurements, the sample complexity required to solve the many-versus-one distinguishing task (\Cref{prob:oblivious-distinguish}) between any well-defined $\rho_\theta$ ($\theta\in\DD_{3\epsilon,p}^{ Q }(\rho_0)$) and $\rho_0$ is
\begin{align}
N=\Omega\left( \sup_{\rho_0 \in \cS^\circ} \inf_{M\in\cM} \max_{\norm{\theta}_p=1}\frac{1}{\epsilon^2}\left(\theta^\top I(\rho_0,M)\theta\right)^{-1}\right)=\Omega\left( \sup_{\rho_0 \in \cS^\circ} \inf_{M\in\cM} \max_{\norm{\alpha}_q \leq 1} \frac{1}{\epsilon^2}\alpha^\top I(\rho_0,M)^{-1} \alpha\right)
\end{align}
with any
\begin{align}
\epsilon\leq\eta^\ob_\comp:=\frac{1}{6 \norm{(\norm{Q_1}_\infty,\ldots,\norm{Q_{d^2-1}}_\infty)}_q},
\end{align}
and for any $p\in[1,\infty]$, where $I(\rho_0,M) = I(\rho_\theta,M)|_{\theta=\0}$ and $\rho_{\theta} = \rho_0 + \frac{1}{d}\sum_{a \in A} \theta_a Q_a$.
\end{theorem}

\begin{proof}
Our hope is to show that there is some threshold $\eta$ such that when $\epsilon<\eta$, we have
\begin{align}
\begin{split}
\min_{\theta\in\DD_{3\epsilon,p}^{ Q }(\rho_0)}\theta^\top I(\rho_0,M)\theta&=\min_{\norm{\theta}_p=3\epsilon}\theta^\top I(\rho_0,M)\theta\\
&=9\epsilon^2\min_{\norm{\theta}_p=1}\theta^\top I(\rho_0,M)\theta\\
&=\frac{9\epsilon^2}{{\max_{\norm{\alpha}_q \leq 1} \alpha^\top I(\rho_0,M)^{-1} \alpha}}
\end{split}
\end{align}
where the last step is due to the duality result in \Cref{lem:duality}. 

Unfortunately, there does not exist such a choice of $\eta$ that applies to all $\rho_0$, because when $\rho_0$ is close to the boundary of the density matrix space, e.g., when $\rho_0$ is singular, $\theta \in \DD_{3\epsilon,p}^{ Q }(\rho_0)$ cannot holds for all $\theta$ satisfying $\norm{\theta}_p = 3\epsilon$ even when $\epsilon$ is small. Luckily, we can restrict the range of $\rho_0$ without significantly increasing the value of 
\begin{equation}
\delta_{\cM}(O)= \inf_{\rho_0 \in \cS^\circ} \sup_{M\in\cM}  \min_{\theta \in \DD_{3\epsilon,p}^{ Q }(\rho_0)} \chi^2_M(\rho_\theta \| \rho_0). 
\end{equation}
Defining
\begin{equation}
    \cS_{1/2} := \Big\{\rho \Big| \rho = \frac{1}{2}\Big(\sigma + \frac{\id}{d}\Big),\text{ for some density matrix }\sigma\Big\}, 
\end{equation}
and using 
\begin{align}
\label{eq:s-1/2-equiv}
\begin{split}
&\quad \theta^\top I\left(\frac{1}{2}\sigma+\frac{\id}{2d},M\right) \theta = \chi^2_M\bigg(\frac{1}{2}\sigma+\frac{\id}{2d} + \frac{1}{d}\sum_a\theta_aQ_a \bigg\| \frac{1}{2}\sigma+\frac{\id}{2d}\bigg)\\
&= \sum_s \frac{\trace(M_s \frac{1}{d} \sum_a \theta_a Q_a)^2}{\trace(M_s (\frac{1}{2}\sigma+\frac{\id}{2d}))}\leq  \sum_s \frac{\trace(M_s \frac{1}{d} \sum_a \theta_a Q_a)^2}{\trace(M_s (\frac{1}{2}\sigma+\frac{1}{2d}\sigma))}\\
&= \frac{2d}{d+1} \sum_s \frac{\trace(M_s \frac{1}{d} \sum_a \theta_a Q_a)^2}{\trace(M_s \sigma)} \\
&\leq 2 \sum_s \frac{\trace(M_s \frac{1}{d} \sum_a \theta_a Q_a)^2}{\trace(M_s \sigma)} = 2 \chi^2_M(\sigma_\theta \| \sigma) = 2 \theta^\top I\left(\sigma,M\right) \theta, 
\end{split}
\end{align}
where $\sigma_\theta = \sigma + \frac{1}{d}\sum_{a} \theta_a Q_a$, we have both 
\begin{multline}
\frac{1}{2} \inf_{\rho_0 \in \cS_{1/2}} \sup_{M\in\cM}  \min_{\theta \in \DD_{3\epsilon/2,p}^{ Q }(\rho_0)} \theta^\top I\left(\rho_0,M\right) \theta \leq \inf_{\rho_0 \in \cS^\circ} \sup_{M\in\cM}  \min_{\theta \in \DD_{3\epsilon,p}^{ Q }(\rho_0)} \theta^\top I\left(\rho_0,M\right) \theta\\ \leq \inf_{\rho_0 \in \cS_{1/2}} \sup_{M\in\cM}  \min_{\theta \in \DD_{3\epsilon,p}^{ Q }(\rho_0)} \theta^\top I\left(\rho_0,M\right) \theta,
\end{multline}
because $2\theta \in \DD_{3\epsilon,p}^{ Q }(\rho_0)$ implies $\theta \in \DD_{3\epsilon/2,p}^{ Q }(\rho_0/2+\id/2d)$, 
and 
\begin{multline}
\frac{1}{2}\inf_{\rho_0 \in \cS_{1/2}} \sup_{M\in\cM}  \min_{\norm{\theta}_p = 3\epsilon} \theta^\top I\left(\rho_0,M\right) \theta \leq \inf_{\rho_0 \in \cS^\circ} \sup_{M\in\cM}  \min_{\norm{\theta}_p = 3\epsilon} \theta^\top I\left(\rho_0,M\right) \theta \\ \leq \inf_{\rho_0 \in \cS_{1/2}} \sup_{M\in\cM}  \min_{\norm{\theta}_p = 3\epsilon} \theta^\top I\left(\rho_0,M\right) \theta. 
\end{multline}

Now we try to find an $\eta$ such that when $\epsilon < \eta$, $\norm{\theta}_p = 3 \epsilon$ implies $\theta \in \DD_{3\epsilon,p}^{ Q }(\rho_0)$ for all $\rho_0 \in \cS_{1/2}$. For example, we can choose $\eta$ as claimed
\begin{align}
\eta =  \frac{1}{6 \norm{(\norm{Q_1}_\infty,\ldots,\norm{Q_{d^2-1}}_\infty)}_q },
\end{align}
such that 
\begin{align}
\norm{ \sum_a \theta_a Q_a }_\infty \leq \sum_a \abs{\theta_a}\norm{Q_a}_\infty \leq  \norm{\theta}_p \norm{(\norm{Q_1}_\infty,\ldots,\norm{Q_{d^2-1}}_\infty)}_q \leq \frac{1}{2}.
\end{align}
This makes sure when $\rho_0 \in \cS_{1/2}$, for any $\norm{\theta}_p = 3\epsilon$, $\rho_\theta$ is well-defined. 
Thus when $\epsilon \leq \eta$, 
\begin{align}
\begin{split}
&\quad  \sup_{\rho_0 \in \cS^\circ} \inf_{M\in\cM} \max_{\theta \in \DD_{3\epsilon,p}^{ Q }(\rho_0)}\left(\theta^\top I(\rho_0,M)\theta\right)^{-1}
\leq \sup_{\rho_0 \in \cS^\circ } \inf_{M\in\cM}  \max_{\norm{\theta}_p = 3\epsilon}\left(\theta^\top I(\rho_0,M)\theta\right)^{-1}
\end{split}
\end{align}
by definition, and 
\begin{align}
\begin{split}
&\quad \sup_{\rho_0 \in \cS^\circ} \inf_{M\in\cM} \max_{\theta \in \DD_{3\epsilon,p}^{ Q }(\rho_0)}\left(\theta^\top I(\rho_0,M)\theta\right)^{-1}\geq \sup_{\rho_0 \in \cS_{1/2}} \inf_{M\in\cM} \max_{\theta \in \DD_{3\epsilon,p}^{ Q }(\rho_0)}\left(\theta^\top I(\rho_0,M)\theta\right)^{-1}\\
&= \sup_{\rho_0 \in \cS_{1/2}} \inf_{M\in\cM} \max_{\norm{\theta}_p = 3\epsilon}\left(\theta^\top I(\rho_0,M)\theta\right)^{-1}\geq \frac{1}{2} \sup_{\rho_0 \in \cS^\circ } \inf_{M\in\cM} \max_{\norm{\theta}_p = 3\epsilon}\left(\theta^\top I(\rho_0,M)\theta\right)^{-1}.
\end{split}
\end{align}
The theorem is then proved using \Cref{lem:duality}. 
\end{proof}

\subsubsection{A general set of observables}

We also compute a threshold for the case of a general set of observables in the distinguishing task between $\rho_0$ and $\rho_{\theta,\varphi} = \rho_0 + \frac{1}{d}\sum_a \theta_a Q_a + \frac{1}{d} \sum_b \varphi_b T_b$, where and $(\theta,\varphi)\in\DD_{3\epsilon,p}^{{ Q },{ T }}(\rho_0)$. To so do, we will explore different choices of dual basis ${ Q } = \{Q_a\}_{a\in A}$ and ${ T } = \{T_b\}_{b\in B}$ and pick a suitable one where a threshold can be easily obtain, as the dual basis is no longer uniquely defined as in the case of complete observables.

We first prove the following lemma. 
\begin{lemma}[Independence of the lower bound on the choice of dual basis]
\label{lemma:independence}
The quantity 
\begin{align}
(I(\rho_0,M)^{-1})_{AA}
\end{align}
is invariant under different choices of ${ Q }$, ${ T }$ as long as \eqref{eq:dual} holds. 
\end{lemma}
\begin{proof}
Let both $(Q~T) := (Q_1~Q_2~\cdots~Q_m~T_{m+1}~T_{m+2}~\cdots~T_{d^2-1})$
and $(Q'~T') := (Q'_1~Q'_2~\cdots~Q'_m~T'_{m+1}~T'_{m+2}~\cdots~T'_{d^2-1})$ \sloppy be valid choices of dual basis that satisfy \eqref{eq:dual}, i.e., 
\begin{align}
\trace(O_i Q_a) = d \delta_{ia},\quad \trace(O_i T_b) = 0, 
\end{align}
where we view $(Q~T)$ (or $(Q'~T')$) as a $d \times (d(d^2-1)) $ block row vector where each block is a $d \times d$ dual observable. Then they must be related by the following: 
\begin{align}
\begin{pmatrix}
    Q' & T' 
\end{pmatrix}
=
\begin{pmatrix}
    Q+TC_1 & TC_2
\end{pmatrix} = 
\begin{pmatrix}
    Q & T
\end{pmatrix}
\begin{pmatrix}
    \id & 0 \\
    C_1 & C_2 
\end{pmatrix},
\end{align}
where we $C_1 \in \bR^{\abs{B} \times \abs{A}} = \bR^{d^2-1-m \times m}$ and $C_2 \in = \bR^{\abs{B} \times \abs{B}} = \bR^{d^2-1-m \times d^2-1-m}$ to represent linear transformations on the matrix blocks. $C_1$ can be arbitrary and $C_2$ needs to be invertible. 
Furthermore, 
\begin{align}
\rho_{\theta,\varphi} = \rho_0 + \frac{1}{d}\sum_a \theta_a Q_a + \frac{1}{d} \sum_b \varphi_b T_b = \rho_0 + \frac{1}{d}\sum_a \theta'_a Q'_a + \frac{1}{d} \sum_b \varphi'_b T'_b,
\end{align}
where 
\begin{align}
\begin{pmatrix}
    \theta' \\ \varphi' 
\end{pmatrix} = 
\begin{pmatrix}
    \id & 0 \\
    C_1 & C_2 
\end{pmatrix}^{-1}
\begin{pmatrix}
    \theta \\ \varphi
\end{pmatrix}
= \begin{pmatrix}
    \id & 0 \\
    - C_2^{-1} C_1 & C_2^{-1}
\end{pmatrix} 
\begin{pmatrix}
    \theta \\ \varphi
\end{pmatrix}
= \begin{pmatrix}
    \theta \\
    - C_2^{-1} C_1 \theta + C_2^{-1} \varphi
\end{pmatrix}. 
\end{align}
As a result, the corresponding FIMs are related by the following: 
\begin{align}
I(\rho_0,M)' &= 
\begin{pmatrix}
    \id & C_1^\top \\
    0 & C_2^\top 
\end{pmatrix} 
\begin{pmatrix}
    I(\rho_0,M)_{AA} & I(\rho_0,M)_{AB} \\ 
    I(\rho_0,M)_{BA} & I(\rho_0,M)_{BB}
\end{pmatrix}
\begin{pmatrix}
    \id & 0 \\
    C_1 & C_2 
\end{pmatrix}.
\label{eq:fim-trans}
\end{align}
Since $({I(\rho_0,M)}^{-1})_{AA}$ is equal to the inverse of the Schur complement of ${I(\rho_0,M)}$, i.e. 
\begin{align}
    ({I(\rho_0,M)}^{-1})_{AA} = ( I(\rho_0,M)_{AA} - I(\rho_0,M)_{AB}  I(\rho_0,M)_{BB}^{-1} I(\rho_0,M)_{BA} )^{-1},
\end{align}
we only need to show the Schur complement $I(\rho_0,M)_{AA} - I(\rho_0,M)_{AB}  I(\rho_0,M)_{BB}^{-1} I(\rho_0,M)_{BA}$ is invariant under the basis transformation. 
This can be seen from \eqref{eq:fim-trans}. 
\allowdisplaybreaks
\begin{align}
\begin{split}
&\quad I(\rho_0,M)_{AA}' - I(\rho_0,M)_{AB}'  {I(\rho_0,M)_{BB}'}^{-1} I(\rho_0,M)_{BA}' \\
&= 
I(\rho_0,M)_{AA} + C_1^\top I(\rho_0,M)_{BA} + I(\rho_0,M)_{AB} C_1 + C_1^\top I(\rho_0,M)_{BB} C_1 \\
&\qquad - (I(\rho_0,M)_{AB} + C_1^\top I(\rho_0,M)_{BB}) I(\rho_0,M)_{BB}^{-1} (I(\rho_0,M)_{BA} + I(\rho_0,M)_{BB} C_1)\\
& = I(\rho_0,M)_{AA} - I(\rho_0,M)_{AB}  {I(\rho_0,M)_{BB}}^{-1} I(\rho_0,M)_{BA}. 
\end{split}
\end{align}
\end{proof}

Now we define 
\begin{equation}
    \Gamma^{\ob}_{p}(\{O_i\}_{i=1}^m) := \sup_{\rho_0 \in \cS^\circ} \inf_{M\in\cM}  \max_{\norm{\alpha}_q \leq 1} \alpha^\top (I(\rho_0,M)^{-1}_{AA} \alpha ,
\end{equation}
which is only a function of observables $\{O_i\}_{i=1}^m$ because the expression is independent of different choices of the dual basis. Below we show there is a unique choice of basis that allows us to derive the following threshold result. 

\begin{theorem}[Lower bound in the high-precision regime with general observables, $c=1$, and $p$-norm error]\label{thm:distinguish_thres_general}
Using the adaptive measurement strategy with single-copy measurements, we consider the many-versus-one distinguishing task (\Cref{prob:oblivious-distinguish}) between any well-defined $\rho_{\theta,\varphi}$ ($(\theta,\varphi)\in\DD_{3\epsilon,p}^{{ Q },{ T }}(\rho_0)$) and $\rho_0$. 
The sample complexity required to solve it is
\begin{align}
\begin{split}
N&=\Omega\left(\sup_{\rho_0 \in \cS^\circ} \inf_{M\in\cM} \max_{\norm{\theta}_p=1}\frac{1}{\epsilon^2}\left((\theta,\varphi)^\top I(\rho_0,M)(\theta,\varphi)\right)^{-1}\right)\\
&=\Omega\left(\sup_{\rho_0 \in \cS^\circ} \inf_{M\in\cM} \max_{\norm{\alpha}_q \leq 1} \frac{1}{\epsilon^2}(\alpha,0)^\top I(\rho_0,M)^{-1} (\alpha,0)\right) = \Omega\left(\frac{\Gamma^{\ob}_{p}(\{O_i\}_{i=1}^m)}{\epsilon^2}\right), \\
\end{split}
\end{align}
with any
\begin{align}\label{eq:thres_general_distinguish}
\epsilon\leq\eta^{\ob} := 
% \eta^{\mathrm{ob}} := 
\frac{1}{
12\sqrt{2d(d+1)}
\norm{\left(\norm{Q_1}_2,\ldots,\norm{Q_{\abs{A}}}_2\right)}_q 
}
% \inf_{\rho_0 \in \cS_{1/2}} \frac{1}{6 \norm{\left(\norm{Q'_1(\rho_0)}_\infty,\ldots,\norm{Q'_{\abs{A}}(\rho_0)}_\infty\right)}_q },
\end{align}
and for any $p\in[1,\infty]$, where $I(\rho_0,M):= I(\rho_{\theta,\varphi},M)|_{\theta=\varphi=\0}$ and $\rho_{\theta,\varphi} = \rho_0 + \frac{1}{d}\sum_{a \in A} \theta_a Q_a + \frac{1}{d}\sum_{b \in B} \varphi_b T_b$. 
Note that here $\{Q_a\}_{a\in A}$ in \eqref{eq:thres_general_distinguish} can be any choice satisfying \eqref{eq:dual}. 
\end{theorem}

\begin{proof}
Throughout the proof we fix $Q$ and $T$ to be some constant dual basis satisfying \eqref{eq:dual}. Note that the sample complexity lower bound \eqref{eq:single-lower} maximizes over $Q$ and $T$, and fixing $Q$ and $T$ will not increase its value.

\vspace{0.08in}
\noindent{\textbf{Step~1: Simplified situation: $I(\rho_0,M)$ is block-diagonal.}}
\vspace{0.08in}

We first consider the situation the FIM $I(\rho_0,M)$ is block-diagonal for some $\rho_0$ and $M$, i.e., $I(\rho_0,M)_{AB} = 0$. 
Then we want to find a threshold on $\epsilon$ such that 
\begin{align}
\begin{split}
\min_{ (\theta,\varphi)\in\DD_{3\epsilon,p}^{{ Q },{ T }}(\rho_0)} \chi^2_{M} (\rho_{\theta,\varphi} \| \rho_0) 
&= \min_{ (\theta,\varphi)\in\DD_{3\epsilon,p}^{{ Q },{ T }}(\rho_0)}(\theta,\varphi)^\top I(\rho_0,M) (\theta,\varphi) \\
&= \min_{\norm{\theta}_p = 3\epsilon} (\theta,\varphi)^\top I(\rho_0,M) (\theta,\varphi).
\end{split}\label{eq:gel-th-expr} 
\end{align}
Since the FIM $I(\rho_0,M)$ is block-diagonal,   
\begin{align}
\begin{split}
\min_{\norm{\theta}_p = 3\epsilon} (\theta,\varphi)^\top I(\rho_0,M) (\theta,\varphi) &= \min_{\norm{\theta}_p = 3\epsilon} (\theta^\top I(\rho_0,M)_{AA} \theta + \varphi^\top I(\rho_0,M)_{BB} \varphi) \\
&= \min_{\norm{\theta}_p = 3\epsilon} \theta^\top I(\rho_0,M)_{AA} \theta.
\end{split}
\end{align}
Furthermore, if for all $\theta$ satisfying $\norm{\theta}_p = 3\epsilon$, $(\theta,0) \in \DD_{3\epsilon,p}^{{ Q },{ T }}(\rho_0)$, then 
\begin{align}
\label{eq:temp-11}
\min_{(\theta,\varphi)\in\DD_{3\epsilon,p}^{{ Q },{ T }}(\rho_0)} (\theta,\varphi)^\top I(\rho_0,M) (\theta,\varphi) = \min_{\norm{\theta}_p = 3\epsilon} \theta^\top I(\rho_0,M)_{AA} \theta,
\end{align}
because
\begin{align}
\min_{(\theta,\varphi)\in\DD_{3\epsilon,p}^{{ Q },{ T }}(\rho_0)} (\theta,\varphi)^\top I(\rho_0,M) (\theta,\varphi) \leq \min_{(\theta,\varphi)\in\DD_{3\epsilon,p}^{{ Q },{ T }}(\rho_0)} (\theta,0)^\top I(\rho_0,M) (\theta,0) = \min_{\norm{\theta}_p = 3\epsilon} \theta^\top I(\rho_0,M)_{AA} \theta,
\end{align}
and 
\begin{align}
\min_{(\theta,\varphi)\in\DD_{3\epsilon,p}^{{ Q },{ T }}(\rho_0)} (\theta,\varphi)^\top I(\rho_0,M) (\theta,\varphi) \geq \min_{\norm{\theta}_p = 3\epsilon} (\theta,\varphi)^\top I(\rho_0,M) (\theta,\varphi) = \min_{\norm{\theta}_p = 3\epsilon} \theta^\top I(\rho_0,M)_{AA} \theta. 
\end{align}
Therefore, under the assumption that $I(\rho_0,M)$ is block-diagonal, \eqref{eq:gel-th-expr} holds if $\norm{\theta}_p = 3\epsilon$ implies $(\theta,0)\in\DD_{3\epsilon,p}^{{ Q },{ T }}(\rho_0)$. 

\vspace{0.08in}
\noindent{\textbf{Step~2: Basis/coordinate transformation to make $I(\rho_0,M)$ block-diagonal.}}
\vspace{0.08in}

Next, we show any $I(\rho_0,M)$ can be made block-diagonal by properly choosing basis $\{Q_a,T_b\}$ for each $\rho_0$ and $M$. The basis transformation 
\begin{align}
\begin{pmatrix}
    Q' & T' 
\end{pmatrix}
=
\begin{pmatrix}
    Q+TC_1 & TC_2
\end{pmatrix} = 
\begin{pmatrix}
    Q & T
\end{pmatrix}
\begin{pmatrix}
    \id & 0 \\
    C_1 & C_2 
\end{pmatrix}
\end{align}
corresponds to the new FIM
\begin{align}
I(\rho_0,M)' = 
\begin{pmatrix}
    \id & C_1^\top \\
    0 & C_2^\top 
\end{pmatrix} 
\begin{pmatrix}
    I(\rho_0,M)_{AA} & I(\rho_0,M)_{AB} \\ 
    I(\rho_0,M)_{BA} & I(\rho_0,M)_{BB}
\end{pmatrix}
\begin{pmatrix}
    \id & 0 \\
    C_1 & C_2 
\end{pmatrix}.
\end{align}
Here we take $C_2 = \id$, which means
\begin{align}
T' = T
\end{align}
and 
\begin{align}
C_1 =  - ( I(\rho_0,M)_{BB} ) ^{-1}I(\rho_0,M)_{BA} 
\label{eq:block-diagonal-basis}
\end{align}
where $(I(\rho_0,M)_{BB})^{-1}$ is the pseudoinverse of $I(\rho_0,M)_{BB}$ on its support, is a solution of $C_1$ such that 
\begin{align}
I(\rho_0,M)'_{AB} =  I(\rho_0,M)_{AB}  C_2 + C_1^\top  I(\rho_0,M)_{BB} C_2  = 0.
\end{align}
The corresponding choice of basis $Q' = Q+TC_1$ makes the corresponding FIM $I(\rho_0,M)'$ block-diagonal. It implies 
\begin{equation}
\frac{1}{6 \norm{\left(\norm{Q'_1}_\infty,\ldots,\norm{Q'_{\abs{A}}}_\infty\right)}_q }
\end{equation}
is a threshold on $\epsilon$ below which \eqref{eq:gel-th-expr} holds after the basis transformation. Below we fix $Q'$ to be the above choice of basis (for different $\rho_0$ and $M$) such that $I(\rho_0,M)'$ is block-diagonal. Additionally, when computing the sample complexity lower bound \eqref{eq:single-lower} the basis transformation 
\begin{equation}
     Q' = Q + T C_1 , \quad T' = T
\end{equation} 
is equivalent to the reparameterization 
\begin{equation}
    \varphi' = C_1\theta + \varphi, \quad \theta' = \theta, 
\end{equation}
in the sense that 
\begin{gather}
\rho_0 + \frac{1}{d}\left(\sum_{a}\theta_a Q_a+\sum_{a}\varphi_a' T_a\right) ~~\Leftrightarrow~~ \rho_0 + \frac{1}{d}\left(\sum_{a}\theta_a Q_a' +\sum_{a}\varphi_a T_a \right), 
 \\
(\theta,\varphi') \in \DD_{3\epsilon,p}^{Q,T}(\rho_0) ~~\Leftrightarrow~~ (\theta,\varphi) \in \DD_{3\epsilon,p}^{Q',T}(\rho_0), 
\label{eq:temp-12} \\
(\theta,\varphi')^\top I(\rho_0,M) (\theta,\varphi') ~~\Leftrightarrow~~
(\theta,\varphi)^\top I(\rho_0,M)' (\theta,\varphi).
\label{eq:temp-13}
\end{gather}
We will use these relations in later derivations.

\vspace{0.08in}
\noindent{\textbf{Step~3: Proving the sample complexity when $\epsilon \leq \eta(M^\star)$ for some measurement $M^\star$.}}
\vspace{0.08in}

Pick any $M^\star(\rho_0) \in \cM$ satisfying for all $\rho_0 \in \cS_{1/2}$ and $\epsilon$ such that $\DD_{3\epsilon,p}^{{ Q },{ T }}(\rho_0)$ is nonempty,
\begin{equation}
\label{eq:optimal-rho-m}
      \inf_{M\in\cM} \max_{(\theta,\varphi) \in \DD_{3\epsilon,p}^{{ Q},{ T }}(\rho_0)} \left((\theta,\varphi)^\top I(\rho_0,M) (\theta,\varphi)\right)^{-1} \geq  \frac{1}{2} \max_{(\theta,\varphi) \in \DD_{3\epsilon,p}^{{ Q},{ T }}(\rho_0)} \left((\theta,\varphi)^\top I(\rho_0,M^\star) (\theta,\varphi)\right)^{-1}.
\end{equation}
Let 
\begin{gather}
\label{eq:Q-prime}
    Q'(\rho_0) := Q + T C_1^\star, \quad C_1^\star =  - (I(\rho_0,M^\star)_{BB})^{-1} I(\rho_0,M^\star)_{BA},\\
\varphi'(\rho_0,\theta) := C_1^\star \theta + \varphi,  
\end{gather}
in the following proof. 
Then, when 
\begin{align}
\epsilon \leq \eta = \inf_{\rho_0 \in \cS_{1/2}} \frac{1}{6 \norm{\left(\norm{Q'_1(\rho_0)}_\infty,\ldots,\norm{Q'_{\abs{A}}(\rho_0)}_\infty\right)}_q },
\end{align}
we have 
\begin{multline}
\sup_{\substack{\rho_0 \in \cS^\circ\\ (Q,T) \in \BB(O)}}  \inf_{M\in\cM} \max_{(\theta,\varphi) \in \DD_{3\epsilon,p}^{{ Q },{ T }}(\rho_0)}\left((\theta,\varphi)^\top I(\rho_0,M)(\theta,\varphi)\right)^{-1} \leq  \\ 
\sup_{\substack{\rho_0 \in \cS^\circ\\ (Q,T) \in \BB(O)}}  \inf_{M\in\cM} \max_{\norm{\theta}_p = 3\epsilon}\left((\theta,\varphi)^\top I(\rho_0,M)(\theta,\varphi)\right)^{-1} 
= \sup_{\rho_0 \in \cS^{\circ}} \inf_{M\in\cM}  \max_{\norm{\theta}_p = 3\epsilon} \left((\theta,\varphi)^\top I(\rho_0,M)(\theta,\varphi)\right)^{-1},
\end{multline}
where we use \lemmaref{lemma:independence}. 
On the other hand, using \eqref{eq:optimal-rho-m}, 
\eqref{eq:temp-12}, 
\eqref{eq:temp-13},
\eqref{eq:temp-11}, 
\lemmaref{lem:duality}, 
\lemmaref{lemma:independence}
and \eqref{eq:s-1/2-equiv}, we have 
\allowdisplaybreaks
\begin{align}
&\quad  \sup_{\substack{\rho_0 \in \cS^\circ}}  \inf_{M\in\cM} \max_{(\theta,\varphi) \in \DD_{3\epsilon,p}^{{ Q },{ T }}(\rho_0)}\left((\theta,\varphi)^\top I(\rho_0,M)(\theta,\varphi)\right)^{-1}
 \\
&\geq \sup_{\rho_0 \in \cS_{1/2}} \inf_{M\in\cM}  \max_{(\theta,\varphi) \in \DD_{3\epsilon,p}^{{ Q },{ T }}(\rho_0)} \left((\theta,\varphi)^\top I(\rho_0,M) (\theta,\varphi)\right)^{-1}\\
&\geq \frac{1}{2} \sup_{\rho_0 \in \cS_{1/2}} \max_{(\theta,\varphi) \in \DD_{3\epsilon,p}^{{ Q },{ T }}(\rho_0)} \left((\theta,\varphi)^\top I(\rho_0,M^\star) (\theta,\varphi)\right)^{-1}\\
&=  \frac{1}{2} \sup_{\rho_0 \in \cS_{1/2}}   \max_{(\theta,\varphi') \in \DD_{3\epsilon,p}^{{ Q },{ T }}(\rho_0)} \left((\theta,\varphi')^\top I(\rho_0,M^\star) (\theta,\varphi')\right)^{-1}\\
&=  \frac{1}{2}\sup_{\rho_0 \in \cS_{1/2}}   \max_{(\theta,\varphi) \in \DD_{3\epsilon,p}^{{ Q' },{ T }}(\rho_0)} \left((\theta,\varphi)^\top I(\rho_0,M^\star)' (\theta,\varphi)\right)^{-1}\\
% &\geq \frac{1}{2} \sup_{\rho_0 \in \cS_{1/2}} \max_{(\theta,\varphi) \in \DD_{3\epsilon,p}^{{ Q' },{ T }}(\rho_0)} \left((\theta,\varphi)^\top I(\rho_0,M^\star)' (\theta,\varphi)\right)^{-1}\\
&=  \frac{1}{2} \sup_{\rho_0 \in \cS_{1/2}} \max_{\norm{\theta}_p = 3\epsilon} \left((\theta,\varphi)^\top I(\rho_0,M^\star)' (\theta,\varphi)\right)^{-1}\\
&= \frac{1}{2} \sup_{\rho_0 \in \cS_{1/2}} \max_{\norm{\alpha}_q \leq 1/(3\epsilon)}  \alpha^\top {I(\rho_0,M^\star)'}^{-1}_{AA} \alpha\\
&= \frac{1}{2} \sup_{\rho_0 \in \cS_{1/2}} \max_{\norm{\alpha}_q \leq 1/(3\epsilon)}  \alpha^\top {I(\rho_0,M^\star)}^{-1}_{AA} \alpha \\
&= \frac{1}{2} \sup_{\rho_0 \in \cS_{1/2}} \max_{\norm{\theta}_p = 3\epsilon} \left((\theta,\varphi)^\top I(\rho_0,M^\star) (\theta,\varphi)\right)^{-1}\\
&\geq \frac{1}{2}  \sup_{\rho_0 \in \cS_{1/2}} \inf_{M \in \cM} \max_{\norm{\theta}_p = 3\epsilon} \left((\theta,\varphi)^\top I(\rho_0,M) (\theta,\varphi)\right)^{-1}\\
&\geq \frac{1}{4}  \sup_{\rho_0 \in \cS^\circ} \inf_{M \in \cM} \max_{\norm{\theta}_p = 3\epsilon} \left((\theta,\varphi)^\top I(\rho_0,M) (\theta,\varphi)\right)^{-1}. 
\end{align}
\allowdisplaybreaks
Note that in the derivation above, $I'$, $Q'$ and $\varphi'$ are functions of $\rho_0$ (and $M^\star$) are defined when fixing $C_1$ in the basis transformation to be $C_1^\star =  - (I(\rho_0,M^\star)_{BB})^{-1} I(\rho_0,M^\star)_{BA}$.

\vspace{0.08in}
\noindent{\textbf{Step~4: Picking $M^\star$ using mixture with Haar random measurement to show $\eta(M^\star) \geq \eta^\ob > 0$.}}
\vspace{0.08in}

Finally, we show we can find some choice of $M^\star(\rho_0)$ such that $\eta$ is lower bounded by a positive constant $\eta^\ob$, independent of $\epsilon$. In particular, we should be careful in situations where $(I(\rho_0,M^\star)_{BB})^{-1}$ is too large which could make $\norm{Q'_a(\rho_0)}_\infty$ unbounded. This may happen when $I(\rho_0,M^\star)$ is almost singular. To avoid this behavior, we mix $M^\star$ with a Haar random measurement with generally behaves benignly in estimating all parameters including the nuisance parameters.  
For each $\rho_0\in \cS_{1/2}$ and $\epsilon$
such that $\DD_{3\epsilon,p}^{{ Q },{ T }}(\rho_0)$ is nonempty, choose $\widetilde{M}^\star(\rho_0)$ satisfying 
\begin{equation}
      \inf_{M\in\cM} \max_{(\theta,\varphi) \in \DD_{3\epsilon,p}^{{ Q },{ T }}(\rho_0)} \left((\theta,\varphi)^\top I(\rho_0,M) (\theta,\varphi)\right)^{-1} \geq \frac{2}{3} \max_{(\theta,\varphi) \in \DD_{3\epsilon,p}^{{ Q },{ T }}(\rho_0)} \left((\theta,\varphi)^\top I(\rho_0,\widetilde{M}^\star) (\theta,\varphi)\right)^{-1},
\end{equation}
Let
$M^{\rm h}=\{d|v\rangle\langle v| \mathrm{d}v\}_v$ be the Haar
random measurement, where $\mathrm{d}v$ is the Haar measure
on pure states, and choose
\begin{equation}
\label{eq:m-star}
M^\star(\rho_0)
=
\frac34 \widetilde{M}^\star \oplus\frac14 M^{\rm h},
\end{equation}
using the convex combination of measurement in \eqref{eq:convex-meas}. 
Since
\begin{equation}
I(\rho_0,M^\star)
=
\frac34 I(\rho_0,\widetilde{M}^\star)+\frac14 I(\rho_0,M^{\rm h})
\succeq
\frac34 I(\rho_0,\widetilde{M}^\star),    
\end{equation}
$M^\star$ satisfies our requirement \eqref{eq:optimal-rho-m}. 
For the corresponding basis $Q'(\rho_0)=Q+TC_1^\star$, 
% the Haarsecond-moment formula (6.3), tracelessness of $Q'_a(\rho_0)$, and $\langle\psi|\rho_0|\psi\rangle\le 1$ give, 
for every $a\in A$,
\begin{align}
I(\rho_0,M^\star)'_{aa}
&\ge
\frac{1}{4d}\,
\int \mathrm{d} v
\frac{
\langle v|Q'_a(\rho_0)|v \rangle^2
}{
\langle v|\rho_0| v \rangle
}\ge
\frac{1}{4d}\,
\int \mathrm{d} v
\langle v|Q'_a(\rho_0)|v \rangle^2
= 
\frac{\trace(Q'_a(\rho_0)^2)}{4d^2(d+1)}.
\end{align}
On the other hand, the definition of $C_1^\star$ minimizes the upper-left ($I_{AA}$) block of the FIM , so
the Cauchy--Schwarz inequality and $\rho_0\succeq \id/(2d)$ implies
\begin{align}
I(\rho_0,M^\star)'_{aa}
\le I(\rho_0,M^\star)_{aa} = \frac{1}{d^2} \sum_x \frac{\trace(M_x^\star Q_a)^2}{\trace(\rho_0 M_x^\star)}\le \frac{1}{d^2}  \operatorname{tr}(Q_a\rho_0^{-1}Q_a)
\le \frac{2}{d} \operatorname{tr}(Q_a^2).
\end{align}
Combining these inequalities yields the uniform bound
\begin{equation}
\label{eq:temp-19}
\|Q'_a(\rho_0)\|_\infty
\le
\|Q'_a(\rho_0)\|_2
\le
2\sqrt{2d(d+1)}\,
\|Q_a\|_2, 
\end{equation}
where $Q_a$ is a fixed choice of $Q$, independent of $\rho_0$. 
Therefore, 
\begin{equation}
\eta \geq \eta^{\mathrm{ob}} := 
\frac{1}{
12\sqrt{2d(d+1)}
\norm{\left(\norm{Q_1}_2,\ldots,\norm{Q_{\abs{A}}}_2\right)}_q }
>0,
\end{equation}
independently of $\rho_0$ and $\epsilon$.
This lower bound is no larger than
$1/\bigl(6\|(\|Q_a\|_\infty)_{a\in A}\|_q\bigr)$, so for every
positive $\epsilon$ below it $\DD_{3\epsilon,p}^{{ Q },{ T }}(\rho_0)$ is nonempty, as follows by taking $\varphi=0$ and applying the same positivity estimate used in the proof of \Cref{thm:lower_threshold}.
\end{proof}

Note that there is a $\Omega(d)$ gap between $\eta^\ob_{\comp}$, the threshold on $\epsilon$ when $Q$ forms a complete basis, and $\eta^\ob$, the threshold for general $Q$. It implies the value of $\eta^\ob$ might not be tight and we leave the improvement of the threshold in future work. 

\subsection{Few-copy measurements}

We then consider protocols using $(c\geq1)$-copy measurements in $\cM_c$. We now show that for the many-versus-one distinguishing task, if we only care about $\epsilon^2$ term, which is the case when $\epsilon$ is below a certain threshold, then the lower bound for these protocols can only achieve at most an $O(c)$ reduction from the single-copy protocol.

\begin{theorem}[Lower bound for \Cref{prob:distinguish}(\ref{prob:oblivious-distinguish}) and $p$-norm error]\label{thm:distinguish-c-copy}
Using (possibly adaptive) measurement strategy with $c$-copy measurements, consider the many-versus-one distinguishing tasks above with any $\rho_{\theta,\varphi}$ well-defined and $(\theta,\varphi)\in\DD_{3\epsilon,p}^{{ Q },{ T }}(\rho_0)$.  
The sample complexity required to solve this task is
\begin{align}
N =\Omega\left(\frac{\Gamma^{\ob}_{p}(\{O_i\}_{i=1}^m) }{c \epsilon^2}\right), 
\end{align}
for any $\epsilon\leq \min\{\eta^{\ob},\eta^{\ob}_{c} \}$ where
\begin{align}
    \eta^{\ob}_{c} :=  
    \frac{1}{c}\min\left\{ \frac{\eta^\ob \sqrt{2d}}{2} , \frac{3(\eta^\ob  \sqrt{2d})^2}{\sqrt{\Gamma^{\ob}_{p}(\{O_i\}_{i=1}^m)}} \right\}, 
\end{align}
and $p\in[1,\infty]$. In the special case when $Q$ forms a complete basis, it is sufficient to require $\epsilon \leq \min\{\eta^\ob_\comp,\eta^\ob_{c,\comp}\}$ where 
\begin{equation}
    \eta^{\ob}_{c,\comp} := \frac{1}{c}\min\left\{ \frac{\eta^\ob_\comp \sqrt{2}}{2} , \frac{6(\eta^\ob_\comp)^2}{\sqrt{\Gamma^{\ob}_{p}(\{O_i\}_{i=1}^m)}} \right\}.  
\end{equation}
\end{theorem}

\begin{proof}
Our proof will be divided into four steps. In the first step, we show how to upper bound the $\chi^2$ distance by two terms, one contributes dominantly for small $\epsilon$ and the other represents the higher order terms that we will ignore for small $\epsilon$. In the second step, introduce the (near-)optimal choices of $\rho_0$ and $\pi$ the distribution of $(\theta,\varphi)$, which helps us put an upper bound on $\delta_{\cM_c}(O)$ whose inverse provides a sample complexity lower bound. We also define $\Delta_1$ and $\Delta_2$ which bound $\delta_{\cM_c}(O)$ through $\sqrt{\delta_{\cM_c}(O)}\leq \Delta_1 + \Delta_2$. In the third step, we bound $\Delta_1$ by showing it scales as $\Theta(c\epsilon/\sqrt{\Gamma^\ob})$ for $\epsilon < \eta^\ob$. In the fourth step, we find another threshold below which $\Delta_2$ can be ignored.    

\vspace{0.08in}
\noindent{\textbf{Step~1: Bounding the $\chi^2$-distance.}}
\vspace{0.08in}

Note that the sample complexity bound from \eqref{eq:minimax_learning-0} is now $\Omega(c/\delta_{\cM_c}(O))$ where
\begin{align}
\delta_{\cM_c}(O) =  \inf_{\substack{\rho_0 \in \cS^\circ\\ (Q,T) \in \BB(O)}}  \min_{\pi\in\cD_{3\epsilon,p}^{{ Q },{ T }}(\rho_0)}  \sup_{M\in\cM_c}  \bE_{(\theta,\varphi)\sim\pi}
\chi_{M}^2\left(\rho_{\theta,\varphi}^{\otimes c}\|\rho_0^{\otimes c}\right).
\end{align}
We denote the $c$-copy POVM as $M=\{M_s\}_s$, and the classical distributions
\begin{align}
q_s=\tr(M_s \rho_0^{\otimes c}),\qquad 
p_s(\theta,\varphi)=\tr\left(M_s \rho_{\theta,\varphi}^{\otimes c}\right),
\end{align}
the $\chi^2$ divergence induced by $M$ is
\begin{align}
\chi^2_M(\rho_{\theta,\varphi}^{\otimes c}\|\rho_0^{\otimes c})
=\sum_s \frac{\big(p_s(\theta,\varphi)-q_s\big)^2}{q_s}.
\end{align}
We expand $p_s(\theta,\varphi)$ around $\rho_0$. Using the notation 
\begin{align}
\rho_{\theta,\varphi}
=\rho_0 +\frac{(\theta,\varphi)\cdot(\Vec{Q},\Vec{T})}{d},
\end{align}
where $\Vec{Q}=(Q_1,...,Q_{\abs{A}})$ and $\Vec{T}=(T_1,...,T_{\abs{B}})$, the multinomial expansion over subsets $S\subseteq[c]$ is 
\begin{align}
\rho_{\theta,\varphi}^{\otimes c}
=\sum_{S\subseteq[c]}\left(\frac{(\theta,\varphi)\cdot(\Vec{Q},\Vec{T})}{d}\right)^{(S)}
\otimes \rho_0^{([c]\backslash S)},
\end{align}
where $(\cdot)^{(S)}$ denotes placing the operator on the tensor factors indexed by $S$. 
Hence
\begin{align}
p_s(\theta,\varphi)-q_s
=\sum_{\emptyset\neq S\subseteq[c]} 
\tr\left(M_s\left((\theta,\varphi)\cdot(\Vec{Q},\Vec{T})\right)^{(S)}  \otimes \rho_0^{([c]\backslash S)}\right)\frac{1}{d^{|S|}},
\end{align}
Therefore,
\begin{align}
\chi^2_M(\rho_{\theta,\varphi}^{\otimes c}\|\rho_0^{\otimes c})
&=\sum_s \frac{1}{\tr(M_s \rho_0^{\otimes c})}\left(\sum_{\emptyset\neq S\subseteq[c]}\frac{\tr\left(M_s\left((\theta,\varphi)\cdot(\Vec{Q},\Vec{T})\right)^{(S)}  \otimes \rho_0^{([c]\backslash S)}\right)}{d^{|S|}}\right)^{2}. 
\end{align} 
Let $(\bar{\theta},\bar{\varphi}) := (\theta,\varphi)/d$. 
\allowdisplaybreaks
\begin{align}
& \chi^2_M(\rho_{\theta,\varphi}^{\otimes c}\|\rho_0^{\otimes c})\\
= & \sum_s \frac{1}{\tr(M_s \rho_0^{\otimes c})}\left(
\sum_{\emptyset\neq S\subseteq[c]}\frac{\tr\left(M_s\left((\theta,\varphi)\cdot(\Vec{Q},\Vec{T})\right)^{(S)} \rho_0^{([c]\backslash S)}\right)}{d^{|S|}}
\right)^{2} \\
\le& \left(\sum_{\emptyset \ne S \subseteq [c]} \frac{1}{w_S} \sum_s \frac{\tr\left(M_s\left((\theta,\varphi)\cdot(\Vec{Q},\Vec{T})\right)^{(S)} \rho_0^{([c]\backslash S)}\right)^2}{\tr(M_s \rho_0^{\otimes c})d^{2|S|}}\right)\left(\sum_{\emptyset \ne S \subseteq [c]} w_S\right)\\
\leq&\left(
\sum_{\emptyset\neq S\subseteq[c]}\sqrt{ \sum_s \frac{\tr\left(M_s\left((\bar{\theta},\bar{\varphi})\cdot(\Vec{Q},\Vec{T})\right)^{(S)} \rho_0^{([c]\backslash S)}\right)^{2}}{\tr(M_s\rho_0^{\otimes c})}}\right)^{2},
\end{align}
where the third step used Cauchy-Schwarz and the last step optimizes over all choices of $w_S$. 

\vspace{0.08in}
\noindent{\textbf{Step~2: Bounding $\delta_{\cM_c}(O)$ using $\Delta_1$ and $\Delta_2$.}}
\vspace{0.08in}

Let $\rho^\star_0 \in \cS_{1/2}$ be a nearly optimal choice of $\rho^\star_0$ such that 
\begin{equation}
\label{eq:temp-14}
    \inf_{M\in\cM}  \max_{\norm{\alpha}_q \leq 1} \alpha^\top (I(\rho_0^\star,M)^{-1}_{AA} \alpha \geq \frac{1}{2} \sup_{\rho_0 \in \cS^{\circ}} \inf_{M\in\cM}  \max_{\norm{\alpha}_q \leq 1} \alpha^\top (I(\rho_0,M)^{-1}_{AA} \alpha = \frac{1}{2} \Gamma^\ob_{p}(\{O_i\}_{i=1}^m). 
\end{equation}
Note that \eqref{eq:s-1/2-equiv} guarantees the existence of such $\rho_0 \in \cS_{1/2}$. 
Fix $Q$ and $T$ in all discussion below. Let $\pi^\star$ be the optimal probability distribution in ${\cD_{3\epsilon,p}^{{ Q },{ T }}(\rho_0^\star)}$ satisfying 
\begin{align}
    \min_{\pi\in\cD_{3\epsilon,p}^{{ Q },{ T }}(\rho_0^\star)} \sup_{M \in \cMh} \bE_{\pi}  (\theta,\varphi)^\top I(\rho_0^\star,M) (\theta,\varphi) &=  \sup_{M \in \cMh} \bE_{\pi^\star}  (\theta,\varphi)^\top I(\rho_0^\star,M) (\theta,\varphi) 
    \\ &\in \left[ \frac{3}{4},1 \right] \times \sup_{M \in \cM} \bE_{\pi^\star}  (\theta,\varphi)^\top I(\rho_0^\star,M) (\theta,\varphi), 
\end{align}
where 
\begin{equation}
    \cM^{\rm h} = \left\{ M \in \cM: M = \frac{3}{4}\widetilde{M} \oplus \frac{1}{4}M^{\rm h} \text{ for some }\widetilde{M} \in \cM \right\},
\end{equation}
and $M^{\rm h}$ is the Haar random measurement. Note that we pick $\pi^\star$ to be optimal when the measurement is restricted to $\cM^{\rm h}$ for later use in Step 4. In Step 3 only, picking $\pi^\star$ to be optimal when the measurement is arbitrary is sufficient. 

We will use $\bE_{\pi}$ as an abbreviation for $\bE_{(\theta,\varphi) \sim \pi}$. 
Fixing $\rho_0$ to be $\rho_0^\star$ and $\pi$ to be $\pi^\star$ allows us to bound $\delta_{\cM_c}$ properly without having $\inf$ or $\min$ in the expression. 
\begin{align}
\sqrt{\delta_{\cM_c}(O)}
&\leq \sqrt{\sup_{M\in\cM_c}  
  \bE_{\pi^\star} \chi_{M}^2\left(\rho_{\theta,\varphi}^{\star \otimes c}\|\rho_0^{\star \otimes c}\right)} =  \sup_{M\in\cM_c}  
\sqrt{  \bE_{\pi^\star} \chi_{M}^2\left(\rho_{\theta,\varphi}^{\star \otimes c}\|\rho_0^{\star \otimes c}\right)}\\
&\leq   \sup_{M\in\cM_c}
\sum_{\emptyset\neq S\subseteq[c]}\sqrt{  \bE_{\pi^\star}  \sum_s \frac{\tr\left(M_s\left((\bar{\theta},\bar{\varphi})\cdot(\Vec{Q},\Vec{T})\right)^{(S)} \rho_0^{\star  ([c]\backslash S)}\right)^{2}}{\tr(M_s\rho_0^{\star \otimes c})}} \leq \Delta_1 + \Delta_2, 
\end{align}
where the last line following from the Cauchy-Schwarz inequality, and 
\begin{align}
\Delta_1 := 
\sum_{\abs{S} = 1}\sqrt{  \sup_{M\in\cM_c} \bE_{\pi^\star}  \sum_s \frac{\tr\left(M_s\left((\bar{\theta},\bar{\varphi})\cdot(\Vec{Q},\Vec{T})\right)^{(S)} \rho_0^{\star ([c]\backslash S)}\right)^{2}}{\tr(M_s\rho_0^{\star \otimes c})}},
\end{align}
and 
\begin{align}
\Delta_2 :=
\sum_{\abs{S} > 1} \sqrt{  \sup_{M\in\cM_c}  \bE_{\pi^\star}  \sum_s \frac{\tr\left(M_s\left((\bar{\theta},\bar{\varphi})\cdot(\Vec{Q},\Vec{T})\right)^{(S)} \rho_0^{\star ([c]\backslash S)}\right)^{2}}{\tr(M_s\rho_0^{\star \otimes c})}}. 
\end{align}

\vspace{0.08in}
\noindent{\textbf{Step~3: Bounding $\Delta_1$ when $\epsilon \leq \eta^\ob$.}}
\vspace{0.08in}

We first prove the lowest-order term is bounded as expected. 
Define, from any $c$-copy POVM element $M_s$, a single-copy operator
\begin{align}
\label{eq:def-G-s}
G_s:=\tr_{[c]\setminus\{1\}}
\big( (\id \otimes \rho_0^{(c-1)}) M_s \big). 
\end{align}
We can verify that $\{G_s\}_s$ is a valid single-copy measurement.
\begin{itemize}
    \item Positivity: $G_s\ge 0$ since $G_s^\dagger = \tr_{[c]\setminus\{1\}} \big( M_s (\id \otimes \rho_0^{(c-1)})  \big) = \tr_{[c]\setminus\{1\}} \big( (\id \otimes \rho_0^{(c-1)})   M_s \big) = G_s$ by cyclicity and $\tr(A G_s)=\tr(M_s (A\otimes \rho_0^{\otimes(c-1)}))$ for any operator $A$, implying $G_s$ is positive semidefinite. 
    \item Completeness: $\sum_s G_s= \tr_{[c]\setminus\{1\}}(\id\otimes \rho_0^{\otimes(c-1)})=\id$ (because $\tr(\rho_0)=1$). So $G=\{G_s\}_s$ is a valid single-copy POVM.
    \item Crucial identities: For every $x$ and every observable $A$ acting on the first copy,
    \begin{align}
    \frac{\tr\left(M_s(A^{(1)}\otimes \rho_0^{\otimes(c-1)})\right)^2}{\tr(M_s\rho_0^{\otimes c})}=\frac{\tr(AG_s)^2}{\tr(G_s\rho_0)}.
    \end{align}
\end{itemize}
Now, we note that $A=(\bar{\theta},\bar{\varphi})\cdot(\Vec{Q},\Vec{T})$, the $|S|=1$ functional equals the single-copy functional for $G$ with base state $\rho_0$:
\begin{align}
\sum_s \frac{\tr\left(M_s\left((\bar{\theta},\bar{\varphi})\cdot(\Vec{Q},\Vec{T})\right)^{(1)} \rho_0^{ (c-1)}\right)^{2}}{\tr(M_s\rho_0^{ \otimes c})} = \sum_s \frac{\tr\left(G_s\left((\bar{\theta},\bar{\varphi})\cdot(\Vec{Q},\Vec{T})\right)\right)^{2}}{\tr(G_s\rho_0)} = (\theta,\varphi)^\top I(\rho_0,G) (\theta,\varphi).\label{eq:c-to-1}
\end{align}
Here $^{(1)}$ means operator acting on the first qudit. 
The mapping sends any $M\in\mathcal M_c$ to some $G\in\cM$ that is a single-copy POVM. Furthermore, any single-copy POVM $G$ can be written in the form of \eqref{eq:def-G-s} when taking $M_s = G_s \otimes \id^{(c-1)}$. Therefore,
\begin{align}
\sup_{M\in\cM_c} \bE_{\pi^\star}  \sum_s \frac{\tr\left(M_s\left((\bar{\theta},\bar{\varphi})\cdot(\Vec{Q},\Vec{T})\right)^{(1)} \rho_0^{\star \otimes(c-1)}\right)^{2}}{\tr(M_s\rho_0^{\star \otimes c})}
= \sup_{M \in \cM} \bE_{\pi^\star}  (\theta,\varphi)^\top I(\rho_0^\star ,M) (\theta,\varphi) = \left(\frac{\Delta_1}{c}\right)^2.
\end{align}
Note that 
\begin{align}
    \sup_{M \in \cM} \bE_{\pi^\star}  (\theta,\varphi)^\top I(\rho_0^\star,M) (\theta,\varphi) 
    &\in \left[1, \frac{4}{3}\right] \times \min_{\pi\in\cD_{3\epsilon,p}^{{ Q },{ T }}(\rho_0^\star)} \sup_{M \in \cM} \bE_{\pi}  (\theta,\varphi)^\top I(\rho_0^\star,M) (\theta,\varphi) \\
    &= \left[1, \frac{4}{3}\right] \times \sup_{M \in \cM}  \min_{\pi\in\cD_{3\epsilon,p}^{{ Q },{ T }}(\rho_0^\star)}  \bE_{\pi}  (\theta,\varphi)^\top I(\rho_0^\star,M) (\theta,\varphi)\\
    &= \left[1, \frac{4}{3}\right] \times \sup_{M \in \cM}  \min_{(\theta,\varphi)\in\DD_{3\epsilon,p}^{{ Q },{ T }}(\rho_0^\star)}  (\theta,\varphi)^\top I(\rho_0^\star,M) (\theta,\varphi)
\end{align}
From \eqref{eq:optimal-rho-m}, let $M^\star = M^\star(\rho_0^\star)$ as defined in the proof of \Cref{thm:distinguish_thres_general} (\eqref{eq:m-star}), following from the proof of \Cref{thm:distinguish_thres_general}, we have 
\begin{align}
    \sup_{M \in \cM}  \min_{(\theta,\varphi)\in\DD_{3\epsilon,p}^{{ Q },{ T }}(\rho_0^\star)}  (\theta,\varphi)^\top I(\rho_0^\star,M) (\theta,\varphi) 
    & \leq  2 \min_{(\theta,\varphi)\in\DD_{3\epsilon,p}^{{ Q },{ T }}(\rho_0^\star)}  (\theta,\varphi)^\top I(\rho_0^\star,M^\star) (\theta,\varphi)\\
    & = 2  \min_{\norm{\theta}_p = 3\epsilon}  (\theta,\varphi)^\top I(\rho_0^\star,M^\star) (\theta,\varphi)\\
    &\leq 2 \sup_{M \in \cM} \min_{\norm{\theta}_p = 3\epsilon}  (\theta,\varphi)^\top I(\rho_0^\star,M) (\theta,\varphi) \\
    &\leq 4 \inf_{\rho_0 \in \cS^\circ} \sup_{M \in \cM} \min_{\norm{\theta}_p = 3\epsilon}  (\theta,\varphi)^\top I(\rho_0,M) (\theta,\varphi), \label{eq:temp-15}
\end{align}
when $\epsilon < \eta^\ob$, where we use \eqref{eq:temp-14} in the last line. On the other hand, by definition, 
\begin{align}
    \sup_{M \in \cM}  \min_{(\theta,\varphi)\in\DD_{3\epsilon,p}^{{ Q },{ T }}(\rho_0^\star)}  (\theta,\varphi)^\top I(\rho_0^\star,M) (\theta,\varphi) &\geq \sup_{M \in \cM}  \min_{\norm{\theta}_p = 3\epsilon}  (\theta,\varphi)^\top I(\rho_0^\star,M) (\theta,\varphi)\\
    &\geq \inf_{\rho_0 \in \cS^\circ} \sup_{M \in \cM}  \min_{\norm{\theta}_p = 3\epsilon} (\theta,\varphi)^\top I(\rho_0,M) (\theta,\varphi). \label{eq:temp-16}
\end{align}
Combining \eqref{eq:temp-15} and \eqref{eq:temp-16}, we have 
\begin{equation}
\frac{(3\epsilon)^2}{\Gamma^{\ob}_{p}(\{O_i\}_{i=1}^m)} \leq \left(\frac{\Delta_1}{c}\right)^2 \leq \frac{16 (3\epsilon)^2}{3 \Gamma^{\ob}_{p}(\{O_i\}_{i=1}^m)}, 
\end{equation}
when $\epsilon < \eta^\ob$. 

\vspace{0.08in}
\noindent{\textbf{Step~4: Bounding $\Delta_2$ and refinement of threshold.}}
\vspace{0.08in}

It remains to derive the threshold on $\epsilon$ below which we can ignore the contributions from higher-order terms. We assume $\epsilon < \eta^\ob$ already below and tighten it further. Since 
\begin{align}
\sqrt{\delta_{\cM_c}(O)} \leq \Delta_1 + \Delta_2, 
\end{align}
it is sufficient to show 
\begin{equation}
\Delta_2 \leq c \sqrt{\frac{(3\epsilon)^2}{\Gamma^{\ob}_{p}(\{O_i\}_{i=1}^m)}}. 
\end{equation}
Fix a subset $S$ with $|S|=k$. Set
\begin{align}
N_s:=((\rho_0)^{1/2})^{\otimes c} M_s ((\rho_0)^{1/2})^{\otimes c} \qquad(\Rightarrow\ \sum_s N_s=\rho_0^{\otimes c},  \tr(N_s)=\tr(M_s\rho_0^{\otimes c})),
\end{align}
and define
\begin{align}
B_S:=\bigotimes_{i\in S}\big(\rho_0^{ -1/2}\left((\bar{\theta},\bar{\varphi})\cdot(\Vec{Q},\Vec{T})\right)^{(i)}\rho_0^{ -1/2}\big)\otimes\bigotimes_{i\notin S} \id.
\end{align}
Then
\begin{align}
\tr\Big(M_s \left((\bar{\theta},\bar{\varphi})\cdot(\Vec{Q},\Vec{T})\right)^{(S)}\otimes\rho_0^{ ([c]\backslash S )}\Big)=\tr\big(N_sB_S\big),\qquad\tr(M_s\rho_0^{ \otimes c})=\tr(N_s).
\end{align}
By Cauchy–Schwarz in the Hilbert–Schmidt inner product,
\begin{align}
\sum_s \frac{\tr(N_s B_S)^2}{\tr(N_s)}\le \tr\left(B_S\left(\sum_s N_s\right)B_S\right)=\tr\big(B_S\rho_0^{\otimes c}B_S\big).
\end{align}
Evaluating the right-hand side factorizes over tensor slots, giving
\begin{align}
\tr\big(B_S\rho_0^{\otimes c}B_S\big)=\Big[\tr\left(\left((\bar{\theta},\bar{\varphi})\cdot(\Vec{Q},\Vec{T})\right)\rho_0^{ -1}\left((\bar{\theta},\bar{\varphi})\cdot(\Vec{Q},\Vec{T})\right)\right)\Big]^k.
\end{align}
Therefore, we have 
\begin{align}
 \Delta_2  & =  
 \sum_{\abs{S} > 1} \sqrt{   \sup_{M\in\cM_c} \bE_{\pi^\star}  \sum_s \frac{\tr\left(M_s\left((\bar{\theta},\bar{\varphi})\cdot(\Vec{Q},\Vec{T})\right)^{(S)} \rho_0^{\star( [c]\backslash S)}\right)^{2}}{\tr(M_s\rho_0^{\star \otimes c})}}\\
&\leq  \sum_{k=2}^c \binom{c}{k}\left( \bE_{\pi^\star}\bigg[ \Big[\tr\left(\left((\bar{\theta},\bar{\varphi})\cdot(\Vec{Q},\Vec{T})\right)\rho_0^{\star -1}\left((\bar{\theta},\bar{\varphi})\cdot(\Vec{Q},\Vec{T})\right)\right)\Big]^{k}\bigg] \right)^{1/2}.
\end{align}

The main difficulty in bounding $\Delta_2$ is to address the terms with nuisance parameters as the norm of $\varphi$ is not necessarily small. To this end, we need to refine the definition of $M^\star$. Specifically, we define $M^\star_\kappa \in \cM^{\rm h}$ that satisfies
\begin{equation}
    \min_{(\theta,\varphi) \in \DD_{3\epsilon,p}^{{ Q },{ T }}(\rho_0^\star)} (\theta,\varphi)^\top I(\rho_0^\star,M^\star_\kappa) (\theta,\varphi) \geq (1-\kappa) \sup_{M \in \cM^{\rm h}} \min_{(\theta,\varphi) \in \DD_{3\epsilon,p}^{{ Q },{ T }}(\rho_0^\star)} (\theta,\varphi)^\top I(\rho_0^\star,M) (\theta,\varphi). 
\end{equation}
We will take the limit of $\kappa \rightarrow 0$ later. Note that here we do not directly take $\kappa = 0$ as the attainability of the supremum may not be guaranteed. 
When $\epsilon < \eta^\ob$, 
\begin{align}
    \sup_{M \in \cMh} \min_{(\theta,\varphi) \in \DD_{3\epsilon,p}^{{ Q },{ T }}(\rho_0^\star)} (\theta,\varphi)^\top I(\rho_0^\star,M) (\theta,\varphi) 
    &= \sup_{M \in \cMh} \min_{(\theta,\varphi) \in \DD_{3\epsilon,p}^{{ Q' },{ T }}(\rho_0^\star)} (\theta,\varphi)^\top I(\rho_0^\star,M)' (\theta,\varphi)\\
    &= \sup_{M \in \cMh} \min_{\norm{\theta}_p = 3\epsilon} (\theta,\varphi)^\top I(\rho_0^\star,M)' (\theta,\varphi) \\
    &\in \left[\frac{3}{4},1\right] \times \sup_{M \in \cM} \min_{\norm{\theta}_p = 3\epsilon} (\theta,\varphi)^\top I(\rho_0^\star,M)' (\theta,\varphi),
\end{align}
and then 
\begin{equation}
\label{eq:temp-20}
    \sup_{M \in \cM^{\rm h}} \min_{(\theta,\varphi) \in \DD_{3\epsilon,p}^{{ Q },{ T }}(\rho_0^\star)} (\theta,\varphi)^\top I(\rho_0^\star,M) (\theta,\varphi) \in \left[\frac{3}{4},2\right] \times \frac{(3\epsilon)^2 }{ \Gamma^{\ob}_p(\{O\}_{i=1}^m)}. 
\end{equation}
We first put a bound on 
\begin{equation}
\label{eq:phi-quad}
    \bE_{\pi^{\star'}} \varphi^\top I(\rho_0^\star,M^\star_\kappa)'_{BB} \varphi \leq \kappa \sup_{M \in \cM^{\rm h}} \min_{(\theta,\varphi) \in \DD_{3\epsilon,p}^{{ Q },{ T }}(\rho_0^\star)} (\theta,\varphi)^\top I(\rho_0^\star,M) (\theta,\varphi), 
\end{equation}
where $I(\rho_0^\star,M^\star_\kappa)'$ is the block-diagonal form of $I(\rho_0^\star,M^\star_\kappa)$ after the basis transformation with $C_1^\star =  - (I(\rho_0^\star,M^\star_\kappa)_{BB})^{-1}I(\rho_0^\star,M^\star_{\kappa})_{BA}$. We will also use notations $Q'$ and $\varphi'$ later to the transformed $Q$ and $\varphi$ that corresponds to the block-diagonalization of $I(\rho_0^\star,M^\star_\kappa)'_{BB}$. Above $\pi^{\star'} \in \cD_{3\epsilon,p}^{Q',T}(\rho_0)$ is defined as 
\begin{equation}
   (\theta,\varphi) \sim \pi^{\star'} ~~\Leftrightarrow~~ (\theta,\varphi + C_1^\star \theta) = (\theta,\varphi') \sim \pi^{\star}. 
\end{equation}
To see \eqref{eq:phi-quad} holds, we first note that when $(\theta,\varphi) \in \DD_{3\epsilon,p}^{Q',T}(\rho_0)$, 
$(\theta,0) \in \DD_{3\epsilon,p}^{Q',T}(\rho_0)$ because $\norm{\theta}_p = 3\epsilon$ and $\epsilon < \eta^\ob$ which implies $\norm{\frac{1}{d}\sum_{a\in A} \theta_a Q'_a}_\infty \leq 1/(2d)$ and $\rho^\star_{\theta,0} \succeq 0$. 
\begin{align}
    \bE_{\pi^{\star'}} \theta^\top I(\rho_0^\star,M^\star_\kappa)'_{AA} \theta 
    &\geq \min_{(\theta,\varphi) \in \DD_{3\epsilon,p}^{{ Q' },{ T }}(\rho_0^\star)} (\theta,\varphi)^\top I(\rho_0^\star,M^\star_\kappa)' (\theta,\varphi) \\
    &= \min_{(\theta,\varphi) \in \DD_{3\epsilon,p}^{{ Q },{ T }}(\rho_0^\star)} (\theta,\varphi)^\top I(\rho_0^\star,M^\star_\kappa) (\theta,\varphi)\\
    &\geq (1-\kappa) \sup_{M \in \cM^{\rm h}} \min_{(\theta,\varphi) \in \DD_{3\epsilon,p}^{{ Q },{ T }}(\rho_0^\star)} (\theta,\varphi)^\top I(\rho_0^\star,M) (\theta,\varphi),\label{eq:temp-17}
\end{align}
and that 
\begin{align}
    \bE_{\pi^{\star'}} (\theta,\varphi)^\top I(\rho_0^\star,M^\star_\kappa)' (\theta,\varphi) 
    &= \bE_{\pi^{\star}} (\theta,\varphi)^\top I(\rho_0^\star,M^\star_\kappa) (\theta,\varphi)\\ 
    &\leq \sup_{M \in \cM^{\rm h}} \bE_{\pi^{\star}}  (\theta,\varphi)^\top I(\rho_0^\star,M) (\theta,\varphi)\\
    &= \sup_{M \in \cM^{\rm h}} \min_{(\theta,\varphi) \in \DD_{3\epsilon,p}^{{ Q },{ T }}(\rho_0^\star)} (\theta,\varphi)^\top I(\rho_0^\star,M) (\theta,\varphi),\label{eq:temp-18}
\end{align}
where the minimax theorem is applied in the last line. \eqref{eq:phi-quad} is then implied by \eqref{eq:temp-17} and \eqref{eq:temp-18} and the fact that $I(\rho_0^\star,M^\star_\kappa)'$ is block-diagonal. 

$\bE_{\pi^{\star'}} \varphi^\top I(\rho_0^\star,M^\star_\kappa)'_{BB} \varphi$ is closely related to $\bE_{\pi^\star}[ [\tr(((\bar{\theta},\bar{\varphi})\cdot(\Vec{Q},\Vec{T}))\rho_0^{ -1}((\bar{\theta},\bar{\varphi})\cdot(\Vec{Q},\Vec{T})))]^{k}] $ appearing in the upper bound on $\Delta_2$ through its Haar random part. Specifically, using $\rho_0 \succeq \id/(2d)$, 
\begin{equation}
    \bE_{\pi^{\star'}}\big[\tr ( (\bar{\varphi}\cdot \Vec{T}) \rho_0^{ -1} (\bar{\varphi}\cdot \Vec{T}) )\big] \leq \frac{2}{d} \bE_{\pi^{\star'}}\big[\tr ((\varphi\cdot \Vec{T})^2)\big],
\end{equation}
and 
\begin{align}
    \bE_{\pi^{\star'}} \varphi^\top I(\rho_0^\star,M^\star_\kappa)'_{BB} \varphi &\geq \frac{1}{4}\bE_{\pi^{\star'}} \varphi^\top I(\rho_0^\star,M^{\rm h})'_{BB} \varphi \\
    &= \frac{1}{4d} \bE_{\pi^{\star'}} \int {\rm d} v \frac{\bra{v} (\varphi\cdot \Vec{T}) \ket{v}^2}{\bra{v} \rho_0 \ket{v}}\\
    &\geq \frac{1}{4d^2(d+1)} \bE_{\pi^{\star'}} \trace( (\varphi\cdot \Vec{T})^2 ). 
\end{align}
Let $a(\theta):=\frac{\tr((\theta\cdot \Vec{{Q'}})\rho_0^{\star -1}(\theta\cdot \Vec{{Q'}}))}{d^2}$ and $b(\varphi) =\frac{2\tr((\varphi\cdot \Vec{T})^2)}{d}$. $a(\theta)$ is upper bounded by  
\begin{align}
a(\theta) 
&\leq \frac{\tr((\theta\cdot \Vec{{Q'}})\rho_0^{\star -1}(\theta\cdot \Vec{{Q'}}))}{d^2} \leq \frac{2\tr((\theta\cdot \Vec{{Q'}})(\theta\cdot \Vec{{Q'}}))}{d} \\
&\leq \frac{2}{d} \bigg( \sum_{a \in A} \abs{\theta_a} \norm{Q'_a}_2 \bigg)^2 \leq \frac{2}{d} \norm{\theta}_p^2 
\norm{\left(\norm{Q_1'}_2,\ldots,\norm{Q_{\abs{A}}'}_2\right)}_q^2 
\\
&\leq 16(d+1) \norm{\theta}_p^2 
\norm{\left(\norm{Q_1}_2,\ldots,\norm{Q_{\abs{A}}}_2\right)}_q^2 = \frac{1}{2d} (\epsilon/\eta^\ob)^2 ,
\end{align}
where we use \eqref{eq:temp-19} in the last line. $b(\varphi)$ is upper bounded by 
\begin{equation}
    \bE_{\pi^{\star'}}b(\varphi) \leq 8d(d+1) \bE_{\pi^{\star'}} \varphi^\top I(\rho_0^\star,M^\star_\kappa)'_{BB} \varphi \leq 16 d (d+1) \kappa  \frac{(3\epsilon)^2 }{ \Gamma^{\ob}_p(\{O\}_{i=1}^m)},
\end{equation}
where we use \eqref{eq:phi-quad} and \eqref{eq:temp-20} in the last inequality. We also define $r(\theta,\varphi) = \tr(((\bar{\theta},\bar{\varphi})\cdot(\Vec{Q}',\Vec{T}))\rho_0^{\star -1}((\bar{\theta},\bar{\varphi})\cdot(\Vec{Q}',\Vec{T})))$. 
Then 
\begin{align}
    \Delta_2 \leq \sum_{k=2}^c \binom{c}{k}  (\bE_{\pi^{\star' }}[ r(\theta,\varphi)^k ])^{1/2}.
\end{align}
In particular, when $k = 1$, 
\begin{align}
    &\quad\, \bE_{\pi^\star}\Big[\tr\left(\left((\bar{\theta},\bar{\varphi})\cdot(\Vec{Q},\Vec{T})\right)\rho_0^{\star -1}\left((\bar{\theta},\bar{\varphi})\cdot(\Vec{Q},\Vec{T})\right)\right)\Big] \\
    &= 
    \bE_{\pi^{\star'}}\Big[\tr\left(\left((\bar{\theta},\bar{\varphi})\cdot(\Vec{Q}',\Vec{T})\right)\rho_0^{\star -1}\left((\bar{\theta},\bar{\varphi})\cdot(\Vec{Q}',\Vec{T})\right)\right)\Big] = \bE_{\pi^{\star'}} r(\theta,\varphi)
    \\ 
    &\leq \bE_{\pi^{\star'}} a(\theta) + \bE_{\pi^{\star'}} b(\varphi) + 2 \Re[  \bE_{\pi^{\star'}}[\tr(( \bar{\theta}\cdot \Vec{Q}')\rho_0^{\star -1}( \bar{\varphi}\cdot \Vec{T}))] ]\\
    &\leq \bE_{\pi^{\star'}} a(\theta) + \bE_{\pi^{\star'}} b(\varphi) + 2 \bE_{\pi^{\star'}} \sqrt{ a(\theta) b(\varphi)}. 
\end{align}
As $\bE_{\pi^{'\star}}b(\varphi) = O(\kappa)$ and $a(\theta)$ is bounded, and the bound above applies to any $\kappa$ close to $0$, we have, as $\kappa \rightarrow 0$, 
\begin{equation}
    (\bE_{\pi^{\star'}} r(\theta,\varphi) - \bE_{\pi^{\star'}} a(\theta)) \rightarrow 0, ~\Rightarrow~ \bE_{\pi^{\star'}} r(\theta,\varphi) \leq \frac{1}{2d} (\epsilon/\eta^\ob)^2.
\end{equation}
In general when $k > 1$, 
\begin{align}
    \bE_{\pi^{\star'}} r(\theta,\varphi)^k -  \bE_{\pi^{\star'}} a(\theta)^k \leq \bE_{\pi^{\star'}} (r(\theta,\varphi) - a(\theta)) (r(\theta,\varphi)^{k-1} + r(\theta,\varphi)^{k-2}a(\theta) + \cdots + a(\theta)^{k-1}).
\end{align}
Since $a(\theta)$ is bounded, 
\begin{equation}
    b(\varphi) = 2 d \,\norm{\rho_{\theta,\varphi}^\star - \rho_{\theta,0}^\star}_2^2 \leq 4 d 
\end{equation}
is also bounded, 
and then $r(\theta,\varphi) \leq a(\theta)+b(\varphi)+2\sqrt{a(\theta)b(\varphi)}$ is bounded, we have 
\begin{align}
    \bE_{\pi^{\star'}}[r(\theta,\varphi)^k] - \bE_{\pi^{\star'}}[a(\theta))^k] \rightarrow 0,\quad\text{as~} \kappa \rightarrow 0. 
\end{align}
Since $\Delta_2$ is independent of $\kappa$, we have 
\begin{align}
    \Delta_2 
    &\leq \liminf_{\kappa \rightarrow 0}  \sum_{k=2}^c \binom{c}{k} (\bE_{\pi^{\star' }}[ r(\theta,\varphi)^k ])^{1/2}\\
    &\leq \liminf_{\kappa \rightarrow 0}  \sum_{k=2}^c \binom{c}{k} (\bE_{\pi^{\star' }}[ a(\theta)^k ])^{1/2}\\
    &\leq  \sum_{k=2}^c \binom{c}{k}  \frac{1}{(2d)^{k/2}} (\epsilon/\eta^\ob)^k \\
    &= \left( \bigg(1 + \frac{\epsilon}{\eta^\ob  \sqrt{2d}} \bigg)^c - 1 - c \frac{\epsilon}{\eta^\ob  \sqrt{2d}} \right) \\
    &\leq \frac{1}{2}{\left(\frac{c \epsilon}{\eta^\ob  \sqrt{2d}}\right)^2} \exp\left(\frac{c\epsilon}{\eta^\ob  \sqrt{2d}}\right),
\end{align}
where in the last step we use $(1+x)^c \leq e^{c x}$ and $e^t \leq 1 + t + t^2e^t /2 $ for $x,t \geq 0$. To find a threshold below which 
\begin{equation}
    \frac{1}{2}{\left(\frac{c \epsilon}{\eta^\ob  \sqrt{2d}}\right)^2} \exp\left(\frac{c\epsilon}{\eta^\ob  \sqrt{2d}}\right) \leq c \sqrt{\frac{(3\epsilon)^2}{\Gamma^{\ob}_{p}(\{O_i\}_{i=1}^m)}}, 
\end{equation}
we can impose, e.g. 
\begin{equation}
    \frac{c\epsilon}{\eta^\ob  \sqrt{2d}} \leq 1/2,\quad 
    {\left(\frac{c \epsilon}{\eta^\ob  \sqrt{2d}}\right)^2} \leq c \sqrt{\frac{(3\epsilon)^2}{\Gamma^{\ob}_{p}(\{O_i\}_{i=1}^m)}},
\end{equation}
which means 
\begin{equation}
\label{eq:temp-111}
    \epsilon \leq \frac{1}{c}\min\left\{ \frac{\eta^\ob \sqrt{2d}}{2} , \frac{3(\eta^\ob  \sqrt{2d})^2}{\sqrt{\Gamma^{\ob}_{p}(\{O_i\}_{i=1}^m)}} \right\} =: \eta^{\ob}_{c} . 
\end{equation}

\vspace{0.08in}
\noindent{\textbf{Special case: Complete observables}}
\vspace{0.08in}

Finally, we remark that in the special case where $Q$ forms a complete basis, the threshold can be significantly tightened. In particular, in Step 4, there are no nuisance parameters, and we will need to bound 
\begin{align}
 \Delta_2  & =  
 \sum_{\abs{S} > 1} \sqrt{   \sup_{M\in\cM_c} \bE_{\pi^\star}  \sum_s \frac{\tr\left(M_s\left(\bar{\theta}\cdot \Vec{Q}\right)^{(S)} \rho_0^{\star( [c]\backslash S)}\right)^{2}}{\tr(M_s\rho_0^{\star \otimes c})}}\\
&\leq  \sum_{k=2}^c \binom{c}{k}\left( \bE_{\pi^\star}\bigg[ \Big[\tr\left(\left(\bar{\theta}\cdot \Vec{Q}\right)\rho_0^{\star -1}\left(\bar{\theta}\cdot \Vec{Q}\right)\right)\Big]^{k}\bigg] \right)^{1/2} = \sum_{k=2}^c \binom{c}{k} ( \bE_{\pi^\star} [a(\theta)^k] )^{1/2},
\end{align}
where 
\begin{align}
a(\theta) 
&= \frac{\tr((\theta\cdot \Vec{{Q}})\rho_0^{\star -1}(\theta\cdot \Vec{{Q}}))}{d^2} \leq \frac{2\tr((\theta\cdot \Vec{{Q}})(\theta\cdot \Vec{{Q}}))}{d} \leq \frac{2}{d} \bigg( \sum_{a \in A} \abs{\theta_a} \norm{Q_a}_2 \bigg)^2  \\
&\leq 2 \bigg( \sum_{a \in A} \abs{\theta_a} \norm{Q_a}_\infty \bigg)^2 \leq 2 \norm{\theta}_p^2 
\norm{\left(\norm{Q_1}_\infty,\ldots,\norm{Q_{\abs{A}}}_\infty\right)}_q^2 = \frac{1}{2} (\epsilon/\eta^\ob_\comp)^2. 
\end{align}
As a result, \eqref{eq:temp-111} will be turned into 
\begin{equation}
    \epsilon \leq \frac{1}{c}\min\left\{ \frac{\eta^\ob_\comp \sqrt{2}}{2} , \frac{6(\eta^\ob_\comp)^2}{\sqrt{\Gamma^{\ob}_{p}(\{O_i\}_{i=1}^m)}} \right\} =: \eta^{\ob}_{c,\comp},
\end{equation}
which provides a threshold on $\epsilon$ when $Q$ forms a complete basis.

\end{proof}

\section{Lower bounds for unbiased, kurtosis-bounded estimation}
\label{sec:unbiased}

In this section, we use the CR bound to prove a lower bound on the sample complexity required to solve \Cref{prob:p-estimation}, i.e., estimation of parameters with $p$-norm error.
We obtain bounds that work for the oblivious and the non-oblivious estimation tasks, respectively.
Importantly, since the CR bound applies to only unbiased estimators, the lower bound is restricted to unbiased estimators. 
Besides the unbiasedness, we also need a condition on the tail of the estimator, to convert a large deviation bound to a CR-type MSE bound. The concrete condition we choose is a bounded \emph{kurtosis}. 
\begin{definition}[Kurtosis]
\label{def:bounded-kurtosis}
A zero-mean random variable $ e$ has kurtosis bounded by $\kappa\ge1$, if
\bb
    \E[ e^4] \le \kappa \E[ e^2]^2 <\infty.
\ee
\end{definition}
\noindent In what follows, we call an unbiased estimator $\htheta$ for $\theta$ $\kappa$-kurtosis-bounded, if $(\htheta_i - \theta_i)$ has kurtosis bounded by $\kappa$ for every index $i$.

\subsection{Single-copy measurements}

\begin{theorem}[Lower bound for estimation with $p$-norm error using unbiased kurtosis-bounded estimators]\label{thm:estimation-lower-single}
Using the adaptive measurement strategy with single-copy measurements, the sample complexity required to obtain an unbiased, $\kappa$-kurtosis-bounded estimator for $\theta$ in \Cref{prob:p-estimation} is lower bounded by
\begin{equation}
\label{eq:unbiased-lower}
    N \ge \frac{1-\sqrt{\kappa\delta}}{\epsilon^2}\sup_{\rho_0 \in \cS^\circ} \inf_{M\in\cM}  {\norm{\diag\big( (I(\rho_0,M)^{-1})_{AA}\big)^{1/2}}_p^2}
    = \frac{1-\sqrt{\kappa\delta}}{\epsilon^2}{ \Gamma_{p}(\{O_i\}_{i=1}^m) },
\end{equation}
for any $\epsilon > 0$, $p \in [2,\infty]$, and $\kappa\delta < 1$ where $\delta$ is the failure probability. Here $I(\rho_0,M) = I(\rho_{\theta,\varphi},M)|_{\theta=\varphi=0}$, $\rho_{\theta,\varphi} = \rho_0 + \frac{1}{d}\sum_{a \in A} \theta_a Q_a + \frac{1}{d}\sum_{b \in B} \varphi_b T_b$, and $\diag(\cdot)$ represents a diagonal matrix whose diagonal entries are those of $(\cdot)$. 
When $p \in [1,2)$, we have instead 
\begin{equation}
    N \ge \frac{(1-\sqrt{\kappa\delta})^3}{\kappa \epsilon^2} { \Gamma_{p}(\{O_i\}_{i=1}^m) }. 
\end{equation}
\end{theorem}

\begin{proof}
Fix any $\rho_0\in \cS^\circ$ and any valid choice of $\{Q_a,T_b\}$. We will show that any unbiased estimator with bounded kurtosis that solves Problem 2 also has small MSE. Then we use the CR bound to finish the proof.

Let $\htheta$ be an estimator satisfying the assumptions. Write
\bb
    e_i := \htheta_i - \theta_i, \qquad G := \{\norm{e}_p < \epsilon\},
\ee
so that $\Pr(\neg G) \le \delta$. Since $\htheta$ is unbiased, each $e_i$ has zero mean, and by \Cref{def:bounded-kurtosis}, $\bE[e_i^4] \le \kappa\, \bE[e_i^2]^2$. By the Cauchy--Schwarz inequality,
\bb
    \bE[e_i^2 \1_{\neg G}] \le \bE[e_i^4]^{1/2} \Pr(\neg G)^{1/2} \le \sqrt{\kappa\delta}\, \bE[e_i^2].
\ee
Since $\bE[e_i^2] = \bE[e_i^2 \1_{G}] + \bE[e_i^2 \1_{\neg G}]$ and $\kappa\delta<1$, we obtain
\bb
    \bE[e_i^2] \le \frac{1}{1-\sqrt{\kappa\delta}}\,\bE[e_i^2 \1_{G}], \qquad i = 1,\ldots,m.
\ee
{
\noindent For $p\in[2,\infty)$,
\bb
    \sum_{i=1}^m \bE[e_i^2]^{p/2}
    \le \frac{\sum_{i=1}^m \bE\big[e_i^2\1_G\big]^{p/2}}{(1-\sqrt{\kappa\delta})^{p/2}}
    \le \frac{\sum_{i=1}^m \bE\big[|e_i|^p\1_G\big]}{(1-\sqrt{\kappa\delta})^{p/2}}
    = \frac{\bE\big[\sum_{i=1}^m |e_i|^p\,\1_G\big]}{(1-\sqrt{\kappa\delta})^{p/2}}
    < \frac{\epsilon^p}{(1-\sqrt{\kappa\delta})^{p/2}}.
\ee
The first inequality follows from $\bE[e_i^2] \le \frac{1}{1-\sqrt{\kappa\delta}}\,\bE[e_i^2\1_G]$. The second inequality is Jensen's inequality. The equality is linearity of expectation. The last inequality holds because $\sum_{i}|e_i|^p\1_G<\epsilon^p$. For $p=\infty$, the same chain reads
\bb
    \max_{i} \bE[e_i^2] \le \frac{\max_{i} \bE\big[e_i^2\1_G\big]}{1-\sqrt{\kappa\delta}} \le \frac{\bE\big[\max_{j} e_j^2\, \1_G\big]}{1-\sqrt{\kappa\delta}} < \frac{\epsilon^2}{1-\sqrt{\kappa\delta}},
\ee
where the second inequality holds because $e_i^2 \le \max_j e_j^2$ pointwise for every $i$.
}
Therefore,
\bb
    \label{eq:mse-from-kurtosis}
    \bigg(\sum_{i=1}^m \bE[(\htheta_i-\theta_i)^2]^{p/2}\bigg)^{1/p} = \norm{(\bE[e_i^2])_{i=1}^m}_{p/2}^{1/2} < \frac{\epsilon}{\sqrt{1-\sqrt{\kappa\delta}}},
\ee
for all $p\in [2,\infty]$.

More generally, when $p \geq 1$, H\"older's inequality and the kurtosis
bound give, for each \(i\),
\begin{align}
(1-\sqrt{\kappa\delta})^{3/2}
    \bigl(\mathbb{E}[e_i^2]\bigr)^{3/2}
&\le \bigl(\mathbb{E}[e_i^2\mathbf{1}_G]\bigr)^{3/2} 
= \bigl(\mathbb{E}[e_i^{2/3+4/3}\mathbf{1}_G]\bigr)^{3/2}  \\
&\le \mathbb{E}[|e_i|^{2/3 \cdot 3/2}\mathbf{1}_G]\,
     \bigl(\mathbb{E}[e_i^{4/3 \cdot 3}\mathbf{1}_G]\bigr)^{1/2} \\
&\le \sqrt{\kappa}\,
     \mathbb{E}[|e_i|\mathbf{1}_G]\,\mathbb{E}[e_i^2].
\end{align}
Thus, $(1-\sqrt{\kappa\delta})^{3/2}
    \bigl(\mathbb{E}[e_i^2]\bigr)^{1/2} \leq \sqrt{\kappa}\, \mathbb{E}[|e_i|\mathbf{1}_G]$ . Applying this bound, we have 
\begin{align}
\left(\sum_{i=1}^m
     \bigl(\mathbb{E}[e_i^2]\bigr)^{p/2}\right)^{1/p} 
&\le \frac{\sqrt{\kappa}}{(1-\sqrt{\kappa\delta})^{3/2}}
     \left(\sum_{i=1}^m
     \bigl( \mathbb{E}[|e_i|\mathbf{1}_G] \bigr)^{p}\right)^{1/p} \\
&\le \frac{\sqrt{\kappa}}{(1-\sqrt{\kappa\delta})^{3/2}}\,
    \Bigl( \mathbb{E}\Bigl[  \sum_{i=1}^m |e_i|^p \mathbf{1}_G \Bigr] \Bigr)^{1/p}\\
&\le \frac{\sqrt{\kappa}\,\varepsilon}
          {(1-\sqrt{\kappa\delta})^{3/2}} .
\end{align}
The last two inequalities follow from Jensen's inequality and the
definition of the success event \(G\), respectively.

\medskip \noindent Now we can use the CR bound.
Let $M_\cT \in \cM_{1,N}$ be a single-copy adaptive measurement protocol described by a learning tree of depth $N$ as in~\Cref{def:learning-tree}. In particular, the single-copy POVM at node $u$ is denoted by $M_u$. Since $\htheta$ is an unbiased estimator constructed by applying $M_\cT$ to $N$ copies of $\rho_{\theta,\varphi}$, the CR bound yields,
\begin{equation}
    \label{eq:non-adaptive-CR}
        V(\rho^{\otimes N}_{0},M_\cT,\htheta) \succeq (I(\rho_{0}^{\otimes N},M_{\cT})^{-1})_{AA},
\end{equation}   
where $V(\rho^{\otimes N}_{0},M_\cT,\htheta)$ is the MSEM of $\htheta$ and $I(\rho_{0}^{\otimes N},M_{\cT})$ is the FIM with respect to all parameters $(\theta,\varphi)$, both evaluated at $\rho_0$, and $(\cdot)_{AA}$ denotes the upper-left block corresponding to $\{\theta_a\}_{a\in A}$. Note the expression involves the block of the inverse FIM rather than the inverse of the block of the FIM, so as to account for the unknown nuisance parameters $\{\varphi_b\}_{b\in B}$~\cite{suzuki2020quantum}. This implies,
\bb
    \tr\left(\diag\left(V(\rho^{\otimes N}_{0},M_\cT,\htheta)\right)^{p/2}\right)^{1/p} \ge \tr\left(\diag\left((I(\rho_{0}^{\otimes N},M_{\cT})^{-1})_{AA}\right)^{p/2}\right)^{1/p}.
\ee
The L.H.S. is the same as \Cref{eq:mse-from-kurtosis} and is thus upper bounded by $\epsilon/\sqrt{1-\sqrt{\kappa\delta}}$ when $p \geq 2$. To compute the R.H.S., note that
\bb
    \label{eq:FIM-chain-rule}
    I(\rho_0^{\otimes N},M_\cT)
    = \sum_{u\in\cT\setminus\mathrm{leaf}(\cT)} p^{\rho_0}(u)\, I(\rho_0,M_u)
    = N\, I(\rho_0,\tM),
\ee
where we evaluate at $(\theta,\varphi)=(0,0)$, i.e., at $\rho_0$, $p^{\rho_0}(u)$ is the probability of reaching node $u$, and $\tM := \{\frac{1}{N}p^{\rho_0}(u)F^u_s\}_{u\in\cT\setminus\mathrm{leaf}(\cT),\,s}$ with $M_u=\{F^u_s\}_s$. The first equality is the chain rule of the FIM.
The second equality holds because $\tM$ is a valid single-copy POVM, i.e., $\tM\in\cM$, whose weights do not depend on $(\theta,\varphi)$. 
Therefore,
\bb
    \norm{\diag\Big(\big(I(\rho_0^{\otimes N},M_\cT)^{-1}\big)_{AA}\Big)^{1/2}}_p
    &= \frac{1}{\sqrt N}\norm{\diag\Big(\big(I(\rho_0,\tM)^{-1}\big)_{AA}\Big)^{1/2}}_p \\
    &\ge \frac{1}{\sqrt N}\inf_{M\in\cM}\norm{\diag\Big(\big(I(\rho_0,M)^{-1}\big)_{AA}\Big)^{1/2}}_p,
\ee
Combining with the bound $\epsilon/\sqrt{1-\sqrt{\kappa\delta}}$ when $p \geq 2$ on the L.H.S. gives
\bb
N \ge \frac{1-\sqrt{\kappa\delta}}{\epsilon^2}\inf_{M\in\cM}\norm{\diag\big((I(\rho_0,M)^{-1})_{AA}\big)^{1/2}}_p^2. 
\ee
Since the protocol needs to work for all $\rho_0\in\cS^\circ$, taking a supremum completes the proof.
Similarly, when $p \geq 1$ in general, we have coefficient $ \frac{(1-\sqrt{\kappa\delta})^3}{\kappa} $ in front of ${ \Gamma_{p}(\{O_i\}_{i=1}^m) }/ \epsilon^2$ instead. 
\end{proof}

\begin{theorem}[Lower bound for oblivious estimation using unbiased kurtosis-bounded estimators]\label{thm:estimation-lower-single-ob}
Using the adaptive measurement strategy with single-copy measurements, the sample complexity required to obtain an unbiased, $\kappa$-kurtosis-bounded estimator $\htheta_\alpha$ of $\theta_\alpha$ for all $\alpha$ satisfying $\norm{\alpha}_q \leq 1$ in \Cref{prob:oblivious-p-estimation} is lower bounded by
\begin{equation}
\label{eq:unbiased-lower-ob}
    N \ge \frac{1-\sqrt{\kappa\delta}}{\epsilon^2}\sup_{\rho_0 \in \cS^\circ} \inf_{M\in\cM} \max_{\norm{\alpha}_q \leq 1} \alpha^\top (I(\rho_0,M)^{-1})_{AA}\,\alpha
    = \frac{1-\sqrt{\kappa\delta}}{\epsilon^2}{ \Gamma_{p}^\ob(\{O_i\}_{i=1}^m) },
\end{equation}
for any $\epsilon > 0$, $p \in [1,\infty]$, and $\kappa\delta < 1$ where $\delta$ is the failure probability. Here $I(\rho_0,M)$ and $\rho_{\theta,\varphi}$ are as in \Cref{thm:estimation-lower-single}.
\end{theorem}

\begin{proof}
Fix any $\rho_0\in \cS^\circ$ and any valid choice of $\{Q_a,T_b\}$. We will show that, for every $\alpha$ satisfying $\norm{\alpha}_q \leq 1$, any unbiased estimator $\htheta_\alpha$ with bounded kurtosis that solves \Cref{prob:oblivious-p-estimation} also has small MSE. Then we use the CR bound to finish the proof, following the proof of \Cref{thm:estimation-lower-single}.

Fix any $\alpha$ satisfying $\norm{\alpha}_q \leq 1$, and let $\htheta_\alpha$ be an estimator satisfying the assumptions. 
Similar derivations as before yield,
\bb
    \label{eq:mse-from-kurtosis-ob}
    \bE[(\htheta_\alpha-\theta_\alpha)^2]\le  \frac{1}{1-\sqrt{\kappa\delta}}\bE[(\htheta_\alpha-\theta_\alpha)^2  \1_{|\htheta_\alpha-\theta_\alpha|\le\varepsilon}]\le \frac{\epsilon^2}{1-\sqrt{\kappa\delta}}.
\ee

\medskip \noindent Now we can use the CR bound.
Let $M_\cT \in \cM_{1,N}$ be a single-copy adaptive measurement protocol as in the proof of \Cref{thm:estimation-lower-single}. Since $\alpha$ is revealed only after all measurements are performed, $M_\cT$ does not depend on $\alpha$. Since $\htheta_\alpha$ is an unbiased estimator of $\theta_\alpha=\alpha\cdot\theta$ constructed by applying $M_\cT$ to $N$ copies of $\rho_{\theta,\varphi}$, the CR bound yields
\bb
    \bE[(\htheta_\alpha-\theta_\alpha)^2] \ge \alpha^\top (I(\rho_{0}^{\otimes N},M_{\cT})^{-1})_{AA}\,\alpha,
\ee
where $I(\rho_{0}^{\otimes N},M_{\cT})$ is the FIM with respect to all parameters $(\theta,\varphi)$ evaluated at $\rho_0$, as in \eqref{eq:non-adaptive-CR}. Combining with \eqref{eq:mse-from-kurtosis-ob}, which holds for every $\alpha$ satisfying $\norm{\alpha}_q \leq 1$, gives
\bb
    \frac{\epsilon^2}{1-\sqrt{\kappa\delta}} \ge \max_{\norm{\alpha}_q \leq 1} \alpha^\top (I(\rho_{0}^{\otimes N},M_{\cT})^{-1})_{AA}\,\alpha.
\ee
By \eqref{eq:FIM-chain-rule}, $I(\rho_0^{\otimes N},M_\cT) = N\, I(\rho_0,\tM)$, where $\tM\in\cM$ depends on $\rho_0$ but not on $\alpha$. Therefore,
\bb
    \max_{\norm{\alpha}_q \leq 1} \alpha^\top \big(I(\rho_0^{\otimes N},M_\cT)^{-1}\big)_{AA}\,\alpha
    &= \frac{1}{N}\max_{\norm{\alpha}_q \leq 1} \alpha^\top \big(I(\rho_0,\tM)^{-1}\big)_{AA}\,\alpha \\
    &\ge \frac{1}{N}\inf_{M\in\cM}\max_{\norm{\alpha}_q \leq 1} \alpha^\top \big(I(\rho_0,M)^{-1}\big)_{AA}\,\alpha.
\ee
Combining the above two expressions gives
\bb
N \ge \frac{1-\sqrt{\kappa\delta}}{\epsilon^2}\inf_{M\in\cM}\max_{\norm{\alpha}_q \leq 1} \alpha^\top \big(I(\rho_0,M)^{-1}\big)_{AA}\,\alpha.
\ee
Since the protocol needs to work for all $\rho_0\in\cS^\circ$, taking a supremum completes the proof.
\end{proof}

\subsection{Few-copy measurements}

We then consider protocols using $(c\geq1)$-copy measurements in $\cM_c$. Similar to the distinguishing task in \Cref{thm:distinguish-c-copy}, we show that the lower bound for these unbiased estimators for \Cref{prob:p-estimation} (and \Cref{prob:oblivious-p-estimation}) can only be reduced by at most a factor of $c$ compared with single-copy measurements. Formally, we have the following theorem

\begin{theorem}[Lower bound for estimation and oblivious estimation using unbiased kurtosis-bounded estimators and $c$-copy measurements]
Using (possibly adaptive) measurement strategy with $c$-copy measurements, the sample complexity required to obtain an unbiased, $\kappa$-kurtosis-bounded estimator for $\theta$ in \Cref{prob:p-estimation} (or an unbiased, $\kappa$-kurtosis-bounded estimator $\htheta_\alpha$ of $\theta_\alpha$ for all $\alpha$ satisfying $\norm{\alpha}_q \leq 1$ in \Cref{prob:oblivious-p-estimation}) is lower bounded by
\begin{equation}
\label{eq:unbiased-lower-c}
    N \ge \frac{1-\sqrt{\kappa\delta}}{c\epsilon^2}\sup_{\rho_0 \in \cS^\circ} \inf_{M\in\cM} \norm{\diag\big( (I(\rho_0,M)^{-1})_{AA}\big)^{1/2}}_p^2
    = \frac{1-\sqrt{\kappa\delta}}{c\epsilon^2}{ \Gamma_{p}(\{O_i\}_{i=1}^m) },
\end{equation}
for any $\epsilon > 0$ and $p \in [2,\infty]$,
or
\begin{equation}
    N \ge \frac{1-\sqrt{\kappa\delta}}{c\epsilon^2}\sup_{\rho_0 \in \cS^\circ} \inf_{M\in\cM} \max_{\norm{\alpha}_q \leq 1} \alpha^\top (I(\rho_0,M)^{-1})_{AA}\,\alpha
    = \frac{1-\sqrt{\kappa\delta}}{c\epsilon^2}{ \Gamma_{p}^\ob(\{O_i\}_{i=1}^m) },
    \label{eq:unbiased-lower-c-ob}
\end{equation}
for any $\epsilon > 0$ and $p \in [1,\infty]$, respectively, where $\kappa\delta < 1$ with $\delta$ the failure probability, and $I(\rho_0,M)$ and $\rho_{\theta,\varphi}$ are as in \Cref{thm:estimation-lower-single}. When $p \in [1,2)$, the bound holds by replacing the coefficient ${1-\sqrt{\kappa\delta}}$ above  in \eqref{eq:unbiased-lower-c} with $(1-\sqrt{\kappa\delta})^3/\kappa$. 
\label{thm:estimation-lower-c}
\end{theorem}

\begin{proof}
The proof follows those of \Cref{thm:estimation-lower-single} and \Cref{thm:estimation-lower-single-ob}, except that the learning tree $\cT$ now has depth $T=N/c$, and each non-leaf node $u$ applies a $c$-copy POVM $M_u\in\cM_c$ to $\rho^{\otimes c}$. The kurtosis argument and the CR bound are unchanged, and the chain rule \eqref{eq:FIM-chain-rule} now reads
\bb
    I(\rho_0^{\otimes N},M_\cT) = \sum_{u\in\cT\setminus\mathrm{leaf}(\cT)} p^{\rho_0}(u)\, I(\rho_0^{\otimes c},M_u).
\ee
It remains to compare $c$-copy FIMs with single-copy ones. We show that for any $M=\{M_s\}_s\in\cM_c$, there exists a single-copy POVM $G\in\cM$, depending on $\rho_0$ and $M$, such that
\bb
\label{eq:c-copy-FIM}
    I(\rho_0^{\otimes c},M) \preceq c^2\, I(\rho_0,G).
\ee
For $i\in[c]$, define $G^{[i]}=\{G^{[i]}_s\}_s$ by $G^{[i]}_s := \tr_{[c]\setminus\{i\}}\big((\id^{(i)}\otimes\rho_0^{([c]\setminus\{i\})})M_s\big)$, which is a valid single-copy POVM by the same argument as for \eqref{eq:def-G-s}. Then
\bb
    (\theta,\varphi)^\top I(\rho_0^{\otimes c},M)\,(\theta,\varphi)
    &= \sum_s \frac{\big(\sum_{i=1}^c \tr\big(M_s\big(((\theta,\varphi)\cdot(\Vec{Q},\Vec{T}))^{(i)}\otimes\rho_0^{([c]\setminus\{i\})}\big)\big)\big)^2}{d^2\,\tr(M_s\rho_0^{\otimes c})} \\
    &= \sum_s \frac{\big(\sum_{i=1}^c \tr\big(G^{[i]}_s\,(\theta,\varphi)\cdot(\Vec{Q},\Vec{T})\big)\big)^2}{d^2\,\tr(M_s\rho_0^{\otimes c})} \\
    &\le c\sum_{i=1}^c\sum_s \frac{\tr\big(G^{[i]}_s\,(\theta,\varphi)\cdot(\Vec{Q},\Vec{T})\big)^2}{d^2\,\tr(G^{[i]}_s\rho_0)} \\
    &= c\sum_{i=1}^c (\theta,\varphi)^\top I(\rho_0,G^{[i]})\,(\theta,\varphi)
    = c^2\, (\theta,\varphi)^\top I(\rho_0,G)\,(\theta,\varphi).
\ee
The first equality is the definition of the FIM, since $\rho_{\theta,\varphi} = \rho_0 + (\theta,\varphi)\cdot(\Vec{Q},\Vec{T})/d$. The second equality follows from the definition of $G^{[i]}$. The inequality follows from the Cauchy--Schwarz inequality $(\sum_{i=1}^c a_i)^2\le c\sum_{i=1}^c a_i^2$ and $\tr(M_s\rho_0^{\otimes c})=\tr(G^{[i]}_s\rho_0)$. The third equality is the definition of the FIM. In the last equality, $G:=\{\frac{1}{c} G^{[i]}_s\}_{(i,s)}\in\cM$, whose FIM is $\frac{1}{c}\sum_{i=1}^c I(\rho_0,G^{[i]})$ since its weights do not depend on $(\theta,\varphi)$. This proves \eqref{eq:c-copy-FIM}.

Applying \eqref{eq:c-copy-FIM} to each $M_u$, with the corresponding single-copy POVM $G_u=\{G_{u,s}\}_s$, we obtain
\bb
    I(\rho_0^{\otimes N},M_\cT) \preceq c^2\sum_{u\in\cT\setminus\mathrm{leaf}(\cT)} p^{\rho_0}(u)\, I(\rho_0,G_u) = Nc\, I(\rho_0,\tilde G),
\ee
where $\tilde G := \{\frac{c}{N}p^{\rho_0}(u)\,G_{u,s}\}_{u,s}\in\cM$ is a valid single-copy POVM since $p^{\rho_0}(u)$ sums to one over each of the $T=N/c$ non-leaf depths. We thus have,
% Since $A\preceq B$ implies $(A^{-1})_{AA}\succeq(B^{-1})_{AA}$ for positive definite $A,B$, we have
\bb
    \norm{\diag\Big(\big(I(\rho_0^{\otimes N},M_\cT)^{-1}\big)_{AA}\Big)^{1/2}}_p
    &\ge \frac{1}{\sqrt{Nc}}\norm{\diag\Big(\big(I(\rho_0,\tilde G)^{-1}\big)_{AA}\Big)^{1/2}}_p \\
    &\ge \frac{1}{\sqrt{Nc}}\inf_{M\in\cM}\norm{\diag\Big(\big(I(\rho_0,M)^{-1}\big)_{AA}\Big)^{1/2}}_p,
\ee
which replaces \eqref{eq:FIM-chain-rule} and the subsequent inequality in the proof of \Cref{thm:estimation-lower-single}, with $N$ replaced by $Nc$. The rest of that proof then gives \eqref{eq:unbiased-lower-c}. Similarly, \eqref{eq:unbiased-lower-c-ob} follows from the proof of \Cref{thm:estimation-lower-single-ob}, with $\alpha^\top(\cdot)\alpha$ in place of the norm.
\end{proof}

\section{Optimal estimator for high-precision shadow tomography}
\label{sec:upper}

Here we construct an estimator that performs optimally for both \Cref{prob:p-estimation} and \Cref{prob:oblivious-p-estimation} when the target precision is sufficiently small. We will first show the estimator that saturates the CR bound (up to a constant factor) for states in a neighborhood $\DDD(\rho_0) \subseteq \cS$ of some specific state $\rho_0$. Then we introduce a tomography procedure that pre-determines $\rho_0$ such that our state $\rho \in \DDD(\rho_0)$. Finally, we analyze the performance of the estimator in terms of the metrics in \Cref{prob:p-estimation} and \Cref{prob:oblivious-p-estimation} and calculate the thresholds of precision below which the sample complexity is optimal. For \Cref{prob:p-estimation}, optimality is understood relative to the lower bound for unbiased, kurtosis-bounded estimators in \Cref{sec:unbiased} (which holds for $p\in[2,\infty]$); note that the final output of our algorithm, after the median-of-means step, is not exactly unbiased.

\subsection{Optimal unbiased estimator within local regions}
\label{sec:local-est}

Although the CR bound is in general not necessarily saturable with finite sample complexity, for our specific shadow tomography problem here, we show the optimal locally unbiased estimator also performs optimally as a globally unbiased estimator within local regions up to a factor of two. The linearity of the probability distribution in $(\theta,\varphi)$ is the key property we use here.  

\begin{lemma}[Optimal estimation within local regions]
\label{lemma:local-opt}
Given $\rho_0 \in \cS^\circ$, consider quantum states parametrized as $\rho_{\theta,\varphi} = \rho_0 + \frac{1}{d}\sum_{a \in A} \theta_a Q_a + \frac{1}{d}\sum_{b \in B} \varphi_b T_b$, where $(\theta,\varphi) \in \DDD(\rho_0) = \DD(\rho_0) \cap -\DD(\rho_0)$. For any fixed POVM $M$ (that induces an invertible FIM), there exists an unbiased estimator $\htheta$ of $\theta$ such that the corresponding MSEM 
\begin{equation}
    V(\rho_{\theta,\varphi},M,\htheta) \preceq  2 (I(\rho_0,M)^{-1})_{AA},
\end{equation}
where $I(\rho_0,M) = I(\rho_{\theta,\varphi},M)|_{\theta=\varphi=0}$. 
\end{lemma}

\begin{proof}
First, we note that $\DDD(\rho_0)$ is the maximal region that satisfies the follow two properties. 
\begin{itemize}
    \item It is a closed, connected region. 
    \item $(\theta,\varphi) \in \DDD(\rho_0) \Leftrightarrow (-\theta,-\varphi) \in \DDD(\rho_0)$. 
\end{itemize}
The first property holds because $\DD(\rho_0)$ is a closed region, and all $\theta \in \DD(\rho_0)$ is connected to $\rho_0$, and the second property directly follows from the definition of $\DDD(\rho_0)$. 

Let $p_{\theta,\varphi}(x) = \trace(\rho_{\theta,\varphi} M_x)$ for some POVM $M = \{M_x\}_x$. 
We define the following estimator which is the optimal locally unbiased estimator at $\theta,\varphi = 0$: 
\begin{equation}
\htheta^{\text{opt},\rho_0}_a(y) = \sum_{c,x} (I(\rho_0,M)^{-1})_{ac} \frac{\partial_c p_{\theta,\varphi}(x)}{p_{\theta,\varphi}(x)}\bigg|_{\theta,\varphi = 0} \delta_{xy},
\end{equation}
where we use index $c$ to represent all parameters in $A \cup B$, and $y$ is the measurement outcome, which is exactly the optimal locally estimator (\eqref{eq:opt}) at $\theta,\varphi = 0$. 

First, we show it is an unbiased estimator within $\DDD(\rho_0)$. Let 
\begin{equation}
f(\theta,\varphi)_a := \bE_{\theta,\varphi}[\htheta^{\text{opt},\rho_0}_a(x)] = \sum_{c,x} (I(\rho_0,M)^{-1})_{ac} \frac{\partial_c p_{\theta,\varphi}(x)}{p_{\theta,\varphi}(x)}\bigg|_{\theta,\varphi = 0}  p_{\theta,\varphi}(x),    
\end{equation}
where the expectation is taken for measurement on state $\rho_{\theta,\varphi}$. Then 
\begin{equation}
f(0,0)_a = \sum_{c,x} (I(\rho_0,M)^{-1})_{ac} \frac{\partial_c p_{\theta,\varphi}(x)}{p_{\theta,\varphi}(x)}  p_{\theta,\varphi}(x)  \bigg|_{\theta,\varphi = 0} = 0. 
\end{equation}
For any $(\theta,\varphi) \in \DDD^\circ(\rho_0)$, where we use $ \DDD^\circ$ to denote the interior of $\DDD$, 
\begin{align}
\partial_{c'} f(\theta,\varphi)_a &= \sum_{c,x} (I(\rho_0,M)^{-1})_{ac} \frac{\partial_c p_{\theta,\varphi}(x)}{p_{\theta,\varphi}(x)}\bigg|_{\theta,\varphi = 0}  \partial_{c'} p_{\theta,\varphi}(x)\\
&= \sum_{c,x} (I(\rho_0,M)^{-1})_{ac} \frac{\partial_c p_{\theta,\varphi}(x)}{p_{\theta,\varphi}(x)}\bigg|_{\theta,\varphi = 0}  \partial_{c'} p_{\theta,\varphi}(x)\big|_{\theta,\varphi = 0}\\
&= \sum_{c'} (I(\rho_0,M)^{-1})_{ac}  I(\rho_0,M)_{cc'} = \delta_{ac},\quad \forall (\theta,\varphi) \in \DDD(\rho_0), 
\end{align}
where we use the fact that $p_{\theta,\varphi}(x)$ is a linear function for all parameters, and thus 
\begin{equation}
   \partial_{c'} p_{\theta,\varphi}(x) =  \partial_{c'} p_{\theta,\varphi}(x)\big|_{\theta,\varphi = 0}. 
\end{equation}
The above implies the unbiasedness of the estimator, i.e. 
\begin{equation}
f(\theta,\varphi) = \theta,\quad \forall (\theta,\varphi) \in \DDD(\rho_0),
\end{equation}
because $\DDD(\rho_0)$ is closed and connected. 

Next, we calculate the covariance matrix (i.e. the MSEM) of the estimator.
\begin{equation}
g_{aa'}(\theta,\varphi) = \bE_{\theta,\varphi}[(\htheta^{\text{opt},\rho_0}_a - \theta_a)(\htheta^{\text{opt},\rho_0}_{a'} - \theta_{a'})] = \bE_{\theta,\varphi}[\htheta^{\text{opt},\rho_0}_a \htheta^{\text{opt},\rho_0}_{a'}] - \theta_a \theta_{a'}.     
\end{equation}
Specifically, {\small
\begin{gather}
g_{aa}(\theta,\varphi) + \theta_a^2 = \bE_{\theta,\varphi}[(\htheta^{\text{opt},\rho_0}_a)^2] = \sum_{x,c,c'} (I(\rho_0,M)^{-1})_{ac} \frac{\partial_c p_{\theta,\varphi}(x)}{p_{\theta,\varphi}(x)} \bigg|_{\theta,\varphi = 0} (I(\rho_0,M)^{-1})_{ac'} \frac{\partial_{c'} p_{\theta,\varphi}(x)}{p_{\theta,\varphi}(x)} \bigg|_{\theta,\varphi = 0}  p_{\theta,\varphi}(x),\\
\begin{split}
 g_{aa}(-\theta,-\varphi) + \theta_a^2 &= \bE_{-\theta,-\varphi}[(\htheta^{\text{opt},\rho_0}_a)^2] \\ & =  \sum_{x,c,c'} (I(\rho_0,M)^{-1})_{ac} \frac{\partial_c p_{\theta,\varphi}(x)}{p_{\theta,\varphi}(x)} \bigg|_{\theta,\varphi = 0} (I(\rho_0,M)^{-1})_{ac'} \frac{\partial_{c'} p_{\theta,\varphi}(x)}{p_{\theta,\varphi}(x)} \bigg|_{\theta,\varphi = 0}  p_{-\theta,-\varphi}(x),
\end{split}\\
\begin{split}
g_{aa}(\theta,\varphi) + g_{aa}(-\theta,-\varphi) + 2\theta_a^2 &= 2 \sum_{x,c,c'} (I(\rho_0,M)^{-1})_{ac} \frac{\partial_c p_{\theta,\varphi}(x)}{p_{\theta,\varphi}(x)} \bigg|_{\theta,\varphi = 0} (I(\rho_0,M)^{-1})_{ac'} \frac{\partial_{c'} p_{\theta,\varphi}(x)}{p_{\theta,\varphi}(x)} \bigg|_{\theta,\varphi = 0}  p_{0,0}(x) \\
&= 2 \sum_{c,c'} (I(\rho_0,M)^{-1})_{ac} I(\rho_0,M)_{cc'} (I(\rho_0,M)^{-1})_{ac'} = 2 (I(\rho_0,M)^{-1})_{aa}, 
\end{split}
\end{gather}
}where we use the fact that $p_{\theta,\varphi}(x) + p_{-\theta,-\varphi}(x) = 2 p_{0,0}(x)$. Analogously, we can show for any real $m$-dimensional vector $(v_{a})_a$, 
\begin{equation}
    \sum_{aa'} v_{a}v_{a'} ( g_{aa'}(\theta,\varphi) + g_{aa'}(-\theta,-\varphi) + 2\theta_a\theta_{a'}) = \sum_{aa'} v_{a}v_{a'} 2 (I(\rho_0,M)^{-1})_{aa'}, 
\end{equation}
which implies 
\begin{equation}
g(\theta,\varphi) \preceq g(\theta,\varphi) + g(-\theta,-\varphi) \preceq 2 (I(\rho_0,M)^{-1})_{AA}. 
\end{equation}
It means the MSEM of the estimator $\htheta^{\text{opt},\rho_0}$ is upper bounded by $2 (I(\rho_0,M)^{-1})_{AA}$, proving the lemma. 
A crucial assumption we use above implicitly is $\rho_{-\theta,-\varphi}$ is well-defined, which is guaranteed because $(\theta,\varphi) \in \DDD(\rho_0) \Leftrightarrow (-\theta,-\varphi) \in \DDD(\rho_0)$. 
\end{proof}

\subsection{Finding \texorpdfstring{$\rho_0$}{rho0} via state tomography}

Here we discuss given an unknown state $\rho$, how to find $\rho_0$ such that $\rho \in \DDD(\rho_0)$ so that the above estimator applies. Here we abuse the notation a bit and say $\rho \in \DDD(\rho_0)$ if and only if $\rho = \rho_{\theta,\varphi}$ for some $(\theta,\varphi) \in \DDD(\rho_0)$. Finding $\rho_0$ is in general a difficult task, especially when $\rho$ is singular. 
However, we can without loss of generality consider only states within a restricted set of states: 
\begin{equation}
    \cS_{1/2} = \Big\{\rho \Big| \rho = \frac{1}{2}\Big(\sigma + \frac{\id}{d}\Big),\text{ for some density matrix }\sigma\Big\}. 
\end{equation}
We can always assume $\rho \in \cS_{1/2}$, because if not, we can apply the following quantum channel on the unknown state $\rho$
\begin{equation}
    \rho \mapsto \frac{1}{2}\rho + \frac{\id}{2d} 
\end{equation}
whose output state belongs to $\cS_{1/2}$ and then perform the estimation on the output state. The channel maps the expectation values $\trace(O_i \rho) \mapsto \frac{1}{2}\trace(O_i \rho)$, which induces at most a constant factor in the estimation precision and the sample complexity bounds. In this case, it is sufficient to find an estimator of $\rho$ that is within $1/(4d)$ of its operator norm. 
\begin{lemma}
\label{lemma:tomography}
For any state $\rho \in \cS_{1/2}$, if $\|\rho_0 - \rho\|_{\infty} \leq {1}/(4d)$, $\rho \in \DDD(\rho_0)$. 
\end{lemma}
\begin{proof}
We first note that $\rho \in \DDD(\rho_0)$ if and only if 
\begin{equation}
    \rho_0 - (\rho -\rho_0) \succeq 0 ~~\Leftrightarrow~~ 
    \rho - (\rho-\rho_0) - (\rho -\rho_0) \succeq 0.
\end{equation}
It holds when $\|\rho_0 - \rho\|_{\infty} \leq {1}/(4d)$ because 
\begin{equation}
    \rho - (\rho-\rho_0) - (\rho -\rho_0) \succeq \rho - 2 \norm{\rho - \rho_0}_\infty \id \succeq \rho - \id/(2d) \succeq 0. 
\end{equation}
\end{proof}
The above implies that any algorithm that produces an estimate of $\rho$ such that its $\infty$-norm distance to $\rho$ is at most $1/(4d)$ will be sufficient to serve as the first step to determine $\rho_0$, prior to applying the unbiased estimator $\htheta^{\text{opt},\rho_0}$ in the second step. To achieve the target accuracy in tomography with single-copy measurements, one can apply the Haar random measurement $\{d\ket{v}\bra{v}{\rm d}v\}$ and use the estimator
\begin{equation}
    \hrho_0(N_0,u) = \frac{1}{N_0} \sum_{i=1}^{N_0} ((d+1)\ket{u_i}\bra{u_i} - \id), 
\end{equation}
where $u = (u_1,\ldots,u_{N_0})$ are the outcomes from measuring $\rho$. With probability $1-\delta$~\cite{kueng2017low,gutaFastStateTomography2020,chenCotler2025Lecture6TomographyII}, 
\begin{equation}
    \norm{\hrho_0(N_0,u) - \rho}_\infty \leq \text{Constant} \times \max\left\{\frac{d+\log(1/\delta)}{N},\sqrt{\frac{d+\log(1/\delta)}{N}}\right\} .
\end{equation}
In particular, $\hrho_0(N_0,u)$ must be a well-defined density matrix when $\rho \in \cS_{1/2}$ and $\norm{\hrho_0(N_0,u) - \rho}_\infty \leq 1/4d$.  
Therefore, $N_0 = O(d^3)$ is sufficient to achieve the target precision with high probability and find a local region $\DDD(\rho_0)$ to apply the optimal local estimator. 

\subsection{Conversion from MSEM to \texorpdfstring{$p$-norm}{p-norm} error}

We showed in \Cref{sec:local-est}, the optimal locally unbiased estimator $\htheta^{\text{opt},\rho_0}$ performs optimally in estimating $\theta$ in $\rho_{\theta,\varphi}$ within $\DDD(\rho_0)$ in the sense that it achieves, up to a factor of two, the optimal MSEM given by the CR bound. Our goal is to learn observables with $p$-norm error. Below we show, using the (coordinate-wise) median-of-means estimator, we can, with probability at least $1-\delta$, obtain a bounded $p$-norm error using the estimator of bounded $p$-average RMSE with an overhead of $O(\log(m/\delta))$. 

\begin{lemma}
\label{lemma:MoM}
Fix the POVM $M$. For any unbiased estimator $\htheta$ that achieves
\begin{equation}
    \norm{\diag(V(\rho_{\theta,\varphi},M,\htheta))^{1/2}}_p \leq \xi,
\end{equation}
there is another estimator $\htheta_\MoM$ that achieves $p$-norm error $\epsilon$ with probability at least $1-\delta$ that uses 
\begin{equation}
  O\left(\log\left(\frac{m}{\delta}\right)\right) \times O\left(\frac{\xi^2}{\epsilon^2}\right)
\end{equation}
samples of $\htheta$. 
\end{lemma}

\begin{proof}
    Given an unbiased estimator $\hat{\theta}$ for $\theta \in \mathbb{R}^m$ satisfying the $p$-average RMSE bound, i.e. 
\begin{align}
\norm{\text{diag}(V(\rho_{\theta,\varphi},M,\htheta))^{1/2}}_p = \left( \sum_{j=1}^m \sigma_j^p \right)^{1/p} \le \xi, 
\end{align}
where $\sigma_j^2 := \Var[\htheta_j] = V(\rho_{\theta,\varphi},M,\htheta)_{jj}$ is the variance of estimating $\theta_j$. Our goal is to construct an estimator $\hat{\theta}_{\text{MoM}}$ such that with probability at least $1-\delta$:
\begin{align}
\norm{\hat{\theta}_{\text{MoM}} - \theta}_p \le \epsilon
\end{align}
using a sample complexity overhead that is poly-logarithmic in $m$ and $1/\delta$.

We start with the definition of the coordinate-wise median-of-means estimator: 
\begin{enumerate}[wide, labelwidth=!,itemindent=!,labelindent=0pt, leftmargin=0em, label={\arabic*.}, parsep=0pt]
    \item \textit{Sampling:} Collect $N_1 = K \times B$ independent samples of $\hat{\theta}$. Divide them into $K$ batches, each of size $B$.
    \item \textit{Batch Averaging:} For each batch $\ell \in \{1, \dots, K\}$, compute the empirical mean:
    \begin{align}
    \hat{\theta}^{[\ell]}_B := \frac{1}{B} \sum_{i=1}^B \hat{\theta}^{[\ell,i]},
    \end{align}
    where $\hat{\theta}^{[\ell,i]}$ is the $i$-th sample of $\htheta$ in the $\ell$-th batch. 
    \item \textit{Coordinate-wise Median:} For each coordinate $j \in \{1, \dots, m\}$, compute the median of the batch means:
    \begin{align}
    \hat{\theta}_{\text{MoM}, j} = \text{Median}\left( \hat{\theta}^{[1]}_{B,j},\ldots,\hat{\theta}^{[K]}_{B,j} \right). 
    \end{align}
\end{enumerate}

We now try to analyze the performance of the estimator. 
\begin{enumerate}[wide, labelwidth=!,itemindent=!,labelindent=0pt, leftmargin=0em, label={\arabic*.}, parsep=0pt]
    \item \textit{Bounding probability for a single batch (Chebyshev). }Fix a coordinate $j$, the variance of the batch mean $\hat{\theta}^{[\ell]}_{B,j}$ is $\sigma_j^2/B$. By Chebyshev's inequality, for any $\ell$:
    \begin{align}
    \Pr\left[ \left| \hat{\theta}^{[\ell]}_{B,j} - \theta_j \right| > 2 \frac{\sigma_j}{\sqrt{B}} \right] \le \frac{\Var[\hat{\theta}^{[\ell]}_{B,j}]}{(2 \sigma_j / \sqrt{B})^2} = \frac{\sigma_j^2/B}{4 \sigma_j^2/B} = \frac{1}{4}.
    \end{align}
    Let us define the ``bad'' event for the $\ell$-th batch on coordinate $j$ as $\mathcal{E}_{\ell, j} := \big\{ | \hat{\theta}^{[\ell]}_{B,j} - \theta_j | > 2 \frac{\sigma_j}{\sqrt{B}} \big\}$. We have established that $\Pr[\mathcal{E}_{\ell, j}] \le 1/4$.
    \item \textit{Bounding probability for the Median (Chernoff/Hoeffding). }For the median $\hat{\theta}_{\text{MoM}, j}$ to deviate from $\theta_j$ by more than $2 \sigma_j / \sqrt{B}$, more than half of the batches must satisfy the bad event $\mathcal{E}_{\ell, j}$. Let $S_j = \sum_{\ell=1}^K \id[\mathcal{E}_{\ell, j}]$, where $\id[\cdot]$ is the indicator function. The expected number of bad batches is $\mathbb{E}[S_j] \le K/4$. The failure condition for the median is $S_j \ge K/2$. Using Hoeffding's inequality:
    \begin{align}
    \Pr\left[ \left| \hat{\theta}_{\text{MoM}, j} - \theta_j \right| > 2 \frac{\sigma_j}{\sqrt{B}} \right] &\le \Pr\left[ S_j \ge \frac{K}{2} \right] \\
    &\le \exp\left( -2K \left( \frac{1}{2} - \frac{1}{4} \right)^2 \right) = \exp\left( -\frac{K}{8} \right).
    \end{align}
    \item \textit{Union Bound over Coordinates. }We require all coordinates to satisfy their respective bounds simultaneously to preserve the sum-structure of the $p$-norm. Apply the union bound over all $m$ coordinates. Let $\mathcal{F}$ be the event that \textit{any} coordinate $j$ fails (i.e., $| \hat{\theta}_{\text{MoM}, j} - \theta_j | > 2 \frac{\sigma_j}{\sqrt{B}}$ for some $j$).
    \begin{align}
    \Pr[\mathcal{F}] \le \sum_{j=1}^m \exp\left( -\frac{K}{8} \right) = m \exp\left( -\frac{K}{8} \right).
    \end{align}
    To ensure this failure probability is at most $\delta$, we set:
    \begin{align}
    m \exp\left( -\frac{K}{8} \right) \le \delta \implies K \ge 8 \ln\left( \frac{m}{\delta} \right).
    \end{align}
    \item \textit{Bounding the $p$-norm error. } Conditioned on the success event $\mathcal{F}^c$ (the complement of failure, which occurs with probability $\ge 1-\delta$), we have that for all $j \in \{1, \dots, m\}$:
    \begin{align}
    \left| \hat{\theta}_{\text{MoM}, j} - \theta_j \right| \le \frac{2 \sigma_j}{\sqrt{B}}.
    \end{align}
    Now, we compute the $p$-norm of the error vector:
    \begin{align}
    \norm{\hat{\theta}_{\text{MoM}} - \theta}_p = \left( \sum_{j=1}^m \left| \hat{\theta}_{\text{MoM}, j} - \theta_j \right|^p \right)^{1/p}
    \end{align}
    Substituting the coordinate-wise bounds:
    \begin{align}
    \norm{\hat{\theta}_{\text{MoM}} - \theta}_p \le \left( \sum_{j=1}^m \left( \frac{2 \sigma_j}{\sqrt{B}} \right)^p \right)^{1/p} = \frac{2}{\sqrt{B}} \left( \sum_{j=1}^m \sigma_j^p \right)^{1/p}.
    \end{align}
    Using the initial assumption that the $p$-average RMSE is bounded by $\xi$:
    \begin{align}
    \norm{\hat{\theta}_{\text{MoM}} - \theta}_p \le \frac{2 \xi}{\sqrt{B}}.
    \end{align}
\end{enumerate}

To achieve a target $p$-norm error of $\epsilon$, we set $\frac{2 \xi}{\sqrt{B}} \le \epsilon$, which implies $B \ge \frac{4 \xi^2}{\epsilon^2}$. The total number of samples required is:
\begin{align}
N_1 = K \times B = O\left( \log\left(\frac{m}{\delta}\right) \right) \times O\left( \frac{\xi^2}{\epsilon^2} \right).
\end{align}
Thus, using the coordinate-wise median-of-means estimator, we obtain a bounded $p$-norm error using the estimator of bounded $p$-average RMSE with an overhead of $O(\log(m/\delta))$.
\end{proof}

\subsection{Algorithm, sample complexity, and threshold}

Here we describe the algorithm that combines the three steps introduced above, analyze the corresponding sample complexity and derive the threshold below which the sample complexity matches our lower bound (up to logarithmic overhead). We consider the case of shadow tomography (\Cref{prob:p-learning}, equivalent to \Cref{prob:p-estimation}) and the case of oblivious single-observable estimation (\Cref{prob:oblivious-p-learning}, equivalent to \Cref{prob:oblivious-p-estimation}) separately below. 

\subsubsection{Shadow estimation with \texorpdfstring{$p$}{p}-norm error}

\begin{theorem}[Upper bound for estimation with \texorpdfstring{$p$}{p}-norm error]\label{thm:p-estimation-upper}
For any $\epsilon > 0$ and $p \in [1,\infty]$ 
there exists an algorithm that uses 
\begin{align}
\begin{split}
N &= O(d^{3}) + O\left(\frac{\log(2m)}{\epsilon^2} \sup_{\substack{\rho_0 \in \cS^\circ}} \inf_{M\in\cM}  {\norm{\diag\big( (I(\rho_0,M)^{-1})_{AA}\big)^{1/2}}^2_p} \right) \\
&= O(d^{3}) + O\left(\frac{\log(2m) \Gamma_{p}(\{O_i\}_{i=1}^m)}{\epsilon^2}\right)
\end{split}
\end{align}
copies of $\rho$ and single-copy measurements to solve \Cref{prob:p-learning} (or equivalently, \Cref{prob:p-estimation}). In particular, when 
\begin{gather}
    \epsilon \leq \overline{\eta} = \sqrt{ \frac{\log(2m)\Gamma_{p}(\{O_i\}_{i=1}^m)}{d^{3}} },\\
    N = O\left(\frac{\log(2m) \Gamma_{p}(\{O_i\}_{i=1}^m)}{\epsilon^2}\right).  
\end{gather}
\end{theorem}

\begin{proof}
For each $\rho_0$, we can pick a POVM $M^\diamond(\rho_0) \in \cM$ that achieves 
\begin{align}
\norm{\diag(I(\rho_0,M^\diamond(\rho_0))^{-1}_{AA})^{1/2}}^2_p  \leq 2 \inf_{M \in \cM } \norm{\diag(I(\rho_0,M)^{-1}_{AA})^{1/2}}^2_p \leq 2 \Gamma_{p}(\{O_i\}_{i=1}^m). \label{eq:M-dia-1}
\end{align}
Then we consider the following algorithm: 
\begin{algorithm}[htbp]
    \DontPrintSemicolon
    \caption{Shadow tomography with bounded $p$-norm error}\label{alg:learning}
    \KwInput{Observables $\{O_i\}_{i=1}^m$, $p \in [1,\infty]$, $\delta, \epsilon>0$, $N = N_0 + N_1$ copies of a $d$-dimensional state $\rho$}
    \KwOutput{Estimators $\hat{o} = (\hat{o}_i,\ldots,\hat{o}_m)$}
    \Goal{With probability at least $1-\delta$, $(\sum_{i=1}^m \abs{\trace(O_i\rho) - \hato_i}^p)^{1/p} < \epsilon$ }
    Mix all $\rho$'s evenly with the maximally mixed state: $\rho \mapsto \frac{1}{2}\rho + \frac{\id}{2d}$\\
    Apply the Haar random measurement on each of the $N_0$ copies of $\rho$ and obtain $\rho_0 = \hrho_0(N_0,u)$ as a coarse estimate $\rho$, where $u = (u_1,\ldots,u_{N_0})$ are the measurement outcomes, such that $\norm{\rho_0 - \rho}_\infty \leq 1/4d$ with high probability\\
    Apply POVM $M^\diamond(\rho_0)$ on each of the $N_1$ copies of $\rho$ and obtain $N_1$ unbiased estimates of $\theta$ using the optimal unbiased estimator $\htheta^{\text{opt},\rho_0}$\\
    Divide the $N_1$ estimates into $K$ groups each with $B$ elements and calculate the coordinate-wise median-of-means estimator $\hat{\theta}_{\text{MoM}}$\\
    Return $\hat{o}_i =2 ( (\hat{\theta}_{\text{MoM}})_i + \trace(O_i \rho_0) )$\\
\end{algorithm}

Here picking 
\begin{gather}
N_0 = O(d^2 (d + \log(1/\delta))), \quad N_1 = O\left(\frac{\log(m/\delta)}{\epsilon^2} \Gamma_{p}(\{O_i\}_{i=1}^m) \right), 
\\
K = O\left( \log\left(\frac{m}{\delta}\right)\right),\quad 
B = O\left(\frac{1}{\epsilon^2} \Gamma_{p}(\{O_i\}_{i=1}^m) \right), 
\end{gather}
is sufficient to guarantee the desired performance of our algorithm. To see this, we first notice by \lemmaref{lemma:tomography}, $\rho \in \DDD(\rho_0)$ with high probability. Then using \lemmaref{lemma:local-opt}, we have 
\begin{align}
V(\rho_{\theta,\varphi},M^\diamond(\rho_0),\htheta^{\text{opt},\rho_0}) \preceq 2 (I(\rho_0,M^\diamond(\rho_0))^{-1})_{AA},
\end{align}
and 
\begin{align}
\norm{\text{diag}(V(\rho_{\theta,\varphi},M^\diamond(\rho_0),\htheta^{\text{opt},\rho_0})^{1/2})}^2_p \leq 2 
\norm{\text{diag}((I(\rho_0,M^\diamond(\rho_0))^{-1})_{AA})^{1/2})}^2_p \leq 4 \Gamma_{p}(\{O_i\}_{i=1}^m). 
\end{align}
Using \lemmaref{lemma:MoM}, it then follows that a coordinate-wise median-of-means estimator with $K = O(\log(m/\delta))$ groups of $B = O(\Gamma_{p}(\{O_i\}_{i=1}^m)/\epsilon^2)$ elements is sufficient for our purpose. 
\end{proof}

\subsubsection{Oblivious single-observable estimation}

\begin{theorem}[Upper bound for oblivious single-observable estimation]\label{thm:ob-estimation-upper}
For any $\epsilon > 0$ and $p \in [1,\infty]$ 
there exists an algorithm that uses 
\allowdisplaybreaks
\begin{align}
N &= O(d^{3}) + O\left(\frac{1}{\epsilon^2} \sup_{\substack{\rho_0 \in \cS^\circ}} \inf_{M\in\cM} \max_{\substack{\alpha\in\bR^m\\ \norm{\alpha}_q \leq 1}} {\alpha^\top (I(\rho_0,M)^{-1})_{AA} \alpha} \right)\\
&= O(d^{3}) + O\left(\frac{\Gamma^{\ob}_{p}(\{O_i\}_{i=1}^m)}{\epsilon^2}\right)
\end{align}
copies of $\rho$ and single-copy measurements to solve \Cref{prob:oblivious-p-learning} (or equivalently, \Cref{prob:oblivious-p-estimation}). When 
\begin{equation}
    \epsilon \leq \overline{\eta}^{\ob} = \sqrt{ \frac{\Gamma^{\ob}_{p}(\{O_i\}_{i=1}^m)}{d^{3}} },
\end{equation}
\begin{equation}
    N = O\left(\frac{\Gamma^{\ob}_{p}(\{O_i\}_{i=1}^m) }{\epsilon^2}\right).  
\end{equation}
\end{theorem}

\begin{proof}
For each $\rho_0$, we pick a POVM $M^\diamond(\rho_0) \in \cM$ that achieves 
\begin{equation}
\max_{\substack{\alpha\in\bR^m\\ \norm{\alpha}_q \leq 1}} {\alpha^\top (I(\rho_0,M^\diamond(\rho_0))^{-1})_{AA} \alpha}  \leq 2 \inf_{M\in\cM} \max_{\substack{\alpha\in\bR^m\\ \norm{\alpha}_q \leq 1}} {\alpha^\top (I(\rho_0,M)^{-1})_{AA} \alpha} \leq 2 \Gamma^{\ob}_{p}(\{O_i\}_{i=1}^m). 
 \label{eq:M-dia-2}
\end{equation}
Then we consider the following algorithm: 
\begin{algorithm}[htbp]
    \DontPrintSemicolon
    \caption{Oblivious estimation of $O_\alpha = \sum_i \alpha_i O_i$ for $\norm{\alpha}_q \leq 1$ }\label{alg:oblivious-learning}
    \KwInput{Observables $\{O_i\}_{i=1}^m$, $p \in [1,\infty]$, $\delta, \epsilon>0$, copies of a $d$-dimensional state $\rho$
    }
    \KwOutput{Estimator $\hat{o}_\alpha$}
    \Goal{With probability at least $1-\delta$, $|\trace(\rho O_\alpha) - \hat{o}_\alpha| < \epsilon$ }
    Mix all $\rho$'s evenly with the maximally mixed state: $\rho \mapsto \frac{1}{2}\rho + \frac{\id}{2d}$\\
    Apply the Haar random measurement on each of the $N_0$ copies of $\rho$ and obtain $\rho_0 = \hrho_0(N_0,u)$ as a coarse estimate $\rho$, where $u = (u_1,\ldots,u_{N_0})$ are the measurement outcomes, such that $\norm{\rho_0 - \rho}_\infty \leq 1/4d$ with high probability\\
    Apply POVM $M^\diamond(\rho_0)$ on each of the $N_1$ copies of $\rho$ and obtain $N_1$ unbiased estimates of $\theta$ using the optimal unbiased estimator $\htheta^{\text{opt},\rho_0}$\\
    Reveal $\alpha \in \bR^m$, which satisfies $\norm{\alpha}_q \leq 1$\\
    Divide the $N_1$ samples of estimates $\htheta_\alpha := \sum_{i=1}^m \alpha_i \htheta^{\text{opt},\rho_0}_i$ into $K$ groups each with $B$ elements and calculate the median-of-means estimator $\hat{\theta}_{\text{MoM},\alpha}$ \\
    Return $\hat{o}_\alpha = 2 (\hat{\theta}_{\text{MoM},\alpha} + \trace(\rho_0 O_\alpha))$\\
    \end{algorithm}

Here picking 
\begin{gather}
N_0 = O(d^2 (d + \log(1/\delta))), \quad N_1 = O\left(\frac{\log(1/\delta)}{\epsilon^2} \Gamma_{p}(\{O_i\}_{i=1}^m) \right),\\
K = O(\log(1/\delta)),\quad B = O\left(\frac{1}{\epsilon^2} \Gamma_{p}^\ob(\{O_i\}_{i=1}^m) \right),
\end{gather}
is sufficient to guarantee the desired performance of our algorithm. To see this, we first notice by \lemmaref{lemma:tomography}, $\rho \in \DDD(\rho_0)$ with high probability. Then using \lemmaref{lemma:local-opt}, we have 
\begin{align}
V(\rho_{\theta,\varphi},M^\diamond(\rho_0),\htheta^{\text{opt},\rho_0}) \preceq 2 (I(\rho_0,M^\diamond(\rho_0))^{-1})_{AA},
\end{align}
and 
\begin{align}
\alpha^\top V(\rho_{\theta,\varphi},M^\diamond(\rho_0),\htheta^{\text{opt},\rho_0}) \alpha \leq  2 \alpha^\top (I(\rho_0,M^\diamond(\rho_0))^{-1})_{AA} \alpha \leq 4 \Gamma^{\ob}_{p}(\{O_i\}_{i=1}^m). 
\end{align}
Furthermore, let $\htheta_\alpha = \sum_{i=1}^m \alpha_i (\htheta^{\text{opt},\rho_0})_i$, we have 
\begin{align}
\Var[\htheta_\alpha] = \alpha^\top V(\rho_{\theta,\varphi},M^\diamond(\rho_0),\htheta^{\text{opt},\rho_0}) \alpha \leq 4 \Gamma^{\ob}_{p}(\{O_i\}_{i=1}^m).  
\end{align}
We can consider a simpler variant of \lemmaref{lemma:MoM}, where only a single parameter $\theta_\alpha = \sum_{i=1}^\alpha \alpha_i \theta_i$ is to be estimated, which implies whenever $\Var[\htheta_\alpha] \leq \xi^2$, $K = O(\log(1/\delta))$ and $B = O(\xi^2/\epsilon^2)$ samples can guarantee the Median-of-Mean estimator $\htheta_{\MoM,\alpha}$ has an additive error within $\epsilon/2$ with probability at least $1-\delta$, which achieves our desired precision. 
\end{proof}

\subsection{Simple relation between \texorpdfstring{$\Gamma_p$}{Gamma\_b} and \texorpdfstring{$\Gamma^\ob_p$}{Gamma\^op\_b}}

Since the oblivious estimation task is no harder than the shadow estimation task, comparing our lower bound for the former (\Cref{thm:distinguish_thres_general}) with our upper bound for the latter (\Cref{thm:p-estimation-upper}), we already know that
\begin{equation}
    \Gamma^\ob_p(\{O_i\}_{i=1}^m) = O(\log (2m)\,\Gamma_p(\{O_i\}_{i=1}^m))
\end{equation}
for all $p \in [1,\infty]$, which means $\Gamma^\ob_p$ is no larger than $\Gamma_p$ up to a logarithmic factor. Here we show that in fact $\Gamma^\ob_p\le\Gamma_p$.
\begin{lemma}\label{lem:simple_relation}
For all $p \in [1,\infty]$, $\Gamma^\ob_p(\{O_i\}_{i=1}^m) \leq \Gamma_p(\{O_i\}_{i=1}^m).$ The equality holds when $p = \infty$. 
\end{lemma}
\begin{proof}
To prove the inequality, we only need to show for any matrix $R \succeq 0$, 
\begin{equation}
\label{eq:gamma-gamma-ob}
\max_{\|\alpha\|_q \le 1} \alpha^\top R \alpha \le \bigl\|\diag(R)^{1/2}\bigr\|_p^2.
\end{equation}
Let $v_i$ denote the $i$-th column of $\sqrt{R}$. Then for any $\alpha$,
\begin{align}
\sqrt{\alpha^\top R \alpha}
= \|\sqrt{R}\alpha\|_2
= \left\|\sum_{i=1}^m \alpha_i v_i\right\|_2 
\le \sum_{i=1}^m |\alpha_i|\|v_i\|_2.
\end{align}
Moreover, $\|v_i\|_2^2 = R_{ii}$. We have 
\begin{align}
\sqrt{\alpha^\top R \alpha} \le \sum_{i=1}^m |\alpha_i| \sqrt{R_{ii}} . 
\end{align}
For any $\|\alpha\|_q\le 1$, by H\"{o}lder's inequality,
\begin{align}
\sqrt{\alpha^\top R \alpha} = \sum_{i=1}^m |\alpha_i| \sqrt{R_{ii}} \leq \left(\sum_{i=1}^m \sqrt{R_{ii}}^p\right)^{1/p} = \bigl\|\diag(R)^{1/2}\bigr\|_p, 
\end{align}
proving \eqref{eq:gamma-gamma-ob}. Finally, we note that when $p = \infty$ and $q = 1$, 
\begin{align}
\max_{\norm{\alpha}_1 \leq 1} \alpha^\top R \alpha = \max_{i} R_{ii} = \bigl\|\diag(R)^{1/2}\bigr\|^2_\infty.
\end{align}
This is because $\alpha^\top R \alpha$ is a convex function in $\alpha$, and the maximum can be taken at extreme points, i.e. when $\alpha$ has only one entry equal to one. 
\end{proof}

\section{Example: Pauli estimation}
\label{sec:example}

We showcase our results with a concrete example of Pauli observable estimation. While the case of $p=\infty$ has been thoroughly investigated in~\cite{chen2024optimal} with tight sample complexity bound obtained, our framework extends the analysis to all $p\in[1,\infty]$ in the high-precision regime. 

\begin{theorem}[Oblivious learning of Pauli observables, single-copy measurements]\label{thm:oblivious-pauli}
    Let $d=2^n$ and 
    $\{O_i\coleq P_i\}_{i=1}^{d^2-1}$ be all $n$-qubit traceless Pauli operators. Then,
    \begin{equation}
    \Omega\!\left({d}\right)\le \Gamma^\ob_p\le
    \left\{
    \begin{aligned}
         &O\!\left({d\log d}\right),\quad&&\textrm{if~}p\in[2,\infty],\\
         &O\!\left({d^{\frac4p-1}\log d}\right),\quad&&\textrm{if~}p\in[1,2).\\
    \end{aligned}
    \right.
    \end{equation}
    Furthermore, the following sample complexity bounds for the oblivious learning problems (\Cref{prob:oblivious-p-learning}) hold:
    \begin{equation}
    \begin{aligned}
        N& = \Omega(\Gamma^\ob_p/\epsilon^2) = \Omega(d/\varepsilon^2),\quad&&\textrm{when}~\varepsilon< \frac16d^{\frac 2p-2}.
    \\N&= O(\Gamma^\ob_p/\epsilon^2)  = \left\{
    \begin{aligned}
         &O\left({d\log d}/\epsilon^2\right),\quad\textrm{if~}p\in[2,\infty],\\
         &O\left({d^{\frac4p-1}\log d}/\epsilon^2\right),\quad\textrm{if~}p\in[1,2),\\
    \end{aligned}
    \right.\quad&&\textrm{when}~
          \varepsilon< c_0d^{-1}.
    \end{aligned}
    \end{equation}
    Here $c_0>0$ is some absolute constant.
\end{theorem}
\begin{proof}
    When $O_i=P_i$ for all $i\in[m]$ where $m=d^2-1$, it is easy to see that the dual operator basis is uniquely determined as $Q_i=P_i$ for all $i\in[m]$.
    Recall the dual representation of $\Gamma_p^\ob$ thanks to \Cref{lem:duality},
    \begin{equation}
        (\Gamma^\ob_p)^{-1} = \inf_{\rho_0\in\cS^\circ}  \sup_{M\in\cM} \min_{\|\theta\|_p = 1} \theta^\top I(\rho_0,M) \theta.
    \end{equation}
    Let us first derive an upper bound for this,
    \begin{align}
        (\Gamma^\ob_p)^{-1} &\le \sup_{M\in\cM}\min_{\|\theta\|_p = 1} \theta^\top I(\id/d,M) \theta
        \\&\le \sup_{M\in\cM}\min_{i\in[m]} I(\id/d,M)_{ii}
        \\&= \sup_{M\in\cM}\min_{i\in[m]} \sum_x\frac{\tr^2(M_xP_i)}{d\Tr(M_x)}
        \\&\le \sup_{M\in\cM}\frac1m\sum_{i\in[m]} \sum_x\frac{\tr^2(M_xP_i)}{d\Tr(M_x)}
        \\&\le \sup_{M\in\cM} \frac{d}{d^2-1}\sum_x\frac{\tr(M_x^2)}{d\Tr(M_x)}
        \\&\le \frac{d}{d^2-1} = O(d^{-1}).
    \end{align}
    The first line fixes $\rho_0=\id/d$. The second line restricts $\theta$ to one-hot vectors. The fifth line uses the twirling formula for Pauli operators: $\sum_{P\in\cP_n}\Tr^2(PM_x)=d\Tr(M_x^2)$. The last line uses $\Tr(M_x^2)\le\tr^2(M_x)$ and the normalization condition of POVMs.

    Next, we derive a lower bound. By fixing the $M$ to be the Haar random measurement $M_\mu\coleq\{d\ketbra{\psi}{\psi}\}_{\psi\sim\mu}$, where $\mu$ is the Haar measure over $d$-dimensional pure states:
    \begin{equation}\label{eq:eg_Gamma_lower}
    \begin{aligned}
        (\Gamma^\ob_p)^{-1} &\ge \inf_{\rho_0\in\cS^\circ} \min_{\|\theta\|_p = 1} \theta^\top I(\rho_0,M_\mu) \theta
        \\&= \inf_{\rho_0\in\cS^\circ} \min_{\|\theta\|_p = 1}\E_{\psi\sim\mu}\frac{\braket{\psi|\sum_i\theta_iP_i|\psi}^2}{d\braket{\psi|\rho_0|\psi}}.
    \end{aligned}
    \end{equation}
    For any $x>0$ and $k\in\mathbb N_+$, we have the following concentration bound:
    \begin{equation}
        \Pr_{\psi\sim\mu}\left(\braket{\psi|\rho_0|\psi}\ge x\right)\le\frac{\E_{\psi\sim\mu}\braket{\psi|\rho_0|\psi}^k}{x^k} = \frac{\Tr\left(\rho_0^{\otimes k}\Pi_k\right)}{\binom{d+k-1}{k}x^k}\le \frac{k!}{d^kx^k}\le \frac{\sqrt{2\pi k}k^k e^{\frac1{12k}}}{e^kd^kx^k}.
    \end{equation}
    The first inequality uses the $k$-th order Markov inequality. The equality uses the Haar integral formula, where $\Pi_k$ is the projector onto the symmetric subspace. The second inequality uses that $\Pi_k\le\id$ and simple algebra on the binomial coefficients. The last inequality uses a non-asymptotic Stirling's upper bound on $k!$~\cite{robbins1955remark}.
    Now we take $x\coleq ({c\log d})/{d}$ and $k = \lceil c\log d\rceil$ for some constant $c\ge4$. The above inequality becomes
    \begin{equation}
        \Pr_{\psi\sim\mu}\left(\braket{\psi|\rho_0|\psi}\ge \frac{c\log d}{d}\right)\le C \frac{\sqrt{\log d}}{d^c},
    \end{equation}
    for some constant $C>0$ that depends only on $c$. Now, fix any $\rho_0\in\cS^\circ$ and $\theta\in\mathbb{R}^m$. Call $\psi$ good if $\braket{\psi|\rho_0|\psi}<(c\log d)/d$ and bad otherwise.
    We have:
    \begin{equation}
    \begin{aligned}
        \E_{\psi\sim\mu}\frac{\braket{\psi|\sum_i\theta_iP_i|\psi}^2}{d\braket{\psi|\rho_0|\psi}} &\ge \int_{\psi:\mathrm{good}}\mathrm{d}\mu_\psi\frac{\braket{\psi|\sum_i\theta_iP_i|\psi}^2}{d\braket{\psi|\rho_0|\psi}}
        \\&>\frac{1}{c\log d}\int_{\psi:\mathrm{good}}\mathrm{d}\mu_\psi{\braket{\psi|\sum_i\theta_iP_i|\psi}^2}
        \\&= \frac{1}{c\log d}\left(\int \mathrm{d}\mu_\psi{\braket{\psi|\sum_i\theta_iP_i|\psi}^2} - \int_{\psi:\mathrm{bad}}\mathrm{d}\mu_\psi{\braket{\psi|\sum_i\theta_iP_i|\psi}^2}\right)
        \\&\ge \frac{1}{c\log d}\left( \frac{d\|\theta\|_2^2}{d(d+1)} - C\frac{\sqrt{\log d}}{d^c}\|\theta\|_1^2\right)
        \\&\ge \frac{\|\theta\|^2_2}{c\log d}\left( \frac1{d+1} - C\frac{\sqrt{\log d}}{d^{c-2}} \right)
        \\&= \|\theta\|_2^2 ~\Omega\left(\frac{1}{d\log d}\right).
    \end{aligned}
    \end{equation}
    The first line restricts the integral to good $\psi$'s. The second line uses the definitions of good $\psi$'s. For the fourth line, the first term uses the Haar integral formula for $k=2$ and $\Tr(P_i)=0$ and $\Tr(P_iP_j)=d\delta_{ij}$; the second term uses $\braket{\psi|P_i|\psi}\le1$. The fifth line uses $\|\theta\|_1^2\le m\|\theta\|_2^2$. 
    The last line holds as long as we take, say, $c\ge 4$.
    Put this back to~\eqref{eq:eg_Gamma_lower}:
    \begin{equation}
        (\Gamma_p^\ob)^{-1}\ge \min_{\|\theta\|_p=1}\|\theta\|_2^2~\Omega\left(\frac{1}{d\log d}\right) = \left\{\begin{aligned}
            &\Omega\left(\frac{1}{d\log d}\right),\quad\textrm{if}~p\in[2,\infty],\\
            &\Omega\left(\frac{1}{d^{\frac{4}{p}-1}\log d}\right),\quad\textrm{if}~p\in[1,2).\\
        \end{aligned}\right.
    \end{equation}
    This concludes our proof for the bounds on $\Gamma_p^\ob$. 
    
    Now we prove the claimed thresholds. The lower bound threshold immediately follows from \Cref{thm:lower_threshold}:
    \begin{equation}
        \eta^{\ob}_{\comp} = \frac{1}{6\left\|\|P_1\|_\infty,\cdots,\|P_m\|_\infty\right\|_q} = \frac16 (d^2-1)^{-\frac1q}\ge \frac16 d^{-\frac2q} = \frac16d^{\frac2p-2}.
    \end{equation}
    This means $\varepsilon\le\textrm{R.H.S.}$ is sufficient to guarantee $\varepsilon\le\eta_{\comp}^\ob$. For the upper bound threshold, combining \Cref{thm:ob-estimation-upper} with our lower bound of $\Gamma_p^\ob=\Omega(d)$, we have
    \begin{equation}
        \overline{\eta}^{\ob}_p = \sqrt{\Gamma_p^\ob/d^3} = \Omega(d^{-1}).
    \end{equation}
    Therefore, there exists an absolute constant $c_0>0$ such that when $\varepsilon<c_0d^{-1}$ it is guaranteed that $\epsilon<\overline{\eta}^{\ob}_p$. This completes the proof.
\end{proof}

\begin{corollary} For shadow estimation of complete Pauli observables, 
    $\Omega(d)\leq \Gamma_2 \leq O(d^3\log d)$. 
\end{corollary}
\begin{proof}
    The lower bound part follows from \Cref{lem:simple_relation}. The upper bound part can be seen as follows:
    \begin{equation}
    \begin{aligned}
        \Gamma_2&=\sup_{\rho_0\in\cS^\circ}\inf_{M\in\cM}\Tr(I_{\rho_0,M}^{-1})
        \\&\le m \sup_{\rho_0\in\cS^\circ}\inf_{M\in\cM}\left(\min_{\|\theta\|_2=1}\theta^\top I_{\rho_0,M}\theta\right)^{-1}
        \\&=m~\Gamma_2^\ob
        \\&=O(d^3\log d).
    \end{aligned}
    \end{equation}
    The second line uses the variational expression for the minimal eigenvalue.
\end{proof}

We also consider Pauli observable estimation with few-copy measurements. We prove the following lower bound regarding \Cref{prob:p-estimation}(\ref{prob:oblivious-p-estimation}) and \Cref{prob:distinguish}(\ref{prob:oblivious-distinguish}). 

\begin{theorem}[Lower bound of Pauli observables for \Cref{prob:p-estimation}(\ref{prob:oblivious-p-estimation}) and \Cref{prob:distinguish}(\ref{prob:oblivious-distinguish}), $c$-copy measurements]\label{thm:lower-c-copy-pauli}
Let $d=2^n$ and $\{P_i\}_{i=1}^{d^2-1}$ be all $n$-qubit traceless Pauli operators. 

Using (possibly adaptive) measurement strategy with $c$-copy measurements, the sample complexity of $\rho$ parametrized by $\{P_i\}_{i=1}^{d^2-1}$ required to obtain a $\kappa$-kurtosis-bounded, unbiased estimator of $\theta$ in \Cref{prob:p-estimation} (or a $\kappa$-kurtosis-bounded, unbiased estimator of $\theta_\alpha$ for all $\alpha$ satisfying $\norm{\alpha}_q \leq 1$ in \Cref{prob:oblivious-p-estimation}) is
\begin{equation}
N = \Omega\left(\frac{(1-\sqrt{\kappa\delta})\,d}{c \epsilon^2} \right),
\label{eq:unbiased-lower-c-Pauli}
\end{equation}
for any $\epsilon > 0$, $p \in [1,\infty]$, and $\kappa\delta < 1$ where $\delta$ is the failure probability.

Using (possibly adaptive) measurement strategy with $c$-copy measurements, consider the many-versus-one distinguishing tasks in \Cref{prob:distinguish}(\ref{prob:oblivious-distinguish}) with states $\rho$ parametrized by $\{P_i\}_{i=1}^{d^2-1}$. The sample complexity required to solve this task is
\begin{align}
N =\Omega\left(\frac{d}{c \epsilon^2}\right), 
\label{eq:distinguish-lower-c-Pauli}
\end{align}
for any 
\begin{align}
\label{eq:distinguish-lower-c-Pauli-th}
\epsilon \leq  
\begin{cases}
 \min\left\{ \frac{1}{6}d^{\frac{2}{p}-2}, 
\frac{c_1}{c d^{\frac{2}{p}-\frac{1}{2}}\sqrt{\log d}  } 
 \right\}
, \quad&\textrm{if}~p\in[1,2),\\
  \min\left\{ \frac{1}{6}d^{\frac{2}{p}-2},
\frac{c_1}{c d^{\frac{5}{2}-\frac{4}{p}}\sqrt{\log d}  } 
  \right\}
, \quad&\textrm{if}~p\in[2,\infty], 
\end{cases} 
\end{align} 
where $c_1 > 0$ is an absolute constant. 
\end{theorem}

\begin{proof}
\eqref{eq:unbiased-lower-c-Pauli} and \eqref{eq:distinguish-lower-c-Pauli} can be obtained by combining \Cref{thm:distinguish-c-copy} and \Cref{thm:estimation-lower-c} with \Cref{thm:oblivious-pauli}. More precisely, for \eqref{eq:unbiased-lower-c-Pauli} with \Cref{prob:p-estimation}, we use $\Gamma_p\ge\Gamma^\ob_p$ (\Cref{lem:simple_relation}) when $p\in[2,\infty]$; when $p\in[1,2)$, since $\norm{\cdot}_2\le\norm{\cdot}_p$, the estimator also achieves a $2$-norm error $\epsilon$ with the same failure probability, and the bound for $p=2$ applies. In the following, we focus on deriving \eqref{eq:distinguish-lower-c-Pauli-th}.

Since the Pauli operators form a complete basis, we have $Q_i=P_i$ and
there are no nuisance parameters. Recall the definition of $a(\theta)$
in the proof of \Cref{thm:distinguish-c-copy}:
\begin{equation}
a(\theta)
:=
\frac{
\operatorname{tr}\!\left[
(\theta\cdot\vec P)(\rho_0^\star)^{-1}
(\theta\cdot\vec P)
\right]
}{d^2}.
\end{equation}
Since $\rho_0^\star\in\cS_{1/2}$,
$(\rho_0^\star)^{-1}\preceq 2d\,\id$, and hence
\begin{equation}
a(\theta)
\le
\frac{2}{d}
\operatorname{tr}\!\left[(\theta\cdot\vec P)^2\right].
\end{equation}
Using the Hilbert--Schmidt orthogonality of the Pauli operators,
$\operatorname{tr}(P_iP_j)=d\delta_{ij}$, we obtain
\begin{equation}
a(\theta)
\le
2\|\theta\|_2^2.
\end{equation}
For $\|\theta\|_p=1$, the standard $\ell_p$--$\ell_2$ relation gives
\begin{equation}
\max_{\|\theta\|_p=1}\|\theta\|_2^2
\le
\begin{cases}
1, & p\in[1,2),\\
(d^2-1)^{1-2/p}, & p\in[2,\infty].
\end{cases}
\end{equation}
Therefore,
\begin{equation}
a(\theta)
\le
\begin{cases}
2\|\theta\|_p^2, & p\in[1,2),\\
2(d^2-1)^{1-2/p}\|\theta\|_p^2, & p\in[2,\infty].
\end{cases}
\end{equation}
In the proof of \Cref{thm:distinguish-c-copy}, the higher-order contribution $\Delta_2$
is bounded by the higher powers of $a(\theta)$. Thus, if
\begin{equation}
a :=
\begin{cases}
2, & p\in[1,2),\\
2(d^2-1)^{1-2/p}, & p\in[2,\infty],
\end{cases}
\end{equation}
then for $\|\theta\|_p=3\epsilon$,
\begin{equation}
\Delta_2
\le
\sum_{k=2}^{c}\binom{c}{k}
\left(3\epsilon\sqrt{a}\right)^k.
\end{equation}
Using the same argument as in \Cref{thm:distinguish-c-copy}, in addition to $\epsilon\le\eta_{\mathrm{comp}}^{\mathrm{ob}}$, it is sufficient to require
\begin{equation}
\epsilon
\le
\frac{1}{c}
\min\left\{
\frac{1}{6\sqrt{a}},
\frac{1}{3a \sqrt{\Gamma_p^{\mathrm{ob}}}}
\right\},
\end{equation}
where the first term guarantees $3c\epsilon\sqrt{a}\le 1/2$, and it is implied by the second term up to an absolute constant since $a\ge2$ and $\Gamma_p^{\mathrm{ob}}=\Omega(d)$.

For Pauli observables,
\begin{equation}
\eta_{\mathrm{comp}}^{\mathrm{ob}}
=
\frac{1}{6(d^2-1)^{1/q}},
\end{equation}
and \Cref{thm:oblivious-pauli} gives
\begin{equation}
\Gamma_p^{\mathrm{ob}}
=
\begin{cases}
O(d^{4/p-1}\log d), & p\in[1,2),\\
O(d\log d), & p\in[2,\infty].
\end{cases}
\end{equation}
Consequently, for some absolute constant $c_1>0$,
\begin{equation}
\epsilon
\le
\begin{cases}
\displaystyle
\min\left\{
\frac{1}{6}d^{2/p-2},
\frac{c_1}{c\,d^{2/p-1/2}\sqrt{\log d}}
\right\},
& p\in[1,2),\\[3mm]
\displaystyle
\min\left\{
\frac{1}{6}d^{2/p-2},
\frac{c_1}{c\,d^{5/2-4/p}\sqrt{\log d}}
\right\},
& p\in[2,\infty],
\end{cases}
\end{equation}
which is the claimed threshold scaling.

\end{proof}

\section*{Acknowledgment}
\addcontentsline{toc}{section}{Acknowledgment}
The authors thank Sitan Chen, Yunchao Liu, and Yuxiang Yang for valuable discussions and feedback.
S.C. acknowledges support from the Institute for Quantum Information and Matter, an NSF Physics Frontiers Center (NSF Grant PHY-2317110).
S.Z. acknowledges funding provided by Perimeter Institute for Theoretical Physics, a research institute supported in part by the Government of Canada through the Department of Innovation, Science and Economic Development Canada and by the Province of Ontario through the Ministry of Colleges and Universities.
S.C. and S.Z. also acknowledge support from the Kavli Institute for Theoretical Physics (NSF Grant PHY-2309135), where part of this work was completed.
W.G. acknowledges support by the Von Neumann Award from Harvard Computer Science and NSF Grant CCF-2430375. 

\vspace{0.1in}
\noindent\emph{Note}:~After our paper appeared, a concurrent and related work~\cite{kwon2026universal} was posted on arXiv, which also connects quantum metrology to quantum learning theory. It obtains sample complexity bounds for maximum-likelihood estimators using different approaches from ours.

\vspace{0.1in}
\noindent\emph{AI Use Disclosure for v3}: In version 3 of this paper on arXiv, ChatGPT 6 Astra and Claude Fable 5.2 were used to help identify and correct two errors in the previous version (v2). (1)~In the sample-complexity expressions, $\inf_M$ and $\sup_{\rho_0}$ appeared in the wrong order; this has now been corrected. The error originated from an incorrect expression for $\delta_{\cM_c}$ in Section~6, where the order of $\sup_M$ and $\inf_{\rho_0}$ was incorrect. (2)~In Section~8, we mistakenly applied the Cram\'{e}r--Rao (CR) bound to a median estimator, which needs not be unbiased. This issue has been corrected by replacing the boundedness condition in v2 with a kurtosis-boundedness condition in v3. ChatGPT 6 Astra and Claude Fable 5.2 also assisted in identifying and correcting several other, less consequential, technical errors, including an issue in the proof of the existence of a threshold in Sections~7.3 and~7.4. All AI's suggestions of revision were independently verified and implemented by the authors. The main conclusions and proof techniques remain unchanged from v2. The authors bear sole responsibility for the correctness of the proofs and results. 

%%%%%%%%%%%%%%%%%%%%%%%%%%%%%%%%%%%%%%%%%%%%%%%%%%%%%%%%%%%%%%%%%%%%%%%%%%%%%%%%%%%%%%%%%%%%%%%%%%%%%%%%%%%%%%%%%%%%%%%%%%%%%%%%%%%%%%%%%%%%%%%%%%%%%%%%%%%%%%%%%%%%%%%%%%%%%%%%%%%%%%%%%%

\bibliographystyle{unsrt}
\bibliography{ref}

\begin{thebibliography}{10}

\bibitem{degen2017quantum}
Christian~L Degen, Friedemann Reinhard, and Paola Cappellaro.
\newblock Quantum sensing.
\newblock {\em Reviews of modern physics}, 89(3):035002, 2017.

\bibitem{pirandola2018advances}
Stefano Pirandola, B~Roy Bardhan, Tobias Gehring, Christian Weedbrook, and Seth Lloyd.
\newblock Advances in photonic quantum sensing.
\newblock {\em Nature Photonics}, 12(12):724--733, 2018.

\bibitem{Dalzell2025QuantumAlgorithms}
Alexander~M. Dalzell, Sam McArdle, Mario Berta, Przemyslaw Bienias, Chi-Fang Chen, Andr{\'a}s Gily{\'e}n, Connor~T. Hann, Michael~J. Kastoryano, Emil~T. Khabiboulline, Aleksander Kubica, Grant Salton, Samson Wang, and Fernando G. S.~L. Brand{\~a}o.
\newblock {\em Quantum Algorithms: A Survey of Applications and End-to-end Complexities}.
\newblock Cambridge University Press, 2025.

\bibitem{harper2020efficient}
Robin Harper, Steven~T Flammia, and Joel~J Wallman.
\newblock Efficient learning of quantum noise.
\newblock {\em Nature Physics}, 16(12):1184--1188, 2020.

\bibitem{hashim2024practical}
Akel Hashim, Long~B Nguyen, Noah Goss, Brian Marinelli, Ravi~K Naik, Trevor Chistolini, Jordan Hines, JP~Marceaux, Yosep Kim, Pranav Gokhale, et~al.
\newblock A practical introduction to benchmarking and characterization of quantum computers.
\newblock {\em arXiv:2408.12064}, 2024.

\bibitem{giovannetti2011advances}
Vittorio Giovannetti, Seth Lloyd, and Lorenzo Maccone.
\newblock Advances in quantum metrology.
\newblock {\em Nature photonics}, 5(4):222--229, 2011.

\bibitem{pezze2018quantum}
Luca Pezze, Augusto Smerzi, Markus~K Oberthaler, Roman Schmied, and Philipp Treutlein.
\newblock Quantum metrology with nonclassical states of atomic ensembles.
\newblock {\em Reviews of Modern Physics}, 90(3):035005, 2018.

\bibitem{van2000asymptotic}
Aad~W Van~der Vaart.
\newblock {\em Asymptotic statistics}, volume~3.
\newblock Cambridge university press, 2000.

\bibitem{arunachalam2017guest}
Srinivasan Arunachalam and Ronald De~Wolf.
\newblock Guest column: A survey of quantum learning theory.
\newblock {\em ACM Sigact News}, 48(2):41--67, 2017.

\bibitem{aaronson2018shadow}
Scott Aaronson.
\newblock Shadow tomography of quantum states.
\newblock In {\em Proceedings of the 50th Annual ACM SIGACT Symposium on Theory of Computing}, pages 325--338, 2018.

\bibitem{kay1993fundamentals}
Steven~M Kay.
\newblock {\em Fundamentals of statistical signal processing: Volume I Estimation theory}.
\newblock Prentice-Hall, Inc., 1993.

\bibitem{lehmann2006theory}
Erich~L Lehmann and George Casella.
\newblock {\em Theory of point estimation}.
\newblock Springer Science \& Business Media, 2006.

\bibitem{chen2024optimal}
Sitan Chen, Weiyuan Gong, and Qi~Ye.
\newblock Optimal tradeoffs for estimating {Pauli} observables.
\newblock In {\em 2024 IEEE 65th Annual Symposium on Foundations of Computer Science (FOCS)}, pages 1086--1105, 2024.

\bibitem{huang2022quantum}
Hsin-Yuan Huang, Michael Broughton, Jordan Cotler, Sitan Chen, Jerry Li, Masoud Mohseni, Hartmut Neven, Ryan Babbush, Richard Kueng, John Preskill, et~al.
\newblock Quantum advantage in learning from experiments.
\newblock {\em Science}, 376(6598):1182--1186, 2022.

\bibitem{huang2021information}
Hsin-Yuan Huang, Richard Kueng, and John Preskill.
\newblock Information-theoretic bounds on quantum advantage in machine learning.
\newblock {\em Physical Review Letters}, 126(19):190505, 2021.

\bibitem{zernike1938concept}
Frederik Zernike.
\newblock The concept of degree of coherence and its application to optical problems.
\newblock {\em Physica}, 5(8):785--795, 1938.

\bibitem{khabiboulline2019quantum}
Emil~T Khabiboulline, Johannes Borregaard, Kristiaan De~Greve, and Mikhail~D Lukin.
\newblock Quantum-assisted telescope arrays.
\newblock {\em Physical review A}, 100(2):022316, 2019.

\bibitem{aharonov2022quantum}
Dorit Aharonov, Jordan Cotler, and Xiao-Liang Qi.
\newblock Quantum algorithmic measurement.
\newblock {\em Nature Communications}, 13(887):1--9, 2022.

\bibitem{bubeck2020entanglement}
Sebastien Bubeck, Sitan Chen, and Jerry Li.
\newblock Entanglement is necessary for optimal quantum property testing.
\newblock In {\em 2020 IEEE 61st Annual Symposium on Foundations of Computer Science (FOCS)}, pages 692--703. IEEE, 2020.

\bibitem{chen2022exponential}
Sitan Chen, Jordan Cotler, Hsin-Yuan Huang, and Jerry Li.
\newblock Exponential separations between learning with and without quantum memory.
\newblock In {\em 2021 IEEE 62nd Annual Symposium on Foundations of Computer Science (FOCS)}, pages 574--585. IEEE, 2022.

\bibitem{chen2023complexity}
Sitan Chen, Jordan Cotler, Hsin-Yuan Huang, and Jerry Li.
\newblock The complexity of nisq.
\newblock {\em Nat. Commun.}, 14(1):6001, 2023.

\bibitem{chen2025efficient}
Sitan Chen and Weiyuan Gong.
\newblock Efficient pauli channel estimation with logarithmic quantum memory.
\newblock {\em PRX Quantum}, 6(2):020323, 2025.

\bibitem{hu2025ansatz}
Hong-Ye Hu, Muzhou Ma, Weiyuan Gong, Qi~Ye, Yu~Tong, Steven~T. Flammia, and Susanne~F. Yelin.
\newblock Ansatz-free {Hamiltonian} learning with heisenberg-limited scaling.
\newblock {\em PRX Quantum}, 6:040315, Oct 2025.

\bibitem{barndorff2000fisher}
O~E Barndorff-Nielsen and R~D Gill.
\newblock Fisher information in quantum statistics.
\newblock {\em J. Phys. A: Math. Gen.}, 33(24):4481--4490, jun 2000.

\bibitem{hayashi2011comparison}
Masahito Hayashi.
\newblock Comparison between the {C}ramer-{R}ao and the mini-max approaches in quantum channel estimation.
\newblock {\em Commun. Math. Phys.}, 304(3):689--709, 2011.

\bibitem{yang2019attaining}
Yuxiang Yang, Giulio Chiribella, and Masahito Hayashi.
\newblock Attaining the ultimate precision limit in quantum state estimation.
\newblock {\em Communications in Mathematical Physics}, 368(1):223--293, 2019.

\bibitem{chen2022complexity}
Sitan Chen, Jordan Cotler, Hsin-Yuan Huang, and Jerry Li.
\newblock The complexity of {NISQ}.
\newblock {\em Nature Communications}, 14(1):6001, 2023.

\bibitem{banaszek2013focus}
Konrad Banaszek, Marcus Cramer, and David Gross.
\newblock Focus on quantum tomography.
\newblock {\em New Journal of Physics}, 15(12):125020, 2013.

\bibitem{blume2010optimal}
Robin Blume-Kohout.
\newblock Optimal, reliable estimation of quantum states.
\newblock {\em New Journal of Physics}, 12(4):043034, 2010.

\bibitem{gross2010quantum}
David Gross, Yi-Kai Liu, Steven~T Flammia, Stephen Becker, and Jens Eisert.
\newblock Quantum state tomography via compressed sensing.
\newblock {\em Physical Review Letters}, 105(15):150401, 2010.

\bibitem{hradil1997quantum}
Zdenek Hradil.
\newblock Quantum state estimation.
\newblock {\em Physical Review A}, 55(3):R1561, 1997.

\bibitem{haah2016sample}
Jeongwan Haah, Aram~W Harrow, Zhengfeng Ji, Xiaodi Wu, and Nengkun Yu.
\newblock Sample-optimal tomography of quantum states.
\newblock In {\em Proceedings of the Forty-eighth Annual ACM Symposium on Theory of Computing}, pages 913--925, 2016.

\bibitem{odonnell2016efficient}
Ryan O'Donnell and John Wright.
\newblock Efficient quantum tomography.
\newblock In {\em Proceedings of the Forty-eighth Annual ACM Symposium on Theory of Computing}, pages 899--912, 2016.

\bibitem{aaronson2018online}
Scott Aaronson, Xinyi Chen, Elad Hazan, Satyen Kale, and Ashwin Nayak.
\newblock Online learning of quantum states.
\newblock In {\em Advances in neural information processing systems}, volume~31, 2018.

\bibitem{aaronson2019gentle}
Scott Aaronson and Guy~N Rothblum.
\newblock Gentle measurement of quantum states and differential privacy.
\newblock In {\em Proceedings of the 51st Annual ACM SIGACT Symposium on Theory of Computing}, pages 322--333, 2019.

\bibitem{buadescu2021improved}
Costin B\u{a}descu and Ryan O'Donnell.
\newblock Improved quantum data analysis.
\newblock In {\em Proceedings of the 53rd Annual ACM SIGACT Symposium on Theory of Computing}, pages 1398--1411, 2021.

\bibitem{brandao2019quantum}
Fernando~GSL Brand{\~a}o, Amir Kalev, Tongyang Li, Cedric Yen-Yu Lin, Krysta~M Svore, and Xiaodi Wu.
\newblock Quantum {SDP} solvers: Large speed-ups, optimality, and applications to quantum learning.
\newblock In {\em 46th International Colloquium on Automata, Languages, and Programming (ICALP 2019)}. Schloss Dagstuhl-Leibniz-Zentrum fuer Informatik, 2019.

\bibitem{gong2023learning}
Weiyuan Gong and Scott Aaronson.
\newblock Learning distributions over quantum measurement outcomes.
\newblock In {\em International Conference on Machine Learning}, pages 11598--11613. PMLR, 2023.

\bibitem{watts2024quantum}
Adam~Bene Watts and John Bostanci.
\newblock Quantum event learning and gentle random measurements.
\newblock In {\em 15th Innovations in Theoretical Computer Science Conference (ITCS 2024)}. Schloss-Dagstuhl-Leibniz Zentrum f{\"u}r Informatik, 2024.

\bibitem{chen2024optimalshadow}
Sitan Chen, Jerry Li, and Allen Liu.
\newblock Optimal high-precision shadow estimation.
\newblock {\em arXiv:2407.13874}, 2024.

\bibitem{pelecanos2025debiased}
Angelos Pelecanos, Jack Spilecki, and John Wright.
\newblock The debiased keyl's algorithm: a new unbiased estimator for full state tomography.
\newblock {\em arXiv:2510.07788}, 2025.

\bibitem{huang2020predicting}
Hsin-Yuan Huang, Richard Kueng, and John Preskill.
\newblock Predicting many properties of a quantum system from very few measurements.
\newblock {\em Nature Physics}, 16(10):1050--1057, October 2020.

\bibitem{elben2023randomized}
Andreas Elben, Steven~T Flammia, Hsin-Yuan Huang, Richard Kueng, John Preskill, Beno{\^\i}t Vermersch, and Peter Zoller.
\newblock The randomized measurement toolbox.
\newblock {\em Nature Reviews Physics}, 5(1):9--24, 2023.

\bibitem{knill2008randomized}
Emanuel Knill, Dietrich Leibfried, Rolf Reichle, Joe Britton, R~Brad Blakestad, John~D Jost, Chris Langer, Roee Ozeri, Signe Seidelin, and David~J Wineland.
\newblock Randomized benchmarking of quantum gates.
\newblock {\em Phys. Rev. A}, 77(1):012307, 2008.

\bibitem{dankert2009exact}
Christoph Dankert, Richard Cleve, Joseph Emerson, and Etera Livine.
\newblock Exact and approximate unitary 2-designs and their application to fidelity estimation.
\newblock {\em Phys. Rev. A}, 80(1):012304, 2009.

\bibitem{emerson2005scalable}
Joseph Emerson, Robert Alicki, and Karol {\.Z}yczkowski.
\newblock Scalable noise estimation with random unitary operators.
\newblock {\em Journal of Optics B: Quantum and Semiclassical Optics}, 7(10):S347, 2005.

\bibitem{brydges2019probing}
Tiff Brydges, Andreas Elben, Petar Jurcevic, Beno{\^\i}t Vermersch, Christine Maier, Ben~P Lanyon, Peter Zoller, Rainer Blatt, and Christian~F Roos.
\newblock Probing {R}{\'e}nyi entanglement entropy via randomized measurements.
\newblock {\em Science}, 364(6437):260--263, 2019.

\bibitem{king2024triply}
Robbie King, David Gosset, Robin Kothari, and Ryan Babbush.
\newblock Triply efficient shadow tomography.
\newblock {\em arXiv:2404.19211}, 2024.

\bibitem{chen2022quantum}
Senrui Chen, Sisi Zhou, Alireza Seif, and Liang Jiang.
\newblock Quantum advantages for {P}auli channel estimation.
\newblock {\em Physical Review A}, 105(3):032435, 2022.

\bibitem{chen2024tight}
Senrui Chen, Changhun Oh, Sisi Zhou, Hsin-Yuan Huang, and Liang Jiang.
\newblock Tight bounds on {Pauli} channel learning without entanglement.
\newblock {\em Physical Review Letters}, 132(18):180805, 2024.

\bibitem{seif2024entanglement}
Alireza Seif, Senrui Chen, Swarnadeep Majumder, Haoran Liao, Derek~S Wang, Moein Malekakhlagh, Ali Javadi-Abhari, Liang Jiang, and Zlatko~K Minev.
\newblock Entanglement-enhanced learning of quantum processes at scale.
\newblock {\em arXiv:2408.03376}, 2024.

\bibitem{chen2021hierarchy}
Sitan Chen, Jordan Cotler, Hsin-Yuan Huang, and Jerry Li.
\newblock A hierarchy for replica quantum advantage.
\newblock {\em arXiv:2111.05874}, 2021.

\bibitem{ye2025exponential}
Qi~Ye, Zhenhuan Liu, and Dong-Ling Deng.
\newblock Exponential advantage from one more replica in estimating nonlinear properties of quantum states.
\newblock {\em arXiv:2509.24000}, 2025.

\bibitem{noller2025infinite}
Jan N{\"o}ller, Viet~T Tran, Mariami Gachechiladze, and Richard Kueng.
\newblock An infinite hierarchy of multi-copy quantum learning tasks.
\newblock {\em arXiv:2510.08070}, 2025.

\bibitem{chen2022tight2}
Sitan Chen, Jerry Li, Brice Huang, and Allen Liu.
\newblock Tight bounds for quantum state certification with incoherent measurements.
\newblock In {\em 2022 IEEE 63rd Annual Symposium on Foundations of Computer Science (FOCS)}, pages 1205--1213. IEEE, 2022.

\bibitem{odonnell2025instance}
Ryan O'Donnell and Chirag Wadhwa.
\newblock Instance-optimal quantum state certification with entangled measurements.
\newblock {\em arXiv:2507.06010}, 2025.

\bibitem{chen2022toward}
Sitan Chen, Jerry Li, and Ryan O'Donnell.
\newblock Toward instance-optimal state certification with incoherent measurements.
\newblock In {\em Conference on Learning Theory}, pages 2541--2596. PMLR, 2022.

\bibitem{fawzi2023quantum}
Omar Fawzi, Nicolas Flammarion, Aur{\'e}lien Garivier, and Aadil Oufkir.
\newblock Quantum channel certification with incoherent strategies.
\newblock In {\em COLT 23-36th Annual Conference on Learning Theory}, pages 1--58, 2023.

\bibitem{chen2024optimalstate}
Sitan Chen, Jerry Li, and Allen Liu.
\newblock An optimal tradeoff between entanglement and copy complexity for state tomography.
\newblock In {\em Proceedings of the 56th Annual ACM Symposium on Theory of Computing}, pages 1331--1342, 2024.

\bibitem{anshu2023survey}
Anurag Anshu and Srinivasan Arunachalam.
\newblock A survey on the complexity of learning quantum states.
\newblock {\em Nature Reviews Physics}, pages 1--11, 2023.

\bibitem{helstrom1967minimum}
Carl~W Helstrom.
\newblock Minimum mean-squared error of estimates in quantum statistics.
\newblock {\em Physics Letters A}, 25(2):101--102, 1967.

\bibitem{helstrom1968minimum}
Carl~Wilhelm Helstrom.
\newblock The minimum variance of estimates in quantum signal detection.
\newblock {\em IEEE Trans. Inf. Theory}, 14(2):234--242, 1968.

\bibitem{helstrom1969quantum}
Carl~W Helstrom.
\newblock Quantum detection and estimation theory.
\newblock {\em Journal of Statistical Physics}, 1:231--252, 1969.

\bibitem{holevo2011probabilistic}
Alexander~S Holevo.
\newblock {\em Probabilistic and statistical aspects of quantum theory}, volume~1.
\newblock Springer Science \& Business Media, 2011.

\bibitem{braunstein1994statistical}
Samuel~L. Braunstein and Carlton~M. Caves.
\newblock Statistical distance and the geometry of quantum states.
\newblock {\em Phys. Rev. Lett.}, 72(22):3439--3443, May 1994.

\bibitem{paris2009quantum}
Matteo~GA Paris.
\newblock Quantum estimation for quantum technology.
\newblock {\em Int. J. Quantum Inf.}, 7(supp01):125--137, 2009.

\bibitem{albarelli2019evaluating}
Francesco Albarelli, Jamie~F Friel, and Animesh Datta.
\newblock Evaluating the holevo {C}ram{\'e}r-{R}ao bound for multiparameter quantum metrology.
\newblock {\em Physical Review Letters}, 123(20):200503, 2019.

\bibitem{gorecki2020optimal}
Wojciech G{\'o}recki, Sisi Zhou, Liang Jiang, and Rafa{\l} Demkowicz-Dobrza{\'n}ski.
\newblock Optimal probes and error-correction schemes in multi-parameter quantum metrology.
\newblock {\em Quantum}, 4:288, 2020.

\bibitem{tsang2020quantum}
Mankei Tsang, Francesco Albarelli, and Animesh Datta.
\newblock Quantum semiparametric estimation.
\newblock {\em Physical Review X}, 10(3):031023, 2020.

\bibitem{sidhu2021tight}
Jasminder~S Sidhu, Yingkai Ouyang, Earl~T Campbell, and Pieter Kok.
\newblock Tight bounds on the simultaneous estimation of incompatible parameters.
\newblock {\em Physical Review X}, 11(1):011028, 2021.

\bibitem{hayashi2023tight}
Masahito Hayashi and Yingkai Ouyang.
\newblock Tight {C}ram{\'e}r-{R}ao type bounds for multiparameter quantum metrology through conic programming.
\newblock {\em Quantum}, 7:1094, 2023.

\bibitem{gardner2024achieving}
James~W Gardner, Tuvia Gefen, Simon~A Haine, Joseph~J Hope, and Yanbei Chen.
\newblock Achieving the fundamental quantum limit of linear waveform estimation.
\newblock {\em Physical Review Letters}, 132(13):130801, 2024.

\bibitem{kahn2009local}
Jonas Kahn and M{\u{a}}d{\u{a}}lin Gu{\c{t}}{\u{a}}.
\newblock Local asymptotic normality for finite dimensional quantum systems.
\newblock {\em Communications in Mathematical Physics}, 289(2):597--652, 2009.

\bibitem{yamagata2013quantum}
Koichi Yamagata, Akio Fujiwara, and Richard~D. Gill.
\newblock {Quantum local asymptotic normality based on a new quantum likelihood ratio}.
\newblock {\em The Annals of Statistics}, 41(4):2197 -- 2217, 2013.

\bibitem{matsumoto2002new}
Keiji Matsumoto.
\newblock A new approach to the {C}ram{\'e}r-{R}ao-type bound of the pure-state model.
\newblock {\em Journal of Physics A: Mathematical and General}, 35(13):3111, 2002.

\bibitem{albarelli2019upper}
Francesco Albarelli, Mankei Tsang, and Animesh Datta.
\newblock Upper bounds on the holevo cram\'{e}r-rao bound for multiparameter quantum parametric and semiparametric estimation.
\newblock {\em arXiv:1911.11036}, 2019.

\bibitem{carollo2019quantumness}
Angelo Carollo, Bernardo Spagnolo, Alexander~A Dubkov, and Davide Valenti.
\newblock On quantumness in multi-parameter quantum estimation.
\newblock {\em Journal of Statistical Mechanics: Theory and Experiment}, 2019(9):094010, 2019.

\bibitem{demkowicz2020multi}
Rafa{\l} Demkowicz-Dobrza{\'n}ski, Wojciech G{\'o}recki, and M{\u{a}}d{\u{a}}lin Gu{\c{t}}{\u{a}}.
\newblock Multi-parameter estimation beyond quantum fisher information.
\newblock {\em Journal of Physics A: Mathematical and Theoretical}, 53(36):363001, 2020.

\bibitem{li2016fisher}
Nan Li, Christopher Ferrie, Jonathan~A Gross, Amir Kalev, and Carlton~M Caves.
\newblock Fisher-symmetric informationally complete measurements for pure states.
\newblock {\em Physical Review Letters}, 116(18):180402, 2016.

\bibitem{zhu2018universally}
Huangjun Zhu and Masahito Hayashi.
\newblock Universally {F}isher-symmetric informationally complete measurements.
\newblock {\em Physical Review Letters}, 120(3):030404, 2018.

\bibitem{vargas2024near}
C~Vargas, L~Pereira, and A~Delgado.
\newblock Near-optimal pure state estimation with adaptive fisher-symmetric measurements.
\newblock {\em arXiv:2412.04555}, 2024.

\bibitem{zhou2026randomized}
Sisi Zhou and Senrui Chen.
\newblock Randomized measurements for multiparameter quantum metrology.
\newblock {\em PRX Quantum}, 7(1):010314, 2026.

\bibitem{harrow2013church}
Aram~W Harrow.
\newblock The church of the symmetric subspace.
\newblock {\em arXiv:1308.6595}, 2013.

\bibitem{yu1997assouad}
Bin Yu.
\newblock {Assouad, Fano, and Le Cam}.
\newblock {\em Festschrift for Lucien Le Cam: research papers in probability and statistics}, pages 423--435, 1997.

\bibitem{sion1958general}
Maurice Sion.
\newblock On general minimax theorems.
\newblock {\em Pacific Journal of Mathematics}, 8(1):171--176, 1958.

\bibitem{vershynin2018high}
Roman Vershynin.
\newblock {\em High-{{Dimensional Probability}}: {{An Introduction}} with {{Applications}} in {{Data Science}}}.
\newblock {Cambridge University Press}, 1 edition, 2018.

\bibitem{hoeffding1963probability}
Wassily Hoeffding.
\newblock Probability {{Inequalities}} for {{Sums}} of {{Bounded Random Variables}}.
\newblock {\em Journal of the American Statistical Association}, 58(301):13--30, 1963.

\bibitem{suzuki2020quantum}
Jun Suzuki, Yuxiang Yang, and Masahito Hayashi.
\newblock Quantum state estimation with nuisance parameters.
\newblock {\em Journal of Physics A: Mathematical and Theoretical}, 53(45):453001, 2020.

\bibitem{kueng2017low}
Richard Kueng, Holger Rauhut, and Ulrich Terstiege.
\newblock Low rank matrix recovery from rank one measurements.
\newblock {\em Applied and Computational Harmonic Analysis}, 42(1):88--116, 2017.

\bibitem{gutaFastStateTomography2020}
M~Guţă, J~Kahn, R~Kueng, and J~A Tropp.
\newblock Fast state tomography with optimal error bounds.
\newblock {\em Journal of Physics A: Mathematical and Theoretical}, 53(20):204001, 2020.

\bibitem{chenCotler2025Lecture6TomographyII}
Sitan Chen and Jordan Cotler.
\newblock Lecture 6: Tomography ii: Single-copy measurements (operator norm).
\newblock Lecture notes for Harvard Physics 272 / CS 2233: Quantum Learning Theory (Fall 2025), September 2025.
\newblock Accessed: 2026-01-16.

\bibitem{robbins1955remark}
Herbert Robbins.
\newblock A remark on {S}tirling's formula.
\newblock {\em The American mathematical monthly}, 62(1):26--29, 1955.

\bibitem{kwon2026universal}
Hyukgun Kwon, Seok~Hyung Lie, and Liang Jiang.
\newblock Universal sample complexity bounds in quantum learning theory via fisher information matrix.
\newblock {\em arXiv:2602.21510v1}, 2026.

\end{thebibliography}
\addcontentsline{toc}{section}{References}

\clearpage

\end{document}